\documentclass[11pt,a4paper,twoside]{article}
\usepackage[utf8]{inputenc}
\usepackage[T1]{fontenc}
\usepackage[american]{babel}
\usepackage{fourier}
\usepackage{amsmath,amssymb}
\usepackage[top=3cm, bottom=3cm, left=3cm, right=3cm]{geometry}
\usepackage{graphicx,color,xcolor}
\usepackage[colorlinks=true,hyperindex=true]{hyperref}
\colorlet{darkblue}{blue!50!black}
\hypersetup{
	colorlinks,%
	citecolor=darkblue,%
	filecolor=red,%
	linkcolor=darkblue,%
	urlcolor=darkblue,%
}
\usepackage{amsthm}
\usepackage[shortlabels]{enumitem}
\usepackage{xspace}
\setlist{noitemsep,topsep=0pt,parsep=5pt,partopsep=0pt}
\usepackage{etoolbox}
\usepackage{needspace}
\usepackage{fancyhdr}
\usepackage{mathtools}
\newcommand{\Id}{\mathrm{Id}}
\newcommand{\e}{\mathrm{e}}
\renewcommand{\d}{\mathrm{d}}
\renewcommand{\i}{\mathrm{i}}
\newcommand{\Ent}{\mathrm{Ent}}
\newcommand{\ep}{\mathrm{ep}}
\newcommand{\bra}{\langle} 
\newcommand{\ket}{\rangle}

\newcommand{\wh}{\widehat}

\newcommand{\un}{\mathbb{1}}
\newcommand{\LL}{\mathbb{L}}
\newcommand{\II}{\mathbb{I}}
\newcommand{\KK}{\mathbb{K}}
\newcommand{\EE}{\mathbb{E}}
\newcommand{\RR}{\mathbb{R}}
\newcommand{\NN}{\mathbb{N}}
\newcommand{\QQ}{\mathbb{Q}}
\newcommand{\CC}{\mathbb{C}}
\newcommand{\ZZ}{\mathbb{Z}}
\newcommand{\PP}{\mathbb{P}}

\newcommand{\bun}{{\boldsymbol{1}}}
\newcommand{\bzero}{{\boldsymbol{0}}}
\newcommand{\z}{{\boldsymbol{\zeta}}}
\newcommand{\bb}{\boldsymbol{\beta}}

\newcommand{\ba}{\boldsymbol{\alpha}}

\newcommand{\bx}{{\boldsymbol{x}}}
\newcommand{\bs}{{\boldsymbol{s}}}
\newcommand{\vs}{\boldsymbol{\varsigma}}

\newcommand{\bomega}{\boldsymbol{\omega}}
\newcommand{\bO}{{\boldsymbol{\Omega}}}
\newcommand{\wlim}{\mathop{\mathrm{w}^\ast\!\mathrm{-lim}}\limits}

\newcommand{\fA}{\mathfrak{A}}
\newcommand{\fI}{\mathfrak{i}}
\newcommand{\fJ}{\mathfrak{J}}
\newcommand{\fX}{\mathfrak{X}}
\newcommand{\bX}{{\boldsymbol{\fX}}}

\DeclareMathOperator{\tr}{tr}
\DeclareMathOperator{\Tr}{Tr}
\DeclareMathOperator{\Ran}{Ran}
\DeclareMathOperator{\Ker}{Ker}

\DeclareMathOperator{\spec}{sp}
\theoremstyle{plain}
\newtheorem{theo}{Theorem}[section]
\newtheorem{prop}[theo]{Proposition}
\newtheorem{coro}[theo]{Corollary}

\newtheorem{lemm}[theo]{Lemma}

\theoremstyle{definition}
\newtheorem{defi}[theo]{Definition}
\newtheorem*{hypo}{Assumption}
\theoremstyle{remark}

\AtBeginEnvironment{theo}{\Needspace{5\baselineskip}}
\AtBeginEnvironment{prop}{\Needspace{5\baselineskip}}
\AtBeginEnvironment{defi}{\Needspace{5\baselineskip}}

\numberwithin{equation}{section}

\usepackage{xargs}
\usepackage[prependcaption]{todonotes}
\newcommandx{\jf}[2][1=]{\todo[inline, author={J.-F.}, linecolor=yellow,backgroundcolor=yellow!25,bordercolor=yellow,#1]{#2}}
\newcommandx{\jfnote}[2][1=]{\todo[author={J.-F.}, linecolor=yellow,backgroundcolor=yellow!25,bordercolor=yellow,#1]{#2}}

\newcommandx{\alain}[2][1=]{\todo[inline, author={Alain}, linecolor=cyan,backgroundcolor=cyan!15,bordercolor=cyan,#1]{#2}}
\newcommandx{\alainnote}[2][1=]{\todo[author={Alain}, linecolor=cyan,backgroundcolor=cyan!15,bordercolor=cyan,#1]{#2}}
\newcommandx{\addalain}[2][1=]{\todo[inline, linecolor=cyan,backgroundcolor=cyan!15,bordercolor=cyan,#1]{#2}}

\newcommandx{\ca}[2][1=]{\todo[inline,author={ca},
linecolor=green,backgroundcolor=green!15,bordercolor=green,#1]{#2}}
\newcommandx{\canot}[2][1=]{\todo[author={ca},
	linecolor=green,backgroundcolor=green!15,bordercolor=green,#1]{#2}}

\newcommand{\be}{\begin{equation}}
\newcommand{\ee}{\end{equation}}
\newcommand{\bea}{\begin{eqnarray}}
\newcommand{\eea}{\end{eqnarray}}
\newcommand{\ds}{\displaystyle}

\newcounter{resultcounter}[section]

\newcommand{\cA}{\mathcal{A}}
\newcommand{\cB}{\mathcal{B}}
\newcommand{\cC}{\mathcal{C}}

\newcommand{\cE}{\mathcal{E}}

\newcommand{\cH}{\mathcal{H}}

\newcommand{\cK}{\mathcal{K}}
\newcommand{\cL}{\mathcal{L}}
\newcommand{\cM}{\mathcal{M}}

\newcommand{\cO}{\mathcal{O}}
\newcommand{\cP}{\mathcal{P}}
\newcommand{\cQ}{\mathcal{Q}}
\newcommand{\cR}{\mathcal{R}}
\newcommand{\cS}{\mathcal{S}}

\newcommand{\cU}{\mathcal{U}}
\newcommand{\cV}{\mathcal{V}}

\newcommand{\cX}{\mathcal{X}}

\title{\textbf{Markovian Repeated Interaction Quantum Systems}}

\author{Jean-Fran\c{c}ois Bougron$^{1,2}$, Alain Joye$^{2}$, Claude-Alain Pillet$^3$
\\ \\ 
$^1$ CY Cergy Paris Universit\'e, CNRS, AGM,  95000 Cergy-Pontoise, France
\\ \\
$^2$ Univ. Grenoble Alpes, CNRS, Institut Fourier, 38000 Grenoble, France
\\ \\
$^3$ Aix Marseille Univ., Université de Toulon, CNRS, CPT, Marseille, France
}

\begin{document}
\def\today{}
\maketitle

\begin{small}
\noindent{\bf Abstract.} We study a class of dynamical semigroups $(\LL^n)_{n\in\NN}$ that emerge, by a Feynman--Kac type formalism, from a random quantum dynamical system 
$(\cL_{\omega_n}\circ\cdots\circ\cL_{\omega_1}(\rho_{\omega_0}))_{n\in\NN}$ 
driven by a Markov chain $(\omega_n)_{n\in\NN}$. We show that the almost sure large time behavior of the system can be extracted from the large $n$ asymptotics of the semigroup, which is in turn directly related to the spectral properties of the generator $\LL$. As a physical application, we consider the case where the $\cL_\omega$'s are the reduced dynamical maps describing the repeated interactions of a system $\cS$ with thermal probes $\cE_\omega$. We study the full statistics of the entropy in this system and derive a fluctuation theorem for the heat exchanges and the associated linear response formulas.
\end{small}

\tableofcontents

\section{Introduction}
\label{sec:intro}

\subsection*{Repeated Interaction Systems}

From a physical perspective, a Repeated Interaction System, RIS for short, is a quantum system $\cS$, together with a sequence\footnote{We denote by $\NN^\ast$ the positive integers and by $\NN=\NN^\ast\cup\{0\}$ the non-negative ones.} $(\cE_n)_{n\in\NN^\ast}$ of quantum probes. $\cS$, the system of interest, is described by a Hilbert space $\cH_{\cS}$. Each probe $\cE_n$ is a quantum system characterized by a Hilbert space $\cH_{\cE_n}$ and a (normal) state or density matrix $\rho_{\cE_n}$. All Hilbert spaces in this work are supposed to be finite dimensional.

The dynamics of the compound system is defined as follows: the system $\cS$, initially in a state $\rho_0$, interacts for a certain time with the first probe $\cE_1$ initially in its state $\rho_{\cE_1}$. Tracing out the degrees of freedom of the first probe $\cE_1$, one gets a state $\rho_1$ for $\cS$ which is then put in contact with the next probe $\cE_2$ in its initial state $\rho_{\cE_2}$. Repeating the procedure defines a sequence $(\rho_n)_{n\in\NN}$ where $\rho_n$ is the state of $\cS$ after $n$ interactions.

A well known example of Repeated Interaction System in physics is the one-atom maser experiment: $\cS$ is then the quantized electromagnetic field of a cavity, and the probes are atoms coming from an oven and passing through this cavity, see, e.g., \cite{FJM, MWM, RBH, WBKM}.

Typically, the interaction process is described as follows:  the free Hamiltonian of the small system, resp.\;of the probe number $n$, is a self-adjoint operator denoted by $H_\cS$, resp.\;$H_{\cE_n}$, and the interaction between $\cS$ and $\cE_n$ is induced by a coupling $V_n$ which is a self-adjoint operator on $\cH_\cS\otimes \cH_{\cE_n}$. The total Hamiltonian governing the interaction of $\cS$ and $\cE_n$ is thus 
\be
H_n\coloneqq H_\cS\otimes\un+\un\otimes H_{\cE_n}+V_n
\label{eq:prodstuff}
\ee
acting on $\cH_\cS\otimes\cH_{\cE_n}$. In the sequel, we will drop the symbols $\otimes \un$ and $\un \otimes$ when the context is clear. The reduced propagator $\cL_n$, obtained by tracing out the probe's degrees of freedom,
$$
\cL_n\rho\coloneqq\tr_{\cH_{\cE_n}}(\e^{-\i\tau_n H_n}(\rho\otimes\rho_{\cE_n})\e^{\i\tau_n H_n}),
$$
describes the evolution of the state of $\cS$ due to its interaction with $\cE_n$ for a duration $\tau_n>0$.

By construction, $\cL_n$ is a Completely Positive Trace Preserving (CPTP for short) map on the set $\cB^1(\cH_\cS)$ of trace class operators on $\cH_\cS$. Consequently, the state $\rho_n$ of the system $\cS$ after interacting with the first $n$ probes is given by 
$$
\rho_n\coloneqq\cL_n\dots\cL_1\rho_0.
$$ 
{RIS} have been introduced and developed in the non-exhaustive list of papers~\cite{attal2006repeated,BJMas,BJMrd,OAM,bruneau2010repeated,Electron,nechita2012random,Bru14,BJM,HJPR,HJPR2,BoBr,movassagh2019ergodic,movassagh2020theory}.

To take into account the uncontrollable odds that may affect the probes and their interaction with the system in real physical applications, it makes sense to generalize the above setup and consider random dynamical systems
\be
\rho_n(\bomega)\coloneqq\cL_{\omega_n}\cdots\cL_{\omega_1}\rho_{\omega_0},
\label{eq:rds}
\ee
where  $\Omega\ni\omega\mapsto\rho_\omega\in\cB^1(\cH_\cS)$,
$\Omega\ni\omega\mapsto\cL_\omega$ takes its values in CPTP-maps on $\cB^1(\cH_\cS)$ and $\bomega=(\omega_n)_{n\in\NN}\in\Omega^\NN$ denotes the sample path of a stochastic process.

One way to think about this setup is to consider that, according to some probabilistic rules, the probes are drawn from a pool which is partitioned into a family $(\cR_\omega)_{\omega\in\Omega}$, each $\cR_\omega$ containing identical probes. Assuming that all the probes are in thermal equilibrium, the physical situation is then close to the one of a small system $\cS$ interacting with several thermal reservoirs $\cR_\omega$. We shall henceforth call the $\cR_\omega$'s {\sl reservoirs}. For a given $\omega$, $\cL_\omega$ represents the effect on $\cS$ of an interaction with the reservoir $\cR_\omega$. 

As an example, consider the previously mentioned one-atom maser. The interaction times $\tau_n$ depend on the flight velocities of the atoms which fluctuate. In this case, a reasonable model would be
$$
\cL_{\tau}\rho\coloneqq\tr_{\cH_{\cE}}(\e^{-\i\tau H}\rho\otimes\rho_{\cE}\e^{\i\tau H}),
$$
the probe Hilbert space $\cH_\cE$, the total Hamiltonian $H$ and the probe state $\rho_{\cE}$ being independent of $n$, and $\boldsymbol{\tau}=(\tau_n)_{n\in\NN^\ast}$ being an i.i.d.\;sequence of random variables. Now assume that the atoms are in thermal equilibrium with the oven, {\sl i.e.,} $\rho_\cE^\beta=\e^{-\beta H_\cE}/\tr(\e^{-\beta H_\cE})$ where $\beta$ denotes the oven's inverse temperature. To take fluctuations of the latter into account, we should consider a process where $\omega_n=(\tau_n,\beta_n)$ and
$$
\cL_{(\tau,\beta)}\rho\coloneqq\tr_{\cH_{\cE}}(\e^{-\i\tau H}\ \rho \otimes \rho_{\cE}^\beta\ \e^{\i\tau H}).
$$
In this case, it may be more realistic to allow for some correlations between successive probes. 

The purpose of this article is to study the large $n$ behavior of the random dynamical system~\eqref{eq:rds} when the driving stochastic process is a Markov chain with a transition matrix $P$ and initial probability vector $\pi$ on a finite state space $\Omega$. We will call such a system a {\sl Markovian Repeated Interaction System}, MRIS for short. We will first consider~\eqref{eq:rds} in the abstract setting where $(\cL_\omega)_{\omega\in\Omega}$ is a family of arbitrary CPTP-maps on $\cB^1(\cH_\cS)$ for a finite dimensional Hilbert space $\cH_\cS$. From a mathematical perspective, our main result is a pointwise ergodic theorem for MRIS. We will then apply this abstract result to the more concrete and physically motivated case where each $\cL_\omega$ describes the interaction of a system $\cS$ with a probe $\cE_\omega$. We shall restrict our attention to MRIS with thermal reservoirs, {\sl i.e.,} to the cases where each probe $\cE_\omega$ is in thermal equilibrium at a given inverse temperature $\beta_\omega$. This specific setting allows us to develop the nonequilibrium thermodynamics of MRIS, using tools from information theory and dynamical systems. We investigate energy transfers between the system $\cS$ and the reservoirs and the induced entropy production. Under microscopic time-reversibility, and invoking results of~\cite{CJPS}, we will derive a strong and detailed form of nonequilibrium fluctuation relations, in the spirit of~\cite{DHP20}. Following~\cite{Gal96}, we will show how these relations reduce to linear response theory near (appropriately defined) thermal equilibrium. To achieve this, we need to consider repeated two-time measurements of the probes, and hence to deal with quantum trajectories naturally associated to the random dynamical system~\eqref{eq:rds}. In this respect, our results are complementary to~\cite{KM04,Kumm06,AGPS15,CP15,CP16}.

Our framework also allows us to consider situations where the driving Markov chain is not homogeneous in time, meaning that the transition matrix $P$ depends on the time step. Assuming the variations between successive transition matrices is small, we derive the large time asymptotics of the corresponding random dynamical system, in this instance of the adiabatic regime.

Several limiting cases of MRIS have been addressed in previous works: these are the deterministic case with a single reservoir, the periodic case where $P$ is a cyclic permutation matrix, and the i.i.d.\;case where $\bomega$ is a sequence of independent and identically distributed random variables. The deterministic case --- the simplest example of RIS --- was introduced and developed in~\cite{attal2006repeated,BJMas}. In particular, the latter reference shows the strict positivity of entropy production even in this simple setup. The periodic case, mentioned in \cite{BJM}, provides another toy model of RIS allowing for a study of its nonequilibrium thermodynamics. The i.i.d.\;case has been analyzed in~\cite{BJMrd,BJMmatrices}, see also~\cite{nechita2012random}. The nonequilibrium thermodynamics of both these models has been recently investigated in~\cite{BoBr}. The papers~\cite{movassagh2019ergodic,movassagh2020theory}  address more general abstract situations where the underlying stochastic process is an ergodic dynamical system, and the large time asymptotics of the states is investigated by means of methods of dynamical systems.

\medskip\noindent{\bf Remark.} We consider Markovian randomness in the RIS framework to take into account classical correlations between fluctuating probes. The case of quantum entanglement between the probes has been studied, for instance, in~\cite{fqw,raquepas2020fermionic,ajr20}. 

\subsection*{Feynman--Kac Formalism}
In quantum mechanics one is generally interested in quantum expectation values of observables which can depend on time and, in particular, in the large time behavior of those expectations. In our model we will focus on observables  which depend on which reservoir $\cS$ is currently interacting with. Consequently, we will study the large $n$ properties of
\begin{equation}\label{eq:defqumarexp}
\tr(\rho_{n}(\bomega)\ X(\omega_{n+1})),
\end{equation}
where $\Omega\ni\omega\mapsto X(\omega)\in\cB(\cH_\cS)$. The idea to deal with this model is to work in the Hilbert space
$$
\cK=\ell^2(\Omega;\cH_\cS)
$$
of square-summable, $\cH_\cS$-valued functions on $\Omega$. Considering the $C^\ast$-algebra $\fA\subset\cB(\cK)$ of functions $X:\Omega\to\cB(\cH_\cS)$ acting as a multiplication operator on $\cK$, a straightforward computation shows that
\begin{equation}\label{eq:fkformalism}
\EE\left[\tr\left(\rho_{n}(\bomega)X(\omega_{n+1})\right)\right]=\tr_\cK\left(X\LL^nR\right),
\end{equation}
where $\EE$ denotes the expectation w.r.t.\;the Markov chain, and $\LL$ is a CPTP map on $\cB^1(\cK)$, constructed with the maps $\cL_\omega$ and the transition matrix $P$, see~\eqref{eq:LLdef}, and\footnote{We shall denote by $1_A$ the indicator function of a set $A$.}
$$
R(\omega)\coloneqq\EE[\rho_{\omega_0}1_{\{\omega_1=\omega\}}].
$$
Consequently, the dynamics of $\cS$ can be modeled with the discrete time semigroup $(\LL^n)_{n\in\NN}$. This structure is called a {\sl Feynman--Kac formalism}. Following~\cite{Pillet1, Pillet, Pillet2}, it turns the study of the large $n$ behavior of~\eqref{eq:defqumarexp} into a spectral problem for $\LL$. See also~\cite{Schenkeretal} for a discussion of a somewhat related problem concerning random time dependent Lindbladians. We note that, in the periodic and i.i.d.\;cases mentioned above, the MRIS dynamics can be modeled by a semigroup acting on $\cB^1(\cH_\cS)$, that is, without the need for a Feynman--Kac formalism.

In the i.i.d.\;case, it has been proved in~\cite{BJMrd} that, if the random CPTP map $\cL_{\omega_0}$ is primitive with nonzero probability, then~\eqref{eq:defqumarexp} converges almost surely in Cesàro mean to the deterministic value $\tr(\rho_+\EE[X(\omega_0)])$, $\rho_+$ being the unique invariant state of the CPTP map $\EE[\cL_{\omega_0}]$. This result was extended to a more general class of i.i.d.\;random dynamical systems in~\cite{BoBr}. In our Markovian framework,
we shall see in Theorem~\ref{thm:ascv} that the mere irreducibility of $\LL$ implies that~\eqref{eq:defqumarexp} converges almost surely in Cesàro mean to $\tr(R_+X)$, where $R_+$ is the unique invariant state of $\LL$. This is in keeping with the recent results in~\cite{movassagh2019ergodic,movassagh2020theory} about ergodic positive linear map valued processes. However, our spectral point of view allowed by the Feynman--Kac formalism is distinct from the approach there, which makes essential use of positivity properties. Moreover, the spectral approach turns out to be instrumental in the subsequent analysis of the thermodynamic properties of MRIS.

A natural extension of the above results concern the case of an inhomogeneous Markov chain. A first step in this direction, dealt with in Section~\ref{sec:adia}, is the case of a Markov chain whose transition matrix changes infinitely slowly with time, the so-called adiabatic inhomogeneous case. Making use of the analogy with the adiabatic RIS studied in~\cite{HJPR,HJPR2}, we get Theorem~\ref{thm:adia}, that describes the large $n$ asymptotics of the expectation~\eqref{eq:fkformalism} in this framework.

At last, the MRIS framework is well suited to study nonequilibrium thermodynamic properties of RIS when each reservoir $\cR_\omega$ is in thermal equilibrium. In Section~\ref{sec:thermo}, we derive the entropy balance relation of MRIS process and investigate its consequences. We study the large time asymptotics of the full
statistics of entropy and energy transfers and give sufficient conditions ensuring the validity of entropic fluctuation relations. For processes running near appropriately defined equilibrium, we derive 
the main ingredients of linear response theory --- Green--Kubo formula, Onsager reciprocity relations and fluctuation-dissipation relations.

\subsection*{Organization of the Paper}

The core of this article is organized as follows. After setting up some conventions and notations, Section~\ref{sec:Feynman-Kac} introduces the Feynman--Kac formalism associated to the abstract framework of Markovian
repeated interaction quantum systems. The section ends with a few useful results on the spectral properties of the 
Feynman--Kac generator $\LL$. In Section~\ref{sec:CPMC}, we formulate and prove our main results on abstract MRIS: a pointwise ergodic theorem and an adiabatic theorem. The nonequilibrium thermodynamics of concrete MRIS is the subject of Section~\ref{sec:thermo}, while Section~\ref{sec:proofs} is devoted to the proofs.

\medskip\noindent{\bf Acknowledgements.} We thank Tristan Benoist for a insightful discussion. The research of J.-F.B. was partially funded by  the Agence Nationale de la Recherche under the programme “Investissements d’avenir” (ANR-15-IDEX-02), Cross Disciplinary Program ``Quantum Engineering Grenoble''.  A.J. and C.-A.P. acknowledge support from Agence Nationale de la Recherche grant NONSTOPS (ANR-17-CE40-0006-01).

\section{The Feynman--Kac Formalism of MRIS}
\label{sec:Feynman-Kac}

\subsection{Notations and Conventions}
\label{sec:NotaConv}

We start by setting up some notations and conventions that will be used in this work. We also recall some basic definitions concerning positive maps on $C^\ast$-algebras and their state spaces.

Let $\cH$ be a complex Hilbert space. The
inner product of two vectors $\varphi,\psi\in\cH$ will be denoted by
$\bra\varphi,\psi\ket$, and supposed to be anti-linear in the first argument and linear in the second one. 

$\cB(\cH)$ is the $C^\ast$-algebra of all linear operators on $\cH$, $\un$ denotes its unit, and $\spec(X)$ is the spectrum of $X\in\cB(\cH)$. Self-adjoint elements of $\cB(\cH)$ with spectrum in $\RR_+$ are said to be non-negative. $\cB(\cH)$ is ordered by the proper convex cone $\cB_+(\cH)$ of these non-negative elements, {\sl i.e.,} $X\ge Y$ iff $X-Y\in\cB_+(\cH)$. If $X\ge\delta\un$ for some $\delta>0$ we say that
$X$ is positive and write $X>0$.

The $\ast$-ideal of all trace class operators on $\cB(\cH)$ is denoted $\cB^1(\cH)$. Equipped with the trace norm $\|\alpha\|_1\coloneqq\tr(\sqrt{\alpha^\ast\alpha})$, $\cB^1(\cH)$ is a Banach space, the predual of $\cB(\cH)$. We write the corresponding duality as
\be
\bra\alpha,X\ket\coloneqq\tr(\alpha^\ast X).
\label{eq:Covid18}
\ee
Since $\cH$ is finite dimensional, $\cB^1(\cH)$ and $\cB(\cH)$ are identical topological vector spaces, but are distinct normed vector spaces. 
$\cB^1(\cH)$ is ordered by the proper convex cone $\cB^1_+(\cH)$ of its non-negative elements. A {\sl state} on $\cB(\cH)$ is an element $\rho\in\cB^1_+(\cH)$ normalized by $\|\rho\|_1=\bra\rho,\un\ket=1$. We denote by $\cB^1_{+1}(\cH)$ the closed convex set of these states. Its elements are also known as {\sl density matrices} or
{\sl statistical operators.} A state $\rho$ is said to be faithful whenever $\rho>0$, and pure whenever it is of rank one.

More generally, if $\cA\subset\cB(\cH)$ is a unital $C^\ast$-subalgebra\footnote{Any finite dimensional unital $C^\ast$-algebra can be realized in this way.}, we denote by $\cA_+\coloneqq\cA\cap\cB_+(\cH)$ the proper cone of non-negative elements of $\cA$, by $\cA_\ast$ its predual and by $\cA_{\ast+1}$ the set of states of $\cA$. The duality~\eqref{eq:Covid18} allows us to identify $\cA_\ast$ with the vector space $\cA$ equipped with the trace norm, and $\cA_{\ast+1}$ with the subset of density matrices in $\cA$.

A linear map $\Phi:\cA\to\cA$ is {\sl positivity preserving,} or simply {\sl positive,} written $\Phi\ge0$, whenever $\Phi(\cA_+)\subset\cA_+$. $\Phi$ is {\sl positivity improving}, written $\Phi>0$, if $\Phi(X)>0$ for $X\ge0$, $X\neq0$. A positive map $\Phi$ is said to be {\sl primitive} if $\Phi^n>0$ for some (and hence all sufficiently large) $n\in\NN$, and {\sl irreducible} if $\e^{t\Phi}>0$
for some (and hence all) $t>0$. Note that positivity improving implies primitivity, which implies irreducibility.

A positive map $\Phi$ on $\cA$ is said to be {\sl completely positive} when $\Phi\otimes\Id$ is a positive 
map on $\cA\otimes\cB(\CC^n)$ for all $n\in\NN$. In view of the identification of $\cA\otimes\cB(\CC^n)$ with the set $\cM_n(\cA)$ of $n\times n$-matrices with entries in $\cA$, the complete positivity of $\Phi$ means that for any finite set $J$ and any square 
matrix $A=[A_{ij}]_{i,j\in J}$ associated to a non-negative map on $\cH\otimes\CC^J$, the map associated to the matrix
$[\Phi(A_{ij})]_{i,j\in J}$ is non-negative. One easily shows that the latter condition is equivalent to
\be
\sum_{i,j\in J}B_i^\ast\Phi(A_i^\ast A_j)B_j\ge0
\label{eq:CPcrit}
\ee
for any finite families $(A_j)_{j\in J}$ and $(B_j)_{j\in J}$ in $\cA$~\cite[Corollary~3.4]{Ta1}. Any CP map $\Phi$ on
$\cA$ has a (non-unique) {\sl Kraus representation}\footnote{The Kraus representation of CP maps on $\cB(\CC^n)$ is discussed in any
textbook on quantum information, see, e.g., \cite[Theorem~2.2]{Pe08}. It is a simple corollary of Stinespring's dilation theorem~\cite[Theorem~3.6]{Ta1}. The slightly more general form used here is a direct consequence of Arveson's extension theorem~\cite[Theorem~1.2.3]{Arv69} which asserts that $\Phi$ extends to a CP map on $\cB(\cH)$.}
$$
\Phi(A)=\sum_{j\in J}V_j^\ast AV_j,
$$
where the finite family $\cV=(V_j)_{j\in J}\subset\cB(\cH)$ is called a Kraus family of $\Phi$. $\cV$ is said to be
irreducible whenever any non-zero element of $\cH$ is cyclic for the algebra generated by $\cV\cup\{\un\}$. 
$\cV$ is primitive when there exists $n>0$ such that any non-zero element of $\cH$ is cyclic for the linear span of
$\cV^n\coloneqq\{V_{j_1}\cdots V_{j_n}\mid j_1,\ldots,j_n\in J\}$. The CP map $\Phi$ is irreducible/primitive iff one (and hence any) of its Kraus family is irreducible/primitive.

In view of the above mentioned identification
of the sets $\cA$ and $\cA_\ast$, the various notions of positivity just introduced on endomorphisms of $\cA$ also apply to linear maps
$\cL:\cA_\ast\to\cA_\ast$. A linear map on $\cA$ is {\sl unital} when $\Phi(\un)=\un$ and CPU when completely positive and unital. A linear map on $\cA_\ast$ is {\sl trace preserving} when $\tr\circ\Phi=\tr$, and CPTP when
completely positive and trace preserving.

Decomposing $\alpha\in\cA_\ast$ into the positive and negative parts of its real and imaginary parts, one
easily shows that any affine map $\cL:\cA_{\ast+1}\to\cA_{\ast+1}$ has a unique linear extension to $\cA_\ast$. By construction, this extension is positive and trace preserving. Reciprocally, any linear, positive
and trace preserving map on $\cA_\ast$ restricts to an affine map on $\cA_{\ast+1}$. Hence, we shall
identify these two kinds of maps in the following. Denoting by $\cL^\ast:\cA\to\cA$ the adjoint of 
$\cL:\cA_\ast\to\cA_\ast$ w.r.t.\;the duality~\eqref{eq:Covid18}, $\cL^\ast$ is positive (resp.\;CPU) iff $\cL$ is positive (resp.\;CPTP). 

When $\cA$ is associated to a quantum mechanical system, self-adjoint elements $X$ of $\cA$ are observables
of this system and density matrices $\rho\in\cA_{\ast+1}$ describe its physical states. The physical quantity
described by the observable $X$ can only take numerical values in $\spec(X)$.
Denoting $E_X(I)$ the spectral projection of $X$ to the part of its spectrum contained in $I$, when the system's state is $\rho$, the probability to observe the
value of $X$ in $I$  is given by $\bra\rho,E_X(I)\ket$. In particular, the quantum mechanical expectation value of $X$ in the state $\rho$ is given by $\bra\rho,X\ket$.

We shall consider finite matrices $P\in\CC^{J\times J}$ as linear maps on the vector space $\CC^J$ interpreted as the commutative unital $C^\ast$-algebra of diagonal $J\times J$-matrices. The unit of this algebra is $\bun=(1,\ldots,1)\in\CC^J$. In this context, $P$ is positive (resp.\;positivity improving) whenever $P_{ij}\ge0$ (resp.\;$P_{ij}>0$) for all $i,j\in J$. We note that these are also the conditions of the classical Perron--Frobenius theory. We observe also that, due to the commutativity of the underlying $C^\ast$-algebra (see~\cite[Corollary~3.5]{Ta1}), $P$ is completely positive iff it is positive and CPU iff it is a right stochastic matrix.

Given a Banach space $\cX$, a finite set $\Omega$ and $p\in[1,\infty]$, we denote by $\ell^p(\Omega;\cX)$ the vector space $\cX^\Omega$ equipped with the $p$-norm
$$
\|x\|_p\coloneqq\begin{cases}
\ds\left(\sum_{\omega\in\Omega}\|x(\omega)\|^p\right)^{1/p}&\text{for }p<\infty;\\
\ds\max_{\omega\in\Omega}\|x(\omega)\|&\text{for }p=\infty.
\end{cases}
$$
We note that if $\cX^\ast$ is dual to $\cX$, with duality $\cX^\ast\times\cX\ni(\varphi,x)\mapsto\bra\varphi,x\ket$ and $\frac1p+\frac1q=1$, then
$\ell^q(\Omega;\cX^\ast)$ is dual to $\ell^p(\Omega;\cX)$ with duality
$\sum_{\omega\in\Omega}\bra\varphi(\omega),x(\omega)\ket$.

The product topology induced on $\bO=\Omega^\NN$ by the discrete topology of the finite set $\Omega$ is metrizable, the compatible metric 
$$
d((\omega_k)_{k\in\NN},(\nu_k)_{k\in\NN})\coloneqq2^{-\inf\{k\in\NN\mid\omega_k\neq\nu_k\}}
$$
makes $\bO$ a compact metric space. We denote by $\cO_n$ the finite algebra generated by the cylinder sets
$$
[\nu_0\cdots\nu_n]\coloneqq\{\bomega\in\bO\mid\omega_k=\nu_k\text{ for }0\le k\le n\},\qquad (\nu_0,\ldots,\nu_n\in\Omega),
$$
and by
$$
\cO\coloneqq\bigvee_{n\in\NN}\cO_n
$$
the generated $\sigma$-algebra.\footnote{By convention the cylinder with empty base $[\ ]$ is $\bO$.}
We recall that $\cO$ coincides with the Borel $\sigma$-algebra of $\bO$.\footnote{We will use the same notation, with the obvious modifications, whenever $\bO=\Omega^{\ZZ}$ or $\bO=\Omega^{\NN^\ast}$.}

Denote by $C(\bO)$ the Banach space of continuous real functions on $\bO$, and by $\cP(\bO)$ the set of probability measures on $(\bO,\cO)$, equipped with the weak-$\ast$ topology. We write the associated duality
as
$$
\langle f,\mu\rangle=\int_\bO f(\bomega)\mu(\d\bomega).
$$
We denote by $\mu_{|\cO_n}$ the restriction of $\mu\in\cP(\bO)$ to $\cO_n$.
Let $\phi$ be the left shift on $\bO$, and $\cP_\phi(\bO)$ the set of $\phi$-invariant elements
of $\cP(\bO)$.

For two probability measures $\nu$, $\mu$ on the same measurable space, we write $\nu\ll\mu$ whenever $\nu$ is absolutely continuous w.r.t.\;$\mu$, denoting by $\frac{\d\nu}{\d\mu}$ the corresponding Radon--Nikodym derivative. In this case, the relative entropy of $\nu$ w.r.t.\;$\mu$ is
$$
\Ent(\nu|\mu)\coloneqq\int\log\left(\frac{\d\nu}{\d\mu}\right)\d\nu,
$$
and for $\alpha\in\RR$, their relative Rényi $\alpha$-entropy is
$$
\Ent_\alpha(\nu|\mu)\coloneqq\log\int\left(\frac{\d\nu}{\d\mu}\right)^\alpha\d\mu.
$$
We recall that $\Ent(\nu|\mu)\ge0$, with equality iff $\nu=\mu$.

\subsection{Markovian Repeated Interaction Quantum Systems}

Let $\cS$ be a quantum mechanical system described by a finite dimensional Hilbert space $\cH_\cS$, denote by
$\cA=\cB(\cH_\cS)$ the associated $C^\ast$-algebra, set $\cA_\ast=\cB^1(\cH_\cS)$, let $(\rho_\omega)_{\omega\in\Omega}\subset\cA_\ast$ be a finite family of states and $(\cL_\omega)_{\omega\in\Omega}$ a finite family of CPTP maps on $\cA_\ast$. 

For a probability vector $\pi\in\RR^\Omega$, and a right stochastic matrix $P\in\RR^{\Omega\times\Omega}$, we consider the probability space $(\bO,\cO,\PP)$, where $\PP\in\cP(\bO)$ is the homogeneous Markov measure uniquely determined by
\begin{equation}\label{eq:markovproperty}
\PP([\omega_0\cdots\omega_n])\coloneqq\pi_{\omega_0}P_{\omega_0\omega_1}\dots P_{\omega_{n-1}\omega_n},
\end{equation}
for any cylinder $[\omega_0\cdots\omega_n]\in\cO_n$. In particular, the probability vector $\pi^{(n)}$ describing the distribution of the random variable $\omega_n$ is given by
$$
\pi^{(n)}_\nu=\PP(\omega_n=\nu)=\sum_{\omega_0,\ldots,\omega_{n-1}\in\Omega}\PP([\omega_0\cdots \omega_{n-1}\nu])=(\pi P^n)_\nu.
$$
We denote by $\EE[\,\cdot\,]$ the expectation functional associated with $\PP$. Given $A\in\cO$, $1_A$ is the characteristic
function of $A$ and $\EE[\,\cdot\,|A]=\EE[\,\cdot\,1_A]/\EE[1_A]$ is the conditional expectation given $A$. Finally, we recall that whenever $\pi$ is a left eigenvector of $P$ to the eigenvalue $1$, then $\PP\in\cP_\phi(\bO)$, {\sl i.e.,} the Markov chain is stationary. 

We associate to $(\pi,P,(\rho_\omega)_{\omega\in\Omega},(\cL_\omega)_{\omega\in\Omega})$ a random dynamics on the system $\cS$  as follows: Each sample path $\bomega\in\bO$ of our Markov chain induces a sequence of states $(\rho_n(\bomega))_{n\in\NN}$, where
$$
\rho_n(\bomega)\coloneqq\cL_{\omega_n}\cdots\cL_{\omega_1}\rho_{\omega_0},
$$
is the state of $\cS$ started in $\rho_{\omega_0}$ after interacting with the sequence of reservoirs $\cR_{\omega_1},\ldots,\cR_{\omega_n}$.

By Stinespring's dilation Theorem~\cite[Theorem~3.6]{Ta1}, any such random dynamical system describes a MRIS, {\sl i.e.,} there exists a family of probes $(\cE_\omega)_{\omega\in\Omega}$ such that
$$
\cL_\omega(\rho)=\tr_{\cH_{\cE_\omega}}(U_\omega(\rho\otimes\rho_{\cE_\omega})U_\omega^\ast)
$$
for some probe states $\rho_{\cE_\omega}\in\cB^1_{+1}(\cH_{\cE_\omega})$, unitary propagators $U_\omega$ on $\cH_\cS\otimes\cH_{\cE_\omega}$ and all $\omega\in\Omega$.

\subsection{Extended Observables and Extended States}

The natural semigroup structure of autonomous dynamical systems gets lost in this
random setting. However, using the Markov property, it is possible to restore this structure
at the level of ``classical'' expectations. To this end, let us introduce the extended Hilbert space
of $\cH_\cS$-valued functions on the noise space $\Omega$
$$
\cK=\ell^2(\Omega;\cH_\cS).
$$
Identifying $\cK$ with $\cH_\cS\otimes\CC^\Omega$ leads to
the identification of $\cB(\cK)$ with $\cA^{\Omega\times\Omega}$, the $C^\ast$-algebra of $\Omega\times\Omega$ matrices with entries in $\cA$.
Further, identifying $X\in\ell^\infty(\Omega;\cA)$ with a diagonal matrix in $\cA^{\Omega\times\Omega}$ gives a natural isometric injection
\be
\fA=\ell^\infty(\Omega;\cA)\hookrightarrow\cB(\cK),
\label{eq:extobs}
\ee
which makes $\fA$ a $C^\ast$-subalgebra of $\cB(\cK)$. We shall say that elements of $\fA$ are {\sl extended observables} of the MRIS associated with $(\pi,P,(\rho_\omega)_{\omega\in\Omega},(\cL_\omega)_{\omega\in\Omega})$. The map
$$
X\mapsto\Tr X\coloneqq\sum_{\omega\in\Omega}\tr X(\omega),
$$
is the restriction to $\fA$ of the trace of $\cB(\cK)$ and defines a trace on $\fA$, {\sl i.e.,} a positive linear functional such that $\Tr(XY)=\Tr(YX)$ for all $X,Y\in\fA$.

Dual to the injection~\eqref{eq:extobs} is the surjection
\be
\cB^1(\cK)\rightarrow\fA_\ast=\ell^1(\Omega;\cA_\ast),
\label{eq:extstate}
\ee
which sends $R=[R_{\omega\nu}]_{\omega,\nu\in\Omega}\in\cA_\ast^{\Omega\times\Omega}$ to the diagonal map $\omega\mapsto R_{\omega\omega}$. Note that elements of $\cB^1_{+1}(\cK)$
are sent to positive valued elements of $\ell^1(\Omega;\cA_\ast)$ of $1$-norm $1$. These are precisely the states of $\fA$ w.r.t.\;the duality
\be
\bra R,X\ket\coloneqq\Tr(R^\ast X)=\sum_{\omega\in\Omega}\bra R(\omega),X(\omega)\ket.
\label{eq:fAduality}
\ee
Positive and normalized elements of $\fA_\ast$, images of states $R\in\cB^1_{+1}(\cK)$ by~\eqref{eq:extstate}, will be called {\sl extended states} of the MRIS, the set of these extended states is denoted $\fA_{\ast+1}$. We shall associate to $(\pi,P,(\rho_\omega)_{\omega\in\Omega},(\cL_\omega)_{\omega\in\Omega})$ the sequence of extended state $(R_n)_{n\in\NN}$ defined by
$$
R_n(\omega)\coloneqq\pi^{(n+1)}_\omega\EE[\rho_n(\bomega)\,|\,\omega_{n+1}=\omega]=\EE[\rho_n(\bomega)1_{\omega_{n+1}=\omega}],
$$
and in particular
\be
R_0(\omega)=\sum_{\nu\in\Omega}\pi_\nu P_{\nu\omega}\rho_\nu.
\label{eq:R0form}
\ee
Observe that these states have ``marginals''
$$
\pi^{(n+1)}_\omega=\tr R_n(\omega),\qquad \bar\rho_n=\EE[\rho_n(\bomega)]=\sum_{\omega\in\Omega}R_n(\omega).
$$

\subsection{The Semigroup}

\begin{lemm}\label{lem:LL}
Let $(R_n)_{n\in\NN}\subset\fA_\ast$ be the sequence of extended states of the MRIS associated to $(\pi,P,(\rho_\omega)_{\omega\in\Omega},(\cL_\omega)_{\omega\in\Omega})$. Then, for $n\in\NN$, one has
\be
R_n=\LL^nR_0,
\label{eq:semigroup}
\ee
where\/ $\LL:\fA_\ast\to\fA_\ast$ is the CPTP map defined by
\be
(\LL R)(\omega)\coloneqq\sum_{\nu\in\Omega}P_{\nu\omega}\cL_\nu R(\nu).
\label{eq:LLdef}
\ee
In particular, given an extended observable $X\in\fA$, one has
\be
\EE[\bra\rho_n(\bomega),X(\omega_{n+1})\ket]=\bra\LL^nR_0,X\ket.
\label{eq:tomatos}
\ee
\end{lemm}
\begin{proof}
The relation~\eqref{eq:semigroup} clearly holds for $n=0$. For any $n\in\NN^\ast$ and $\omega\in\Omega$ one has
\begin{align*}
R_n(\omega)&=\sum_{\omega_0,\ldots,\omega_{n}\in\Omega}\pi_{\omega_0}P_{\omega_0\omega_1}\cdots P_{\omega_n\omega}\cL_{\omega_n}\cdots\cL_{\omega_1}\rho_{\omega_0}\\
&=\sum_{\omega_n\in\Omega}P_{\omega_n\omega}\cL_{\omega_n}\sum_{\omega_0,\ldots,\omega_{n-1}\in\Omega}\pi_{\omega_0}P_{\omega_0\omega_1}\cdots P_{\omega_{n-1}\omega_n}\cL_{\omega_{n-1}}\cdots\cL_{\omega_1}\rho_{\omega_0}\\
&=\sum_{\omega_n\in\Omega}P_{\omega_n\omega}\cL_{\omega_n}R_{n-1}(\omega_{n})=(\LL R_{n-1})(\omega),
\end{align*}
which implies~\eqref{eq:semigroup} for all $n\in\NN$. Consequently, one has\footnote{We denote by $\EE[\,\cdot\,|\omega_n]$ the conditional expectation w.r.t.\;the random variable $\omega_n$.}
\begin{align*}
\EE[\bra\rho_n(\bomega),X(\omega_{n+1})\ket]&=\EE[\bra\EE[\rho_n(\bomega)\,|\,\omega_{n+1}],X(\omega_{n+1})\ket]\\
&=\sum_{\omega\in\Omega}\pi^{(n+1)}_\omega\bra\EE[\rho_n(\bomega)\,|\,\omega_{n+1}=\omega],X(\omega)\ket\\
&=\sum_{\omega\in\Omega}\bra R_n(\omega),X(\omega)\ket=\bra\LL^nR_0,X\ket,
\end{align*}
which proves~\eqref{eq:tomatos}. An elementary calculation shows that the adjoint of\/ $\LL$ w.r.t.\;the duality~\eqref{eq:fAduality}
is given by
$$
(\LL^\ast X)(\omega)=\sum_{\nu\in\Omega}P_{\omega\nu}\cL_\omega^\ast X(\nu).
$$
One also easily checks (e.g., using~\eqref{eq:CPcrit}) that $\LL^\ast$ is a CPU-map on $\fA$. By duality, $\LL$ is a CPTP-map on $\fA_\ast$.
\end{proof}

By Brouwer's fixed point theorem, $\LL$ has a fixed point on the convex set $\fA_{\ast+1}$ of extended states.
Such a fixed point will be called {\sl Extended Steady State} (ESS) of the MRIS.

Let $R_+\in\fA_{\ast+1}$ be an ESS, and set
\begin{equation}
\pi_{+\omega}\coloneqq\tr R_+(\omega),\qquad \rho_{+\omega}\coloneqq
\begin{cases}
	\ds\frac{\cL_\omega R_+(\omega)}{\pi_{+\omega}}&\text{if }\pi_{+\omega}\neq0,\\[10pt]
	\ds\frac\un{\tr\un}&\text{otherwise.}
\end{cases}
\label{eq:R+magic}
\end{equation}
Since $R_+$ is positive with $\Tr R_+=1$, $\pi_+$ is a probability vector, and it follows from 
$$
\pi_{+\omega}=\tr R_+(\omega)=\tr(\LL R_+)(\omega)=\sum_{\nu\in\Omega}P_{\nu\omega}\tr\cL_\nu R_+(\nu)
=\sum_{\nu\in\Omega}P_{\nu\omega}\tr R_+(\nu)=\sum_{\nu\in\Omega}\pi_{+\nu}P_{\nu\omega}
$$
that it is an invariant probability for the driving Markov chain. By construction, $\rho_+=(\rho_{+\omega})_{\omega\in\Omega}$ is a family of states on $\cA$.
Since $R_+(\omega)\ge0$, $R_+(\omega)=0$ whenever $\pi_{+\omega}=0$, so that the identity $\cL_\omega R_+(\omega)=\pi_{+\omega}\rho_{+\omega}$ holds for all $\omega\in\Omega$. The formula
\be
\sum_{\nu\in\Omega}P_{\nu\omega}\pi_{+\nu}\rho_{+\nu}=\sum_{\nu\in\Omega}P_{\nu\omega}\cL_\nu R_+(\nu)
=(\LL R_+)(\omega)=R_+(\omega),
\label{eq:rho+magic}
\ee
thus allows to reconstruct the ESS $R_+$ from the invariant probability $\pi_+$ and the family  $\rho_{+}$.

Denoting by $\PP_+$ the stationary Markov measure on $(\bO,\cO)$ associated to $(\pi_+,P)$, one has,
with $\rho_{+n}(\bomega)\coloneqq\cL_{\omega_n}\cdots\cL_{\omega_1}\rho_{+\omega_0}$,
\begin{align*}
	\EE_+[\rho_{+n}(\bomega)1_{\omega_{n+1}=\omega}]
	&=\sum_{\omega_0,\ldots,\omega_n\in\Omega}\pi_{+\omega_0}P_{\omega_0\omega_1}\cdots P_{\omega_n\omega}
	\cL_{\omega_n}\cdots\cL_{\omega_1}\rho_{+\omega_0}\\
	&=\sum_{\omega_0,\ldots,\omega_n\in\Omega}P_{\omega_n\omega}\cL_{\omega_n}\cdots P_{\omega_0\omega_1}\cL_{\omega_0}R_+(\omega_0)\\
	&=(\LL^{n+1}R_+)(\omega)=R_+(\omega)
\end{align*}
so that
\be
\EE_+[\bra\rho_{+n}(\bomega),X(\omega_{n+1})\ket]=\bra R_+,X\ket,
\label{eq:R+Form}
\ee
for all $X\in\fA$ and $n\in\NN$. 

\begin{defi}
We shall say that the MRIS associated to $(\pi_+,P,(\rho_{+\omega})_{\omega\in\Omega},(\cL_\omega)_{\omega\in\Omega})$ is the {\sl stationary MRIS}\/ induced by the ESS $R_+$.
\end{defi}

\noindent{\bf Remark.} Given a $P$-invariant probability $\pi_+$ and a family of states $(\rho_{+\omega})_{\omega\in\Omega}$
on $\cA$ such that
$$
\cL_\omega\sum_{\mu\in\Omega}P_{\mu\omega}\pi_{+\mu}\rho_{+\mu}=\pi_{+\omega}\rho_{+\omega}
$$
for all $\omega\in\Omega$, the extended state $R_+$ defined by~\eqref{eq:rho+magic} satisfies
$$
(\LL R_+)(\omega)=\sum_{\nu\in\Omega}P_{\nu\omega}\cL_\nu\sum_{\mu\in\Omega} P_{\mu\nu}\pi_{+\mu}\rho_{+\mu}
=\sum_{\nu\in\Omega}P_{\nu\omega}\pi_{+\nu}\rho_{+\nu}=R_+(\omega).
$$

\subsection{Primitivity and Irreducibility as Ergodic Properties}
\label{sec:ljpm}

The irreducibility and primitivity properties of CP maps $\Phi$ will play a central role in our approach.
As introduced in Section~\ref{sec:NotaConv}, they are intimately linked to the natural order structure of
a $C^\ast$-algebra. In this section, we recall their important connections with the spectral properties of $\Phi$ and the ergodic properties of the semigroup $(\Phi^n)_{n\in\NN}$. We also discuss some issues more specific to
the CPTP map~\eqref{eq:LLdef}.

Let $\Phi$ be a CPTP map on the predual $\cC_\ast$ of the finite dimensional $C^\ast$-algebra $\cC\subset\cB(\cH)$. The following statements are well known (see, e.g., \cite{EHK} for details):

\medskip
\begin{enumerate}[label=(\roman*)]
\item $\spec(\Phi)$ is a subset of the closed unit disk containing $1$, and the eigenspace of $\Phi$ associated to the eigenvalue $1$ contains a state $\rho$.
\item $\Phi$ is irreducible iff its eigenvalue $1$ is simple. In this case, the eigenspace of $\Phi$ associated to the eigenvalue $1$ contains a unique state $\rho$. Moreover, $\rho$ is faithful and for any $R\in\cC_\ast$ one has
$$
\left|\frac1n\sum_{k=0}^{n-1}\Phi^k(R)-\bra R,\un\ket\rho\right|=\|R\|\cO(n^{-1}),
$$
as $n\to\infty$.
\item $\Phi$ is primitive,  iff its simple eigenvalue $1$ is gapped, {\sl i.e.,} 
$$
\Delta=-\log\max\{|z|\mid z\in\spec(\Phi)\setminus\{1\}\}>0.
$$ 
Moreover, for any $R\in\cC_\ast$ and $\epsilon>0$ one has
$$
|\Phi^n(R)-\bra R,\un\ket\rho|=\|R\|\cO(\e^{-n(\Delta-\epsilon)}),
$$
as $n\to\infty$.
\end{enumerate}

\medskip
We conclude this section with some useful connections between the properties of $P$ and $\cL_\omega$ and that of\, $\LL$.

\begin{lemm}\label{lem:Lspec}
Set
$$
\bar\cL\coloneqq\sum_{\omega\in\Omega}\pi_\omega\cL_\omega,	
$$
where $\pi$ is a faithful probability vector, and note that $\bar\cL$ is a CPTP map on $\fA_\ast$.
\begin{enumerate}[label=(\roman*)]
\item $\LL$ is positivity improving iff $P$ and every $\cL_\omega$ are.
\item If\/ $\LL$ is irreducible (resp.\;primitive), so are $P$ and $\bar\cL$.
\item If\/ $P$ is positivity improving, then\/ $\LL$ is irreducible (resp.\;primitive) iff $\bar\cL$ is.
\item If every $\cL_\omega$ is positivity improving, then $\LL$ is irreducible (resp.\;primitive) iff\/ $P$ is.
\end{enumerate}
\end{lemm}

\begin{proof}
(i) Assume that $P$ and each $\cL_\omega$ are positivity improving. Given non-zero $R\in\fA_{\ast+}$ and $X\in\fA_+$, there exists $\omega,\nu\in\Omega$ such that both $R(\omega)$ and $X(\nu)$ are non-negative and non-zero. It follows that $\cL_\omega R(\omega)>0$ and hence
$$
\bra\LL R,X\ket\ge P_{\omega\nu}\bra\cL_\omega R(\omega),X(\nu)\ket>0,
$$
which shows that $\LL$ is positivity improving. Reciprocally, let $\LL$ be positivity improving,
$\xi,\eta\in\Omega$, and $\rho\in\cA_{\ast+}$, $A\in\cA_+$ be non-zero. With $R(\omega)=\delta_{\omega\xi}\rho$ and $X(\omega)=\delta_{\omega\eta}A$, one has
$$
0<\bra\LL R,X\ket=P_{\xi\eta}\bra\cL_\xi\rho,A\ket,
$$
from which we can conclude that $P$ and all $\cL_\omega$ are positivity improving.

\medskip\noindent(ii) For $R\in\fA_\ast$ and $n\in\NN$ one easily checks that
\be
(\LL^n R)(\omega)=\sum_{\omega_1,\ldots,\omega_n\in\Omega}P_{\omega_1\omega_2}\cdots P_{\omega_n\omega}\cL_{\omega_n}\cdots\cL_{\omega_1}R(\omega_1).
\label{eq:cofee}
\ee
For $u\in\RR^\Omega$, let us set $R_u(\omega)\coloneqq u_{\omega}\rho$ and $X_u(\omega)\coloneqq u_\omega\un$ where $\rho$ is an arbitrary density matrix. Formula~\eqref{eq:cofee} gives that
$$
\bra\LL^nR_u,X_v\ket=\bra u,P^n v\ket
$$
for any $u,v\in\RR^\Omega$ and $n\in\NN$, and it follows that $P$ is irreducible (resp.\;primitive) if $\LL$ is.
Setting $R(\omega)\coloneqq\pi_\omega\rho$ and $\epsilon\coloneqq\min_\omega\pi_\omega$, and using the fact that the matrix elements of $P$ are all bounded above by $1$, we further deduce from~\eqref{eq:cofee} 
\begin{align*}
(\LL^n R)(\omega)&\le\sum_{\omega_1,\ldots,\omega_n\in\Omega}\cL_{\omega_n}\cdots\cL_{\omega_1}\pi_{\omega_1}\rho\\
&\le\epsilon^{-n+1}\sum_{\omega_1,\ldots,\omega_n\in\Omega}\pi_{\omega_n}\cL_{\omega_n}\cdots\pi_{\omega_1}\cL_{\omega_1}\rho=\epsilon^{-n+1}\bar\cL^n\rho,
\end{align*}
from which one easily concludes that $\bar\cL$ is irreducible (resp.\;primitive) whenever $\LL$ is.

\medskip\noindent(iii) If $P$ is positivity improving, then there exists $\delta$ such that
$0<\delta\le P_{\nu\omega}$ for all $\omega,\nu\in\Omega$. Setting again $R(\omega)\coloneqq\pi_\omega\rho$, it follows from~\eqref{eq:cofee} that
$$
\delta^n\bar\cL^n\rho\le(\LL^n R)(\omega),
$$
from which it follows that $\LL$ is irreducible (resp.\;primitive) whenever $\bar\cL$ is. The reciprocal property follows from Part~(ii).

\medskip\noindent(iv) Let $R\in\fA_{\ast+}$ be non-zero, so that $R(\nu)$ is non-zero for some
$\nu\in\Omega$. Formula~\eqref{eq:cofee} yields
$$
(\LL^n R)(\omega)\ge\sum_{\omega_2,\ldots,\omega_n\in\Omega}P_{\nu\omega_2}\cdots P_{\omega_n\omega}\cL_{\omega_n}\cdots\cL_{\omega_2}\cL_{\nu}R(\nu),
$$
and if every $\cL_\omega$ is positivity improving, then
$$
\delta\coloneqq\min_{\omega_2,\ldots,\omega_n\in\Omega}\min\spec(\cL_{\omega_n}\cdots\cL_{\omega_2}\cL_{\nu}R(\nu))>0,
$$
so that
$$
(\LL^n R)(\omega)\ge (P^n)_{\nu\omega}\delta\un,
$$
from which we can conclude that $\LL$ is irreducible (resp.\;primitive) whenever $P$ is. The reciprocal property
again follows from Part~(ii).
\end{proof}

\section{Main Abstract Results}
\label{sec:CPMC}
\subsection{A Pointwise Ergodic Theorem}
\label{sec:largetime}\label{sec:positivity}\label{sec:ergo}

Central to our results are consequences of spectral properties of the map $\LL$ on the large time behavior of MRIS. Below, we formulate a pointwise ergodic theorem for MRIS which applies when the driving Markov chain is stationary. 
The case of an inhomogeneous driving process is considered in 
Section~\ref{sec:adia} in the adiabatic limit. Applications of these results to the nonequilibrium thermodynamic properties of MRIS will be
given in the subsequent Section~\ref{sec:thermo}.

Our ergodic theorem relies on the minimal

\begin{hypo}[{{\bf STAT}}]\label{STAT}
Either the Markov chain is stationary, i.e., $\pi P=\pi$, or it admits a faithful stationary state $\pi_+$.
\end{hypo}

We note that whenever $\LL$ is irreducible (or primitive, or positivity improving), then, by Lemma~\ref{lem:Lspec}~(ii), the unique left eigenvector of $P$ is faithful, so that Assumption~(\nameref{STAT}) is satisfied.

We shall invoke the following vector-valued random ergodic theorem of Beck and Schwartz.

\begin{theo}[{\cite[Theorem~2]{BS57}}]
\label{thm:Beck+Schwartz}
Let $\cX$ be a reflexive Banach-space and let $(S,\Sigma,m)$ be a $\sigma$-finite measure space. Let there be defined a strongly measurable function $S\ni s\mapsto T_s$ with values in the Banach-space $\cB(\cX)$ of bounded linear operators on $\cX$. Suppose that $\|T_s\|\le1$ for all $s\in S$. Let $h$ be a measure-preserving transformation on $(S,\Sigma, m)$. Then for each $X\in L^1(S,\cX)$ there is an $\bar X\in L^1(S,\cX)$ such that 
$$
\lim_{N\to\infty}\frac1N\sum_{n=0}^{N-1}T_{s}T_{h(s)}\cdots T_{h^{n-1}(s)}X(h^n(s))=\bar X(s), 
$$
strongly in $\cX$, a.e.\;in $S$, and
$$
\bar X(s)=T_s(\bar X(h(s)))
$$
a.e.\;in S. Moreover, if $m(S)<\infty$, then $\cX$ is also the limit in the mean of order $1$.
\end{theo}

Let us start by assuming that the Markov chain is stationary. A direct application of Theorem~\ref{thm:Beck+Schwartz} shows that $\PP$-a.e.\;and in $L^1(\bO,\PP;\fA)$ the limit
$$
\bar{X}(\bomega)\coloneqq\lim_{N\to\infty}\frac1N\sum_{n=0}^{N-1}\cL^\ast_{\omega_1}\cdots\cL^\ast_{\omega_n}X(\omega_{n+1})
$$
exists for all $X\in\fA$. 
Moreover, the covariance relation $\bar{X}(\bomega)=\cL_{\omega_1}^\ast\bar{X}\circ\phi(\bomega)$ implies
\begin{align*}
	\bar{\bar{X}}(\omega)
	\coloneqq&\EE[\bar{X}(\bomega)\,|\,\omega_1=\omega]\\
	=&\EE[\cL^\ast_{\omega_1}\bar{X}(\phi(\bomega))\,|\,\omega_1=\omega]\\
	=&\cL^\ast_{\omega}\EE[\bar{X}(\phi(\bomega))\,|\,\omega_1=\omega]\\
	=&\cL^\ast_{\omega}\EE[\bar{X}(\bomega)\,|\,\omega_0=\omega]\\
	=&\cL^\ast_{\omega}\EE[\EE[\bar{X}(\bomega)\,|\,\omega_1]\,|\,\omega_0=\omega]\\
	=&\cL^\ast_{\omega}\EE[\bar{\bar X}(\omega_1)\,|\,\omega_0=\omega]\\
	=&\sum_{\nu\in\Omega}P_{\omega\nu}\cL^\ast_{\omega}\bar{\bar X}(\nu)
	=(\LL^\ast\bar{\bar X})(\omega).
\end{align*}

Assuming now that $\pi_+$ is a faithful stationary state, {\sl i.e.,} $\pi_+P=\pi_+>0$, then
the stationary Markov measure $\PP_+$ on $(\bO,\cO)$ with transition matrix $P$ and faithful invariant probability $\pi_+$ satisfies
$$
\frac{\d\PP}{\d\PP_+}(\bomega)=\frac{\pi(\omega_0)}{\pi_+(\omega_0)},
$$
which gives that $\PP$ is absolutely continuous w.r.t.\;$\PP_+$. Consequently, any $\PP_+$-a.s.\;property also holds $\PP$-almost surely. Summarizing, we have proved

\begin{theo}\label{thm:pergo} Under Assumption~{\rm (\nameref{STAT})}, the limit
\be
\bar{X}(\bomega)\coloneqq\lim_{N\to\infty}\frac1N\sum_{n=0}^{N-1}\cL^\ast_{\omega_1}\cdots\cL^\ast_{\omega_n}X(\omega_{n+1})
\label{eq:ergav}
\ee
exists\/ $\PP$-almost surely and in $L^1(\bO,\PP;\fA)$ for any $X\in\fA$. The limiting function is such that 
\be
\cL^\ast_{\omega_1}\bar{X}(\phi(\bomega))=\bar{X}(\bomega). 
\label{eq:barXcovariance}
\ee
Moreover, the extended observable $\bar{\bar{X}}\in\fA$ defined by
$$
\bar{\bar{X}}(\omega)\coloneqq\EE[\bar{X}(\bomega)\,|\,\omega_1=\omega],
$$
satisfies
\be
\LL^\ast\bar{\bar{X}}=\bar{\bar{X}}.
\label{eq:LstarFix}
\ee
\end{theo}

As an immediate corollary, we get that 
\be
\lim_{N\rightarrow\infty}\frac{1}{N}\sum_{n=0}^{N-1}\bra\rho_{n}(\bomega),X(\omega_{n+1})\ket
=\bra\rho_{\omega_0},\bar X(\bomega)\ket=\bra\rho_k({\bomega}),\bar X\circ\phi^k(\bomega)\ket
\label{eq:expval}
\ee
$\PP$-almost surely and in $L^1(\bO,\PP)$, for any $k\in\NN$.

A less immediate corollary of Theorem~\ref{thm:pergo} is the following, which only requires spectral properties of
the map $\LL$.

\begin{theo}\label{thm:ascv}
If\/ $\LL$ is irreducible, its unique ESS $R_+$ is faithful and for any $X\in\fA$ one has\,\footnote{When $\LL$ is irreducible, we sometimes omit to mention the dependence of $\bar X(\bomega)$
on $\bomega$ and write its $\PP$-a.s.\;value as $\bar X$.}
$$
\bar X(\bomega)=
\lim_{N\rightarrow+\infty}\frac{1}{N}\sum_{n=0}^{N-1}\cL^\ast_{\omega_1}\cdots\cL^\ast_{\omega_n}X(\omega_{n+1})
=\bra R_+, X\ket\un
$$
for\/ $\PP$-almost every $\bomega\in\bO$ and in $L^1(\bO,\PP;\fA)$. In particular, for any initial state $\rho\in\fA_{\ast+1}$ and any $X\in\fA$,
$$
\lim_{N\rightarrow+\infty}\frac{1}{N}\sum_{n=0}^{N-1}\bra\rho_n(\bomega),X(\omega_{n+1})\ket=\bra R_+, X\ket
$$
holds\/ $\PP$-almost surely and in $L^1(\bO,\PP)$. Moreover, if\/ $\LL$ is primitive, then
$$
\lim_{n\to\infty}\EE_+[\bra\rho_n(\bomega),X(\omega_{n+1})\ket]=\bra R_+, X\ket.
$$
\end{theo}

\begin{proof}Since $\LL$ is irreducible, so is $P$ by Lemma~\ref{lem:Lspec}~(ii), and hence Assumption~(\nameref{STAT}) is satisfied and Theorem~\ref{thm:pergo} applies. Moreover, the measure $\PP_+$ is $\phi$-ergodic~\cite[Theorem~1.19]{Wa}. Iterating the covariance relation~\eqref{eq:barXcovariance} and invoking the Markov property lead us to
\be
\begin{split}
\EE_+[\bar X(\bomega)\,|\,\cO_n]
&=\EE_+[\cL^\ast_{\omega_1}\cdots\cL^\ast_{\omega_{n-1}}\bar X\circ\phi^{n-1}(\bomega)\,|\,\cO_n]\\
&=\cL^\ast_{\omega_1}\cdots\cL^\ast_{\omega_{n-1}}\EE_+[\bar X\circ\phi^{n-1}(\bomega)\,|\,\omega_{n}]\\
&=\cL^\ast_{\omega_1}\cdots\cL^\ast_{\omega_{n-1}}\bar{ \bar X}(\omega_{n}).
\end{split}
\label{eq:qbit}
\ee
$\LL$ being irreducible, it has a simple eigenvalue $1$ with left/right eigenvector $\un/R_+$.\footnote{Here, $\un$ denotes the unit of $\fA$.} From $\LL^\ast\un=\un$ and~\eqref{eq:LstarFix} we deduce that $\bar{\bar X}=\bra\lambda,X\ket\un$ for some $\lambda\in\fA_\ast$. Relation~\eqref{eq:qbit} further yields
$$
\EE_+[\bar X(\bomega)\,|\,\cO_n]=\bra\lambda,X\ket\un,
$$
for all $n\in\NN$. Letting $n\to\infty$ we get~\cite[Corollary~5.22]{Br}
$$
\bar X(\bomega)=\lim_{n\to\infty}\EE_+[\bar X(\bomega)\,|\,\cO_n]=\bra\lambda,X\ket\un,
$$
and in particular,
$$
\bra\rho_{\omega_0},\bar X(\bomega)\ket=\bra\lambda,X\ket\bra\rho_{\omega_0},\un\ket=\bra\lambda,X\ket.
$$
Finally, using~\eqref{eq:tomatos} and invoking the mean ergodic theorem gives
$$
\EE[\bra\rho_{\omega_0},\bar X(\bomega)\ket]
=\lim_{N\to\infty}\frac1N\sum_{n=0}^{N-1}\EE[\bra\rho_n(\bomega),X(\omega_{n+1})]
=\lim_{N\to\infty}\frac1N\sum_{n=0}^{N-1}\bra\LL^nR_0,X\ket
=\bra R_+,X\ket,
$$
from which we conclude that $\lambda=R_+$ and $\bar X(\bomega)=\bra R_+,X\ket\un$. The remaining statements are obvious.
\end{proof}

\subsection{Inhomogeneous MRIS in the Adiabatic Limit}\label{sec:adia}

In this section, we consider the possibility of driving our MRIS with an {\sl inhomogeneous} Markov chain, {\sl i.e.,} we allow the transition matrix $P$ to depend on the time step. More specifically, we address the regime of slow variation from one time step to the next, known as the {\sl adiabatic regime}, by applying the results of~\cite{HJPR,HJPR2} to control the expectations $$
\EE[\bra\rho_n(\bomega),X(\omega_{n+1})\ket]
$$ 
when $n$ goes to infinity.

\medskip
Notice that the whole construction that has been made in Section~\ref{sec:Feynman-Kac} does not rely on the homogeneity of the Markov chain, so that, given a sequence $(P^{(n)})_{n\in\NN^\ast}$ of right stochastic matrices, we can replace~\eqref{eq:markovproperty} with
\begin{equation*}
\PP([\omega_0\cdots\omega_n])\coloneqq\pi_{\omega_0}P_{\omega_0\omega_1}^{(1)}P_{\omega_1\omega_2}^{(2)}\cdots P_{\omega_{n-1}\omega_n}^{(n)}.
\end{equation*}
Further, defining the sequence $(\LL_n)_{n\in\NN^\ast}$ by
$$
(\LL_n R)(\omega)\coloneqq\sum_{\nu\in\Omega}P^{(n)}_{\nu\omega}\cL_\nu R(\nu),
$$
one can write, in analogy with Lemma~\ref{lem:LL},
\begin{equation}\label{eq:inhomogenousfk}
\EE[\bra\rho_n(\bomega),X(\omega_{n+1})\ket]=\bra\LL_n\cdots\LL_1R_0,X\ket.
\end{equation}
At such a level of generality, not much can be said of~\eqref{eq:inhomogenousfk}. Hence, we consider an adiabatic regime in which successive transition matrices are obtained by sampling a smooth family of transition matrices defined on $[0,1]$, with smaller and smaller variations between successive transition matrices, following~\cite{HJPR,HJPR2}.

As before, let $(\cL_\omega)_{\omega\in\Omega}$ be a finite family of CPTP maps on $\cA_\ast$. Let $[0,1]\ni s\mapsto P(s)$ be a $C^2$ function taking its values in the set of right stochastic $\Omega\times\Omega$-matrices. Define
$$
(\LL(s)R)(\omega)\coloneqq\sum_{\nu\in\Omega}P_{\nu\omega}(s)\cL_\nu R(\nu)
$$
so that $[0,1]\ni s\mapsto\LL(s)$ is also of class $C^2$. For $\epsilon=1/N$ denote by $\PP_\epsilon$ the Markov probability measure on $\Omega^{N+1}$ defined as
$$
\PP_\epsilon(\bomega=(\omega_0,\ldots,\omega_N))
=\pi_{\omega_0}P_{\omega_0\omega_1}(0)P_{\omega_1\omega_2}(\epsilon)\cdots P_{\omega_{N-1}\omega_N}((N-1)\epsilon).
$$
It was shown in~\cite{HJPR2} that, under spectral hypotheses on the family $(\LL(s))_{s\in[0,1]}$, the product 
$$
\LL_\epsilon^{(n)}=\LL(n\epsilon)\LL((n-1)\epsilon)\cdots\LL(\epsilon)
$$
admits an asymptotics for $n\in\{0,\dots,N\}$, in the adiabatic limit $\epsilon\downarrow0$. More precisely, we have 

\begin{theo}[{\cite[Corollary~3.14]{HJPR2}}]\label{thm:adia}
Assume that, for all $s\in [0,1]$, $\LL(s)$ is primitive. Denote by $R_+(s)$ its unique (and faithful) invariant state, and fix $R\in\fA_{\ast+1}$. Then, there exist $0<l<1$, $0<\epsilon_0<1$ and $C<\infty$,  such that, for all $\epsilon\in]0,\epsilon_0]$ with $1/\epsilon\in\NN$, and all $n \in \{0,\dots,1/\epsilon\}\subset\NN$,
\begin{align*}
\|\LL_\epsilon^{(n)}R-R_+(n\epsilon)\|\leq C\left(\frac\epsilon{(1-l)}+l^{n+1}\right).
\end{align*}
In particular, for any $X\in\fA$, 
\begin{equation*}
\EE_\epsilon[\bra\rho_n(\bomega),X(\omega_{n+1})\ket]=\bra R_+(n\epsilon),X\ket+\|X\|\,O(\epsilon/(1-l)+l^{n+1}).
\end{equation*}
If, moreover, we write $n=t/\epsilon$, with $0<t<1$, then
\begin{equation*}
\EE_\epsilon[\bra\rho_{\frac t\epsilon}(\bomega),X(\omega_{\frac{t+\epsilon}\epsilon})\ket]=\bra R_+(t),X\ket+\|X\|\,O(\epsilon).
\end{equation*}
\end{theo}

\medskip\noindent{\bf Remark.} A similar statement holds when the $\LL(s)$'s are irreducible, with a more complicated expression for the asymptotic state, in terms of all eigenvectors corresponding to the peripheral eigenvalues.

\section{Entropic Fluctuations and Linear Response Theory}\label{sec:thermo}

This section is devoted to applications to linear response theory and entropic fluctuations within our Markovian framework, generalizing~\cite{BoBr} which addresses these aspects for {\sl periodic} and {\sl uniform i.i.d.} repeated interaction systems. 

Within our setup, a periodic RIS (called cyclic in~\cite{BoBr}) is a MRIS such that $P$ is the permutation matrix associated to a cycle
$(12\cdots m)\in S_m$, so that $\bomega$ is the repetition of the word $12\cdots m$ with probability $1$. A i.i.d.\;RIS corresponds to the case $P_{\nu\omega}=\pi_\omega$ for all $\omega\in\Omega$ where $\pi$ is a probability vector. In this case $\bomega$ is a sequence of independent random variables~\cite{BJM}. The
special case of uniform distribution is simply called a Random RIS in~\cite{BoBr}.

In both cases, the reduced dynamics of the small system yields a discrete time quantum dynamical semigroup on $\cA_\ast=\cB^1(\cH_\cS)$, {with $\dim \cH_\cS<\infty$}. Indeed, in the periodic case, 
$$
\rho_{nm}(\bomega)=(\cL_m\cdots\cL_1)^n\rho
$$
while, in the random case,
$$
\EE[\rho_n(\bomega)]=\left(\frac1{|\Omega|}\sum_{\omega\in\Omega}\cL_\omega\right)^n\rho.
$$
The results of~\cite{BoBr} are largely based on this semigroup structure. We shall invoke Lemma~\ref{lem:LL}, to extend these results to the MRIS framework. 

\medskip\noindent{\bf Remark.} Continuous time semigroups of CPTP maps, often also called {\sl quantum Markov semigroups,} are generated by Lindbladians. As dynamics of a small quantum system $\cS$ interacting with several extended thermal reservoirs, they emerge in the van Hove weak coupling limit~\cite{davies1974markovian,davies1976markovian,davies1975markovian,davies1976quantum}. 
The interested reader should consult~\cite{LS,JPW} for discussions of the nonequilibrium thermodynamics of such systems.

\subsection{Entropy Balance}

Returning to the concrete MRIS with CPTP and dual CPU maps
$$
\cL_\omega\rho\coloneqq\tr_{\cH_{\cE_\omega}}\left(U_\omega(\rho\otimes\rho_{\cE_\omega})U_\omega^\ast\right),\qquad
\cL_\omega^\ast X=\tr_{\cH_{\cE_\omega}}\left(U_\omega^\ast(X\otimes\un)U_\omega(\un\otimes\rho_{\cE_\omega})\right),
$$
where $U_\omega\coloneqq\e^{-\i\tau_\omega(H_\cS+H_{\cE_\omega}+V_\omega)}$, we shall now assume that the reservoirs are in thermal equilibrium,
more precisely

\begin{hypo}[{{\bf KMS}}]\label{KMS}
There is $\bb=(\beta_\omega)_{\omega\in\Omega}\in\RR_+^\Omega$ such that, for all $\omega\in\Omega$,
\begin{equation*}
\rho_{\cE_\omega}=\e^{-\beta_\omega(H_{\cE_\omega}-F_\omega)},
\end{equation*}
where
$$
F_\omega\coloneqq-\frac1\beta_\omega\log\tr\left(\e^{-\beta_\omega H_{\cE_\omega}}\right)
$$
is the free energy of a probe from reservoir $\cR_\omega$.
\end{hypo}

{\noindent\bf Remark.} Given $\bb\in\RR_+^\Omega$ and assuming each $\rho_{\cE_\omega}$ to be faithful, it is of course possible to redefine the probe Hamiltonians in such a way that~(\nameref{KMS}) holds, at the cost of absorbing the change of $H_{\cE_\omega}$ in the interaction $V_\omega$. However, thermodynamic considerations require 
the propagators $U_\omega$ to be independent of $\bb$, so that such circumstances are excluded in the following.  

\medskip
The energy lost by the system $\cS$ during the $n+1^\mathrm{th}$ interaction, which we interpret as the amount of heat dumped in the reservoir $\cR_{\omega_{n+1}}$, is
$$
\Delta Q_{n+1}(\bomega)\coloneqq\tr\left(\rho_n(\bomega)\otimes\rho_{\cE_{\omega_{n+1}}}
(U_{\omega_{n+1}}^\ast H_{\cE_{\omega_{n+1}}} U_{\omega_{n+1}}-H_{\cE_{\omega_{n+1}}})\right)
=-\bra\rho_n(\bomega),J(\omega_{n+1})\ket
$$
with
$$
J(\omega)\coloneqq\tr_{\cH_{\cE_\omega}}(U_\omega^\ast[U_\omega,H_{\cE_\omega}](\un\otimes\rho_{\cE_\omega})).
$$
Further setting
\be
J_\nu:\Omega\ni\omega\mapsto\delta_{\nu\omega}J(\nu),
\label{eq:springoutthere}
\ee
yields an extended observable $J_\nu\in\fA$ describing the energy transferred from reservoir $\cR_\nu$ to the system $\cS$ during a single interaction. 

Accordingly, the time-averaged quantum mechanical expectation values of the heat extracted from the reservoir $\cR_\nu$ during a single interaction is given by
\begin{equation}\label{eq:asympenergy}
\lim_{N\rightarrow\infty}\frac{1}{N}\sum_{n=0}^{N-1}\bra\rho_n(\bomega),J_\nu(\omega_{n+1})\ket.
\end{equation}
If $\LL$ is irreducible, then Theorem~\ref{thm:ascv} implies that this limit exists for $\PP$-almost every $\bomega\in\bO$, and coincides with the ESS ensemble average
\be
\bra R_+,J_\nu\ket=\EE_+[\bra\rho_+(\omega_0),J_\nu(\omega_1)\ket].
\label{eq:greenwater}
\ee
The von Neumann entropy~\cite[Section~3.3]{Pe08} of the system $\cS$ after completion of the $n^\mathrm{th}$ interaction is
$$
S(\rho_n(\bomega))\coloneqq-\bra\rho_n(\bomega),\log\rho_n(\bomega)\ket,
$$
so that, during the $n+1^\mathrm{th}$ interaction, the system entropy decreases by
$$
\Delta S_{n+1}(\bomega)\coloneqq S(\rho_n(\bomega))-S(\rho_{n+1}(\bomega)).
$$
Recall that the relative entropy of a state $\mu\in\cA_\ast$ relatively to another state $\rho\in\cA_\ast$ is 
$$
\Ent(\mu\vert\rho)\coloneqq\begin{cases}
	\tr(\mu(\log\mu-\log\rho))&\text{if }\Ran\mu\subset\Ran\rho;\\
	+\infty&\text{otherwise},
\end{cases} 
$$ 
and that $\Ent(\mu\vert\rho)\geq 0$ with equality iff $\mu=\rho$ (see~\cite[Section~3.4]{Pe08}). Setting
$$
\ep_n(\bomega)
\coloneqq\Ent\left(U_{\omega_{n+1}}\left(\rho_n(\bomega)\otimes\rho_{\cE_{\omega_{n+1}}}\right)U_{\omega_{n+1}}^\ast
\,\bigg|\,\rho_{n+1}(\bomega) \otimes \rho_{\cE_{\omega_{n+1}}}\right),
$$
an elementary calculation yields
\begin{align*}
\ep_n(\bomega)=&-S_n(\bomega)-\beta_{\omega_{n+1}}\tr\left(\left(\rho_n(\bomega)\otimes\rho_{\cE_{\omega_{n+1}}}\right)(H_{\cE_{\omega_{n+1}}}-F_{\omega_{n+1}})\right)\\
&+S_{n+1}(\bomega)+\beta_{\omega_{n+1}}\tr\left(
\left(\rho_n(\bomega)\otimes\rho_{\cE_{\omega_{n+1}}}\right)U_{\omega_{n+1}}^\ast(H_{\cE_{\omega_{n+1}}}-F_{\omega_{n+1}})U_{\omega_{n+1}}\right),
\end{align*}
which can be rewritten as the one-step entropy balance relation
\be
\Delta S_{n+1}(\bomega)+\ep_n(\bomega)=\beta_{\omega_{n+1}}\Delta Q_{n+1}(\bomega).
\label{eq:balalaika}
\ee
Identifying the right-hand side of this identity with the amount of entropy dissipated in the
reservoir $\cR_{\omega_{n+1}}$, $\ep_n(\bomega)$ can be interpreted as the entropy produced by the
interaction process. The inequality $\ep_n(\bomega)\ge0$ thus becomes the expression of the $2^\mathrm{nd}$--law of thermodynamics, and yields Landauer's lower bound
$$
\Delta Q_{n+1}(\bomega)\ge\frac{\Delta S_{n+1}(\bomega)}{\beta_{\omega_{n+1}}},
$$ 
on the energetic cost of a reduction of the system entropy (see~\cite{RW14,JPlan,HJPR,HJPR2} for more details and discussions). Summing over $n$ we get
$$
\frac{S_0(\bomega)-S_N(\bomega)}{N}+\frac1N\sum_{n=0}^{N-1}\ep_n(\bomega)
=-\sum_{\nu\in\Omega}\beta_\nu\frac1N\sum_{n=0}^{N-1}\bra\rho_n(\bomega),J_\nu(\omega_{n+1})\ket.
$$
Observing that $0\le S_n(\bomega)\le\log\dim\cH_\cS$ and recalling~\eqref{eq:expval}, we deduce that whenever the limit~\eqref{eq:asympenergy}
exists, the following expression of the time-averaged entropy production
\be
\overline{\ep}(\bomega)\coloneqq\lim_{N\to\infty}\frac1N\sum_{n=0}^{N-1}\ep_n(\bomega)
=-\sum_{\nu\in\Omega}\beta_\nu\bra\rho_{\omega_0},\bar J_\nu(\bomega)\ket\ge0
\label{eq:barsigmadef}
\ee
holds, where $\bar J_\nu$ is the ergodic average~\eqref{eq:ergav}. This applies, in particular, under Assumption~(\nameref{STAT}), and expresses then the $2^\mathrm{nd}$--law in the context of steady-state thermodynamics. If, moreover, $\LL$ is irreducible,
then~\eqref{eq:greenwater} yields
\be
\overline{\ep}(\bomega)=-\sum_{\nu\in\Omega}\beta_\nu\bra R_+,J_\nu\ket
\label{eq:epirred}
\ee
$\PP$-a.s.

\medskip\noindent{\bf Remark.} Assuming~(\nameref{STAT}) but replacing Assumption~(\nameref{KMS}) by the faithfulness of the probe states $\rho_{\cE_\omega}$, setting\footnote{Recall the convention made after~\eqref{eq:prodstuff}}
$$
J_S(\omega)\coloneqq\tr_{\cH_{\cE_\omega}}\left(U_\omega^\ast[S_{\cE_\omega},U_\omega](\un\otimes\rho_{\cE_\omega})\right),
\qquad
S_{\cE_\omega}\coloneqq-\log\rho_{\cE_\omega},
$$
and repeating the previous calculation leads to the entropy balance equation
\be
\Delta S_{n+1}(\bomega)+\ep_n(\bomega)=\bra\rho_n(\bomega),J_S(\omega_{n+1})\ket
\label{eq:balalaika2}
\ee
with the time-average
$$
\overline{\ep}(\bomega)=\lim_{N\to\infty}\frac1N\sum_{n=0}^{N-1}\ep_n(\bomega)
=\bra\rho_{\omega_0},\overline{J}_S(\bomega)\ket.
$$
The irreducibility of $\LL$ further leads to the $\PP$-a.s. identity
\be
\overline{\ep}(\bomega)=\bra R_+,J_S\ket.
\label{eq:lesperles}
\ee
The extended observable $J_S$ describes the entropy dumped into the reservoirs during a
single interaction.

\medskip
The vanishing of entropy production is a signature of thermodynamic equilibrium, as such, it will play a central role in our discussion of linear response in Section~\ref{sec:linearresponse}. Our next result is a quite general necessary and sufficient condition for
the vanishing of entropy production in MRIS (see Section~\ref{sect:proofch3ne} for the proof).

\begin{prop}\label{prop:ch3ne}
Assume that the probe states $\rho_{\cE_\omega}$ are faithful and that\/ $\LL$ is irreducible. Then, the time averaged entropy production~\eqref{eq:lesperles} vanishes\/ $\PP$-a.s.\;iff 
the family of states $(\rho_{+\omega})_{\omega\in\Omega}$ associated with the unique ESS $R_+$ in~\eqref{eq:R+magic} satisfies
\be
U_\omega(\rho_{+\nu}\otimes\rho_{\cE_\omega})U_\omega^\ast=\rho_{+\omega}\otimes\rho_{\cE_\omega}
\label{eq:lespecheurs}
\ee
for all pairs $(\nu,\omega)\in\Omega\times\Omega$ such that $P_{\nu\omega}>0$. In this case, further assuming~{\rm (\nameref{KMS})}, the entropy balance equation
\be
S(\rho_{+\omega_{n+1}})-S(\rho_{+\omega_n})=\beta_{\omega_{n+1}}\bra\rho_{+\omega_n},J(\omega_{n+1})\ket=0.
\label{eq:Geneva}
\ee
holds\/ $\PP_+$-a.s., and in particular,
\be
\bar J_\nu=\bra R_+,J_\nu\ket=0
\label{eq:stillraining}
\ee
for all $\nu\in\Omega$.
\end{prop}

\noindent{\bf Remarks.} 1. Whenever $P_{\nu\omega}>0$ and $P_{\nu'\omega}>0$, it follows from~\eqref{eq:lespecheurs} that $\rho_{+\nu}=\rho_{+\nu'}$.

\medskip\noindent 2. The very special case where all reservoirs are identical in the sense that all the
CPTP maps $\cL_\omega$ coincide with some $\cL$ will be of interest 
in Section~\ref{sec:linearresponse}. One easily checks that in this circumstance the mean state 
$$
\bar\rho_+\coloneqq\sum_{\omega\in\Omega}R_+(\omega)=\sum_{\omega\in\Omega}\pi_{+\omega}\rho_{+\omega}
$$
satisfies $\cL\bar\rho_+=\bar\rho_+$. Setting $R(\omega)=\pi_{+\omega}\bar\rho_+$, we derive
$$
(\LL R)(\omega)=\sum_{\nu\in\Omega}P_{\nu\omega}\cL R(\nu)
=\sum_{\nu\in\Omega}\pi_{+\nu}P_{\nu\omega}\cL\bar\rho_+
=\sum_{\nu\in\Omega}\pi_{+\nu}P_{\nu\omega}\bar\rho_+
=\pi_{+\omega}\bar\rho_+=R(\omega).
$$
If $\LL$ is irreducible, then we can conclude that $R=R_+$, and~\eqref{eq:R+magic} yields
$$
\pi_{+\omega}\rho_{+\omega}=\cL R(\omega)=\pi_{+\omega}\cL\bar\rho_+=\pi_{+\omega}\bar\rho_+,
$$
so that the states $\rho_{+\omega}$ all coincide with $\bar\rho_+$.

\subsection{Full Statistics of Entropy}

While the entropy balance equation~\eqref{eq:balalaika} and the related Landauer bound are clearly supporting the
proposed interpretation of the quantum expectation value of the observable $J_\nu$, as defined in~\eqref{eq:springoutthere}, the fact that a real measurement of the energy transfer between the system $\cS$ and the
reservoir $\cR_\nu$ is unlikely to proceed through an ``instantaneous'' measurement of $J_\nu$ casts some doubts on the
physical meaning of its higher moments, and more generally of its spectral measure.
This issue is not specific to RIS but is relevant to more general open quantum systems (see~\cite{TLH07} and~\cite[Section~5.10]{JOPP}). In this section, we adopt an operational point of view and consider a two-time measurement protocol of the {\sl entropy observables}
\be
S_{\cE_\omega}\coloneqq-\log\rho_{\cE_\omega}=\beta_\omega(H_{\cE_\omega}-F_\omega).
\label{eq:wednesday}
\ee
Set $\Sigma\coloneqq\cup_{\omega\in\Omega}\spec(S_{\cE_\omega})$ and for $s\in\Sigma$ let $\Pi_s(\omega)$
denote the spectral projection of $S_{\cE_\omega}$ for the eigenvalue $s$.\footnote{The convention here is that $\Pi_s(\omega)=0$ whenever $s\not\in\spec(S_{\cE_\omega})$.} 

Let the state of the joint system just before the coupling of $\cS$ to $\cE_{\omega_{n+1}}$ be the product state $\rho=\rho_\cS\otimes\rho_{\cE_{\omega_{n+1}}}$. A first measurement of $S_{\cE_{\omega_{n+1}}}$ before the $n+1^\mathrm{th}$ interaction  thus yields the value $s\in\Sigma$ with probability 
$$
p(s)=\tr\left(\rho\Pi_s(\omega_{n+1})\right)
$$
and leaves the joint system in the state
$$
\rho_{|s}=
\frac{\Pi_s(\omega_{n+1})\rho\Pi_s(\omega_{n+1})}{\tr\left(\rho\Pi_s(\omega_{n+1})\right)}.
$$
Once the interaction between $\cE_{\omega_{n+1}}$ and $\cS$ is complete, the state of the joint system has evolved to
$$
\rho'_{|s}=U_{\omega_{n+1}}\rho_{|s}U_{\omega_{n+1}}^\ast=
\frac{U_{\omega_{n+1}}\Pi_s(\omega_{n+1})\rho\Pi_s(\omega_{n+1})U_{\omega_{n+1}}^\ast}{\tr\left(\rho\Pi_s(\omega_{n+1})\right)}
$$
and the probability for the outcome of a second measurement of $S_{\cE_{\omega_{n+1}}}$ to be $s'\in\Sigma$ is
$$
p(s'\,|\,s)=\tr(\rho'_{|s}\Pi_{s'}(\omega_{n+1}))=
\frac{\tr\left(U_{\omega_{n+1}}\Pi_s(\omega_{n+1})\rho\Pi_s(\omega_{n+1})U_{\omega_{n+1}}^\ast\Pi_{s'}(\omega_{n+1})\right)}{\tr\left(\rho\Pi_s(\omega_{n+1})\right)},
$$
the second measurement reducing the state of the joint system to
$$
\rho_{|ss'}=\frac{\Pi_{s'}(\omega_{n+1})U_{\omega_{n+1}}\Pi_s(\omega_{n+1})\rho\Pi_s(\omega_{n+1})U_{\omega_{n+1}}^\ast\Pi_{s'}(\omega_{n+1})}{\tr\left(U_{\omega_{n+1}}\Pi_s(\omega_{n+1})\rho\Pi_s(\omega_{n+1})U_{\omega_{n+1}}^\ast\Pi_{s'}(\omega_{n+1})\right)}.
$$
Applying Bayes' rule, we conclude that the joint probability for the two successive measurements of $S_{\cE_{\omega_{n+1}}}$ to have the outcome $\xi=(s,s')\in\Sigma\times\Sigma$ is
$$
p(\xi)=p(s'\,|\,s)p(s)
=\tr\left(U_{\omega_{n+1}}\Pi_s(\omega_{n+1})\rho\Pi_s(\omega_{n+1})U_{\omega_{n+1}}^\ast\Pi_{s'}(\omega_{n+1})\right)
=\tr_{\cH_\cS}(\cL_{\omega_{n+1},\xi}\rho_\cS),
$$
where $\cL_{\omega,\xi}$ is the CP map defined by
\be
\cL_{\omega,\xi}\rho\coloneqq\e^{-s}
\tr_{\cH_{\cE_\omega}}((\un\otimes\Pi_{s'}(\omega))U_\omega(\rho\otimes\Pi_{s}(\omega))U_\omega^\ast).
\label{eq:twotimeentropy}
\ee
Moreover, the state of the system $\cS$ after the two measurements is
$$
\rho_\cS'=\tr_{\cH_{\cE_{\omega_{n+1}}}}(\rho'_{|ss'})
=\frac{\cL_{\omega_{n+1},\xi}\rho_\cS}{\tr_{\cH_\cS}(\cL_{\omega_{n+1},\xi}\rho_\cS)}.
$$
Collecting the Markov sample path $\omega_1,\omega_2,\ldots$ and the sequence $\xi_1,\xi_2,\ldots$ of measurement outcomes into a single {\sl extended quantum trajectory} $\bx=(x_k)_{k\in\NN^\ast}$, with 
$$
x_k=(\omega_k,\xi_k)\in\fX:=\bigcup_{\omega\in\Omega}\{\omega\}\times\Xi_\omega,\qquad
\Xi_\omega\coloneqq\spec(S_{\cE_\omega})\times\spec(S_{\cE_\omega}),
$$ 
yields the {\sl Entropy Process}, a stochastic process with path space $\bX:=\fX^{\NN^\ast}$ and
law $\QQ_{R_0}\in\cP(\bX)$ determined by\footnote{Recall~\eqref{eq:R0form}.}
\begin{align*}
\QQ_{R_0}([x_1\cdots x_n])
&\coloneqq\sum_{\omega_0\in\Omega}\pi_{\omega_0}P_{\omega_0\omega_1}\cdots P_{\omega_{n-1}\omega_n}
\tr(\cL_{\omega_n,\xi_n}\cdots\cL_{\omega_1,\xi_1}\rho_{\omega_0})\\[4pt]
&=P_{\omega_1\omega_2}\cdots P_{\omega_{n-1}\omega_n}
\tr(\cL_{\omega_n,\xi_n}\cdots\cL_{\omega_1,\xi_1}R_0(\omega_1)).
\end{align*}
According to a widely used physical terminology, the measure $\QQ_{R_0}$ on the $\sigma$-algebra $\cX$ is the {\sl Full Statistics of Entropy}, we denote the associated expectation functional by $\EE_{R_0}$. 

Observing that
$$
\sum_{x_1\in\fX}\QQ_{R_0}([x_1\cdots x_n])
=\QQ_{\LL R_0}([x_2\cdots x_n])
$$
we conclude that if $R_0$ is an ESS, then the entropy process is stationary, $\QQ_{R_0}\in\cP_\phi(\bX)$.

\subsubsection{Level-1: Entropy Production}

In this section, we consider the statistical properties of the entropy increments. To this end, we set
$$
\Sigma\times\Sigma\ni\xi=(s,s')\mapsto\delta\xi\coloneqq s'-s,
$$
and for $\boldsymbol{f}\in C(\bX,\RR^d)$,
$$
S_n\boldsymbol{f}=\sum_{k=0}^{n-1}\boldsymbol{f}\circ\phi^k.
$$
Introducing the functions
$$
\boldsymbol{\fJ}
=(\fJ_\omega)_{\omega\in\Omega}\in C(\bX,\RR^\Omega),\qquad
\fJ_\omega(\bx)=1_{\{\omega_1=\omega\}}\delta\xi_1,
$$
the total entropy increments in each reservoir during the first $N$ interactions is described by the random vector $S_N\boldsymbol{\fJ}$.
Our next results concern the large $N$ asymptotics of the sequence $(S_N\boldsymbol{\fJ})_{N\in\NN^\ast}$. We shall see, in particular, that at the level of expectations this asymptotics is governed by the observables~\eqref{eq:springoutthere}. To formulate these results we need to introduce a few new objects. For $\alpha\in\RR$ let
\be
\cL_\omega^{[\alpha]}
\coloneqq\sum_{\xi\in\Xi_\omega}\e^{-\alpha \delta\xi}\cL_{\omega,\xi},
\label{eq:cLalphadef}
\ee
and for $\ba=(\alpha_\nu)_{\nu\in\Omega}\in\RR^\Omega$, define a CP map on $\fA_\ast$ by
\be
(\LL^{[\ba]}R)(\omega)\coloneqq\sum_{\nu\in\Omega}P_{\nu\omega}\cL_\nu^{[\alpha_\nu]}R(\nu).
\label{eq:LLalphadef}
\ee
Finally, denote by $\ell(\ba)$ the spectral radius of $\LL^{[\ba]}$, which coincides with its dominant eigenvalue. The next two theorems summarize the asymptotic properties of the entropy increments,
their proofs will be given in Section~\ref{sec:prooflimits}.

\begin{theo}[Limit Theorems and Large Deviations Principle]\label{thm:Limits}
Let $R_0$ be an arbitrary extended state in $\fA_{\ast+1}$.
\hspace{0em}%
\begin{enumerate}[label=(\roman*)]
\item\label{it:Limits-i} Under Assumption~{\rm(\nameref{STAT})}, for all $\nu\in\Omega$, one has
$$
\lim_{N\to\infty}\EE_{R_0}\left[\frac1NS_N\fJ_{\nu}\right]=-\beta_\nu\EE[\bra\rho_{\omega_0},\bar J_\nu(\bomega)\ket].
$$
\end{enumerate}

\noindent In the following, we assume that\/ $\LL$ is irreducible, with ESS $R_+$.

\medskip
\begin{enumerate}[label=(\roman*)]
\setcounter{enumi}{1}
\item\label{it:Limits-ii} The strong law of large numbers holds, i.e., $\QQ_{R_0}$-a.s.,
$$
\lim_{N\to\infty}\frac1NS_N\fJ_\nu=-\beta_\nu\bra R_+,J_\nu\ket.
$$
\item\label{it:Limits-iii} The limit\footnote{Here and in the sequel ``$\,\boldsymbol{\cdot}\,$'' denotes the Euclidean inner product.}
\be
e(\ba)\coloneqq\lim_{N\to\infty}\frac1N\log\EE_{R_0}[\e^{-\ba\cdot S_N\boldsymbol{\fJ}}]=\log\ell(\ba)
\label{eq:Cumulus}
\ee
is a real-analytic convex function on\/ $\RR^\Omega$.
\item\label{it:Limits-iv} The central limit theorem holds, i.e., as $N\to\infty$
$$
\frac1{\sqrt{N}}\left(S_N\boldsymbol{\fJ}-\EE_{R_0}\left[S_N\boldsymbol{\fJ}\right]\right)
$$
converges in law towards a centered Gaussian vector with covariance matrix
\be
C_{\omega\nu}=\ell_{\omega\nu}-\ell_\omega\ell_\nu,
\label{eq:CentralCova}
\ee
where
$$
\ell_\omega\coloneqq(\partial_{\alpha_\omega}\ell)(\bzero),\qquad
\ell_{\omega\nu}\coloneqq(\partial_{\alpha_\nu}\partial_{\alpha_\omega}\ell)(\bzero).
$$
\item\label{it:Limits-v} The steady measure $\QQ_{R_+}$ is $\phi$-mixing, in particular 
\be
\wlim_{n\to\infty}\QQ_{R_0}\circ\phi^{-n}=\QQ_{R_+},
\label{eq:return}
\ee
for any $R_0\in\fA_{\ast+1}$.
\item\label{it:Limits-vi} For any Borel set $G\subset\RR^\Omega$, the following large deviation estimates
hold:
\begin{equation}
\label{eq:ld}
\begin{split}
-\inf_{\bs\in\mathring{G}}I(\bs)
&\leq\liminf_{N\rightarrow\infty}\frac{1}{N}\log\QQ_{R_0}\left(\frac1NS_N\boldsymbol{\fJ}\in G\right)\\
&\leq\limsup_{N\rightarrow\infty}\frac{1}{N}\log\QQ_{R_0}\left(\frac1NS_N\boldsymbol{\fJ}\in G\right) 
\leq-\inf_{\bs\in\overline{G}}I(\bs),
\end{split}
\end{equation}
where $\mathring{G}$/$\overline{G}$ denote the interior/closure of $G$ and the good rate function $\bs\mapsto I(\bs)$ is given by the Legendre--Fenchel transform of the function $\ba\mapsto e(-\ba)$,
\be
I(\bs)\coloneqq\sup_{\ba\in\RR^\Omega}\left(\ba\cdot\bs-e(-\ba)\right). 
\label{eq:I-lim}
\ee
\end{enumerate}
\end{theo}

\subsubsection{Level-3: The Entropy Process}

The empirical measures of the entropy process is the family $(\mu_n)_{n\in\NN^\ast}$ of $\cP(\bX)$-valued random variables 
$$
\bX\ni\bx\mapsto\mu_n\coloneqq\frac1n\sum_{k=0}^{n-1}\delta_{\phi^k(\bx)}.
$$
We denote by
$$
\bX_\mathrm{fin}=\bigcup_{k\ge0}\fX^k
$$
the set of finite words $\bx=x_1\cdots x_n$ on the alphabet $\fX$, and let $|\boldsymbol{x}|=n$ be the length of $\boldsymbol{x}$.

\begin{theo}\label{thm:Level3}
Assume that\/ $\LL$ is irreducible with ESS $R_+$, and let $R_0\in\fA_{\ast+1}$.
\begin{enumerate}[label=(\roman*)]
\item\label{it:Level3-i} For $f\in C(\bX)$, the limit
\be
Q(f)\coloneqq\lim_{n\to\infty}\frac1n\log\EE_{R_0}\left[\e^{n\langle f,\mu_n\rangle}\right]
\label{eq:Qflim}
\ee
exists, defines a $1$-Lipschitz function on $C(\bX)$, and does not depend on the choice of $R_0$.
\item\label{it:Level3-ii} For any Borel set $M\subset\cP(\bX)$, the following large deviation estimates
hold:
\begin{equation}
\label{eq:ld3}
-\inf_{\mu\in\mathring{M}}\II(\mu)
\leq\liminf_{n\rightarrow\infty}\frac{1}{n}\log\QQ_{R_0}\left(\mu_n\in M\right)
\leq\limsup_{n\rightarrow\infty}\frac{1}{n}\log\QQ_{R_0}\left(\mu_n\in M\right) 
\leq-\inf_{\mu\in\overline{M}}\II(\mu),
\end{equation}
where $\mathring{M}$/$\overline{M}$ denote the interior/closure of $M$ and the good convex rate function $\mu\mapsto\II(\mu)$ is the restriction of the Legendre--Fenchel transform of the function $f\mapsto Q(f)$
to $\cP(\bX)$, i.e.,
\be
\II(\mu)\coloneqq\sup_{f\in C(\bX)}\left(\langle f,\mu\rangle-Q(f)\right),
\label{eq:II-lim}
\ee
and satisfies $\II(\mu)=+\infty$ for $\mu\in\cP(\bX)\setminus\cP_\phi(\bX)$.
\item\label{it:Level3-iii} For $\mu\in\cP_\phi(\bX)$, the limit
$$
\varsigma(\mu)\coloneqq\lim_{n\to\infty}\langle-\frac1n\log(\QQ_{R_+|\cX_n}),\mu\rangle
=-\lim_{n\to\infty}\frac1n\sum_{\bx\in\bX_{\mathrm{fin}}\atop|\bx|=n}\log(\QQ_{R_+}([\bx]))\mu([\bx]),
$$
exists and defines a lower-semicontinuous affine map\/ $\cP_\phi(\bX)\ni\mu\mapsto\varsigma(\mu)$ such that
$\varsigma(\QQ_{R_+})=h(\QQ_{R_+})$, the Kolmogorov--Sinai entropy of\/ $\QQ_{R_+}$ w.r.t.\;the shift $\phi$. Moreover,
$$
\II(\mu)=\varsigma(\mu)-h(\mu)
$$
holds for any $\mu\in\cP_\phi(\bX)$.
\end{enumerate}
\end{theo}

\subsection{The Fluctuation Theorem}

We are now in position to formulate the Fluctuation Relations (FRs for short) satisfied by the full statistics $(S_N\boldsymbol{\fJ})_{N\in\NN^\ast}$. The first FR dates back to 1905 and the celebrated work of Einstein on Brownian motion. The reader is referred to~\cite{RM07} for an overview of the subsequent developments pertaining to classical physical systems. We shall follow the mathematical formulation of FRs through large deviation estimates initiated by Gallavotti and Cohen in their foundational works on steady-state FRs in chaotic dynamics~\cite{GC1,GC2}. See also~\cite{jakvsic2011entropic} for a general approach to transient and steady-state FRs in the framework of classical dynamical systems. The direct quantization of a classical observable obeying a FR fails to satisfy a ``quantum FR'', see~\cite[Section~5.10]{JOPP} for concrete examples. In fact, to obtain an operationally meaningful extension of FRs to the quantum regime, one has to take into account the peculiar status of measurements in quantum mechanics. This precludes, in particular, the existence of steady-state FRs at the quantum scale. We refer the interested reader 
to~\cite{EHM,CHT} for exhaustive reviews and extensive lists of references to the physics literature. The mathematically oriented reader can also consult~\cite{RM06,DRM08,CM12,JPW} for a study of entropic FRs in the Markovian approximation of open quantum systems, \cite{DR09,jakvsic2013entropic,jakvsic2015energy,benoist2015full,benoist2019control,benoist2018heat}
for studies of two-time measurement protocols and to~\cite{benoist2018entropy,benoist2021entropy}
for extensions to repeated quantum measurements.

Fluctuation relations are deeply linked with microscopic time-reversal invariance, as is already apparent in Onsager's theory of irreversible processes~\cite{Ons1,Ons2}. Hence, we need to assume that our MRIS has some form of time-reversal invariance.

The driving Markov chain $(\pi,P)$ is said to be {\sl reversible} whenever the probability $\PP$ is invariant upon reversing the chronological order of events, that is, for any cylinder $[\omega_0,\ldots,\omega_n]$ one has
$$
\PP([\omega_0,\ldots,\omega_n])=\PP([\omega_n,\ldots,\omega_0]).
$$
It is well known, and straightforward to check, that reversibility is equivalent to the so-called {\sl Detailed Balance} condition, namely:
$$
\pi_\omega P_{\omega\nu}=\pi_\nu P_{\nu\omega}
$$
for all $\omega,\nu\in\Omega$. Note that under this condition one has $\pi P=\pi$, so that the following assumption implies~(\nameref{STAT}).

\begin{hypo}[{{\bf DB}}]\label{DB} The driving Markov chain $(\pi,P)$ satisfies detailed balance. 
\end{hypo}

Our second, complementary assumption ensures the reversibility of the interaction processes.

\begin{hypo}[{{\bf TRI}}]\label{TRI}
There are anti-unitary involutions $\theta$ and $\theta_\omega$ acting on $\cH_\cS$ and $\cH_{\cE_\omega}$, such that
$$
\theta_\omega H_{\cE_\omega}=H_{\cE_\omega}\theta_\omega,\qquad (\theta\otimes\theta_\omega)U_\omega=U_\omega^\ast(\theta\otimes\theta_\omega),
$$
for all $\omega\in\Omega$.
\end{hypo}

Whenever Assumption~(\nameref{TRI}) holds, we denote 
$\Theta:X\mapsto\theta X\theta$ the map induced on $\cA$ and on $\fA$.

\begin{theo}[Fluctuation Theorem]\label{thm:FT}
If\/ $\LL$ is irreducible, then Assumptions~{\rm (\nameref{TRI})} and~{\rm (\nameref{DB})} imply that the rate function of the large deviation principle~\eqref{eq:ld} satisfies the Fluctuation Relation
$$
I(-\bs)-I(\bs)=\bun\cdot\bs,
$$
for all $\bs\in\RR^\Omega$, where $\bun=(1,1,\ldots,1)$. This relation is associated with the following symmetry of the cumulant
generating function
\be
e(\bun-\ba)=e(\ba),
\label{eq:cumulsym}
\ee
for all $\ba\in\RR^\Omega$.
\end{theo}

Since $S_N\fJ_{\nu}$ accounts for the entropy produced by the interaction with reservoir $\cR_\nu$ during the $N$
first interactions, the total entropy produced during this period is given by
$$
\sigma_N=\sum_{\nu\in\Omega}S_N\fJ_{\nu}=S_N\bun\cdot\boldsymbol{\fJ}.
$$
Applying the contraction principle~\cite[Theorem~4.2.1]{DZ98} to the map $\boldsymbol{\fJ}\mapsto\bun\cdot\boldsymbol{\fJ}$
immediately yields the following

\begin{coro}\label{cor:epLDP}
Under the assumptions of Theorem~\ref{thm:FT} the sequence of entropy production random variable $(\sigma_N)_{N\in\NN^\ast}$
satisfies a large deviation principle
\begin{equation*}
-\inf_{s\in\mathring{S}}\bar I(s)
\leq\liminf_{N\rightarrow\infty}\frac{1}{N}
\log\QQ_{R_0}\left(\frac{\sigma_N}{N}\in S\right) 
\leq\limsup_{N\rightarrow\infty}\frac{1}{N}
\log\QQ_{R_0}\left(\frac{\sigma_N}{N}\in S\right)
\leq-\inf_{s\in\overline{S}}\bar I(s),
\end{equation*}
for any $R_0\in\fA_{\ast+1}$ and any Borel set $S\subset\RR$, with the rate function
$$
\bar I(s)\coloneqq\inf\{I(\bs)\mid \bs\in\RR^\Omega,\bun\cdot\bs=s\}
$$
satisfying the Fluctuation Relation
$$
\bar I(-s)-\bar I(s)=s.
$$
\end{coro}

The last relation implies
$$
\lim_{\delta\downarrow0}\lim_{N\to\infty}\left(
\frac{\QQ_{R_0}\left(\left|\frac{\sigma_N}{N}+s\right|<\delta\right)}{\QQ_{R_0}\left(\left|\frac{\sigma_N}{N}-s\right|<\delta\right)}\right)^{1/N}=\e^{-s},
$$
which is often written as
$$
\frac{\QQ_{R_0}\left(\frac{\sigma_N}{N}=-s\right)}{\QQ_{R_0}\left(\frac{\sigma_N}{N}=s\right)}\simeq\e^{-sN}
$$
in the physics literature. It shows that negative values of the entropy production rate are exponentially suppressed  relative to positive values, thus providing a deep refinement of the Second Law~\eqref{eq:barsigmadef}. 

\medskip\noindent{\bf Remark.} Theorems~\ref{thm:Limits} and~\ref{thm:FT} as well as Corollary~\ref{cor:epLDP} concern the full statistics of entropy fluxes into the reservoirs. Relation~\eqref{eq:wednesday} links entropy to energy and can be used to convert these results into statements on the full statistics of the energy fluxes $\boldsymbol{\fI}=(\fI_{\omega})_{\omega\in\Omega}$, where $\fI_{\omega}=-\fJ_{\omega}/\beta_\omega$. We shall leave the details of this alternative formulation to the reader.

\medskip
We provide a proof of Theorem~\ref{thm:FT} based on the spectral properties of $\LL^{[\ba]}$ in Section~\ref{ssec:FTproof}. In the remaining part of this section, we discuss some other consequences of time-reversal invariance for the full statistics process and the information theoretic interpretation of the Fluctuation Theorem. We assume both Conditions~(\nameref{TRI}) and~(\nameref{DB}).

We start with two general constructions associated to any stationary process $\QQ\in\cP_\phi(\boldsymbol{\fX})$.
\begin{itemize}
\item Setting $\xi=(s,s')\mapsto\hat\xi=(s',s)$ and $x=(\omega,\xi)\mapsto\hat x=(\omega,\hat\xi)$, 
the involution of $\bX_{\rm fin}$ defined by $\bx=x_1\cdots x_n\mapsto\hat\bx=\hat x_n\cdots\hat x_1$ implements time reversal. One easily checks that 
$$
\wh\QQ([\bx])\coloneqq\QQ([\hat\bx]),
$$
uniquely extends to probability measure $\wh\QQ\in\cP_\phi(\bX)$: the time-reversed process.
\item Let $\overline{\boldsymbol{\fX}}\coloneqq\fX^\ZZ$. There is a unique probability measure $\overline{\QQ}\in\cP_\phi(\overline{\boldsymbol{\fX}})$\footnote{We use the same symbol $\phi$ to denote the left shift on $\boldsymbol{\fX}$ and $\overline{\boldsymbol{\fX}}$.} such that
$$
\overline{\QQ}(Z)=\QQ\circ\phi^n(Z)
$$
for any finite cylinder $Z\subset\overline{\boldsymbol{\fX}}$ such that $\phi^n(Z)$ is $\cX$-measurable.
\end{itemize}

One easily checks that the time-reversed and two-sided measures $\wh\QQ$ and $\overline{\QQ}$ are related by
$$
\overline{\wh\QQ}=\overline{\QQ}\circ\vartheta,
$$
where $\vartheta$ is the time-reversal involution of $\overline{\boldsymbol{\fX}}$ defined by $\vartheta(\bx)_n=\hat x_{-n}$. The expectation functional $\overline{\EE}_{R_+}$ of the two-sided measure $\overline{\QQ}_{R_+}$ associated to $\QQ_{R_+}$ will be particularly useful to write the Green--Kubo formula, in Theorem~\ref{thm:thermo}~(ix) below, in a simple and standard form.

\begin{theo}\label{Thm:Level3FT}
Let \/ $\LL$ be irreducible, with ESS $R_+$. Under Assumptions~{\rm (\nameref{TRI})} and~{\rm (\nameref{DB})}, the following hold.
\begin{enumerate}[label=(\roman*)]
\item The two measures $\QQ_{R_+|\cX_n}$ and $\wh\QQ_{R_+|\cX_n}$ are equivalent for all $n$ and the log-likelihood ratio 
$$
\tilde\sigma_n(\bx)\coloneqq\log\frac{\d\QQ_{R_+|\cX_n}}{\d\wh\QQ_{R_+|\cX_n}}(\bx),
$$
satisfies
\be
\sup_{\bx\in\bX}|\sigma_n(\bx)-\tilde\sigma_n(\bx)|\le c,
\label{eq:physequ}
\ee
for some constant $c$.
\item For $\alpha\in\RR$, one has
$$
\lim_{n\to\infty}\frac1n\Ent_\alpha[\wh\QQ_{R_+}|_{\cX_n}|\QQ_{R_+}|_{\cX_n}]=e(\alpha\bun).
$$
\item $\overline{\ep}=0$ iff\/ $\wh\QQ_{R_+}=\QQ_{R_+}$.
\item The level-$3$ fluctuation relation
$$
\II(\wh\QQ)=\II(\QQ)+\langle \sigma_1,\QQ\rangle,
$$
holds for all $\QQ\in\cP_\phi(\boldsymbol{\fX})$.
\end{enumerate}
\end{theo}

\noindent{\bf Remark.} The estimate~\eqref{eq:physequ} expresses a strong form of equivalence between the steady state thermodynamic entropy  production $\sigma_N$ and the information theoretic quantity $\tilde\sigma_N$. Technically, these two sequences of random variables have exponentially equivalent laws, so that, by~\cite[Theorem~4.2.13]{DZ98}, Corollary~\ref{cor:epLDP} holds when $\sigma_N$ is replaced by $\tilde\sigma_N$. By continuity of the map $\QQ\mapsto\bra\sigma_1,\QQ\ket$, starting with Part~(iv) of the last theorem, the contraction principle~\cite[Theorem~4.2.1]{DZ98} provides an alternative proof of this version of Corollary~\ref{cor:epLDP}. We also observe that, by an elementary feature of Rényi's relative $\alpha$-entropy, the function
$$
\tilde e_N(\alpha)=\log\EE_{R_+}[\e^{-\alpha\tilde\sigma_N}]=\Ent_\alpha[\wh\QQ_{R_+}|_{\cX_N}|\QQ_{R_+}|_{\cX_N}],
$$
satisfies the symmetry $\tilde e_N(1-\alpha)=\tilde e_N(\alpha)$. By Part~(ii) of the previous
theorem, this symmetry is inherited by the function $\alpha\mapsto e(\alpha\bun)$, a
characteristic of the fluctuation theorem in Corollary~\ref{cor:epLDP} ($\tilde{\sigma}_N$ is called {\sl canonical entropy production} in~\cite{jakvsic2017entropic}, where it plays a central role).

\medskip
We close this section with a brief mention of {\sl hypothesis testing of time's arrow,} which deals with the empirical distinguishability of the two measures $\QQ_{R_+}$ and $\wh\QQ_{R_+}$ when $\overline{\ep}>0$. In the limit of large sample, this is asymptotically quantified by various exponents (Stein's, Chernoff's and Hoeffding's). Not surprisingly, these exponents are all related to the generating function $e(\ba)$. We shall not discuss the details of these relations here since they are identical to the one described in~\cite[Sections~2.5 and~2.9]{benoist2018entropy}.

\subsection{Linear Response}
\label{sec:linearresponse}
In this section, we show that, as explained in~\cite{Gal96}, the Fluctuation Theorem~\ref{thm:FT} can be viewed as an extension to far from equilibrium regimes of {\sl Linear Response Theory}. The latter deals with the reaction of physical systems to a small deviation from a thermal equilibrium situation, characterized by a well-defined temperature and vanishing entropy production.
Since RIS in general, and MRIS in particular, evolve according to a time-dependent Hamiltonian, they are not expected to reach thermal equilibrium, even when all reservoirs are in thermal equilibrium {\sl at the same temperature.} In this section we study a class of MRIS admitting a steady state qualifying as a thermal equilibrium in the above sense. Moreover, we shall see that the usual linear response properties near thermal equilibrium, {\sl Green--Kubo Formula, Onsager Reciprocity Relations} and {\sl Fluctuation--Dissipation Relations,} hold in these cases, and can be derived from Theorem~\ref{thm:FT}.

\medskip
Our starting point is a MRIS satisfying the following {\sl equilibrium conditions}

\begin{hypo}[{{\bf EQU}}]\label{EQU}

\hspace{0em}%
\begin{itemize}
	\item (\nameref{KMS}) is satisfied with all reservoirs at the same temperature: $\bb=\bar\beta\bun$.
	\item $\LL$ is irreducible, with the unique ESS $R_+$.
	\item Entropy production vanishes: $\ds\bra R_+,J_S\ket=\bar\beta\sum_{\omega\in\Omega}\bra R_+, J_\omega\ket=0$.
\end{itemize}
\end{hypo}

We consider the family of MRIS obtained from the previous one by modifying the reservoir temperatures, {\sl i.e.},
by replacing the ``equilibrium'' probe states by the family
$$
\rho_{\cE_\omega,\z}\coloneqq\e^{-(\bar\beta-\zeta_\omega)(H_{\cE_\omega}-F_{\omega,\z})},\qquad
F_{\omega,\z}\coloneqq-\frac1{\bar\beta-\zeta_\omega}{\log\tr\e^{-(\bar\beta-\zeta_\omega)H_{\cE_\omega}}},
$$
parametrized by $\z\in\RR^\Omega$. To express the dependence of a quantity $A$ w.r.t.\;the parameter $\z$, we will write $A_\z$. In particular, $A_{\bzero}$ denotes an equilibrium quantity.

As the name suggest, linear response theory aims at expressing the changes of various properties of the perturbed ESS $R_{+\z}$ to first order in the perturbation parameter $\z$. Of particular interest are the so-called {\sl kinetic coefficients} 
$$
L_{\omega\nu}\coloneqq\partial_{\zeta_\nu}\bar J_{\omega\z}|_{\z=0},
$$
where (recall~\eqref{eq:asympenergy}~\eqref{eq:greenwater})
$$
\bar J_{\omega\z}=\langle R_{+\z},J_{\omega\z}\rangle
$$
is the steady-state ensemble average of the energy flux into reservoir $\cR_\omega$.

\begin{theo}\label{thm:thermo}
Under Assumption~{\rm (\nameref{EQU})}, the following statements hold for all $\z\in\RR^\Omega$:
\hspace{0em}%
\begin{enumerate}[label=(\roman*)]
\item\label{it:thermo-i} $\LL_\z$ is irreducible. It is primitive whenever\/ $\LL$ is.
\item\label{it:thermo-ii} For\/ $\PP$-almost every $\bomega\in\bO$ one has
$$
\sum_{\nu\in\Omega}\bar J_{\nu\z}(\bomega)=0.
$$
\item\label{it:thermo-iii} With $\bb^{-1}=((\bar\beta-\zeta_\omega)^{-1})_{\omega\in\Omega}$, the limit
$$
\lim_{N\to\infty}\frac1N\bb^{-1}\cdot S_N\boldsymbol{\fJ}_{\z}=0
$$
holds\/ $\QQ_{R_0\z}$-a.s.
\item\label{it:thermo-iv} The Gaussian measure obtained in Theorem~\ref{thm:Limits}~(iii) as the limiting law of
$$
\frac1{\sqrt{N}}\left(S_N\boldsymbol{\fJ}_{\z}-\EE_{R_{+\z}}\left[S_N\boldsymbol{\fJ}_{\z}\right]\right)
$$
as $N\to\infty$ is supported by the hyperplane
$$
\mathfrak{Z}_\z\coloneqq\{\vs\in\RR^\Omega\mid\bb^{-1}\cdot\vs=0\}.
$$ 
\item\label{it:thermo-v} The cumulant generating function~\eqref{eq:Cumulus} satisfies the translation symmetry
\begin{equation}\label{eq:ch3tsymmetry}
e_\z(\ba+\gamma\bb^{-1})=e_\z(\ba)
\end{equation}
for all $\ba\in\RR^\Omega$ and $\gamma\in\RR$.
\item\label{it:thermo-vi} The rate function of the LDP~\eqref{eq:ld} satisfies
$$
I_\z(\bs)=+\infty
$$
whenever $\bs\not\in\mathfrak{Z}_\z$.
\end{enumerate}		

\medskip\noindent
Assuming also that the time-reversal symmetries~{\rm (\nameref{TRI})} and~{\rm (\nameref{DB})} hold, one further has:

\medskip
\begin{enumerate}[label=(\roman*)]
\setcounter{enumi}{6}
\item\label{it:thermo-vii}
\be
L_{\omega\nu}
=\frac1{2\bar\beta^2}(\partial_{\alpha_\nu}\partial_{\alpha_\omega}e_\bzero)(\bzero).
\label{eq:GCLR}
\ee
\item\label{it:thermo-viii} The Onsager Reciprocity Relations 
$$
L_{\omega\nu} = L_{\nu\omega},
$$
hold.
In addition, the kinetic coefficients are related to the covariance matrix of the limiting Gaussian
measure in Theorem~\ref{thm:Limits}~(iii) by the Fluctuation--Dissipation Relations
$$
L_{\omega\nu}=\frac1{2\bar\beta^2} C_{\omega\nu}.
$$ 
\item\label{it:thermo-ix} The Green--Kubo formula
\be
L_{\omega\nu}=\lim_{\varepsilon\downarrow0}
\frac1{\bar\beta^2}\sum_{n\in\ZZ}\e^{-|n|\varepsilon}\overline{\EE}_{R_+}[ \fJ_{\omega}\circ\phi^n\fJ_{\nu}]
\label{eq:gkmar}
\ee
holds, all the quantities on the right-hand side of the equality~\eqref{eq:gkmar} being evaluated at equilibrium $\z=0$.
\end{enumerate}		
\end{theo}

\section{Proofs}\label{sec:proofs}

\subsection{Proof of Proposition~\ref{prop:ch3ne}}
\label{sect:proofch3ne}

Since $\LL$ is irreducible, it has a unique ESS $R_+$ and the transition matrix $P$ has the unique and faithful invariant probability $\pi_+$. We set $Q_{\nu\omega}=\pi_{+\omega}P_{\omega\nu}\pi_{+\nu}^{-1}$, observing that $Q$ is again a right stochastic matrix with unique and faithful invariant probability $\pi_+$.
Define the following maps on $\fA_\ast$,
\be
\begin{array}{rclcrcl}
(\cL R)(\omega)&\coloneqq&\cL_\omega R(\omega),&\quad&(\Pi R)(\omega)&\coloneqq&\pi_{+\omega} R(\omega),\\[6pt]
(\cP R)(\omega)&\coloneqq&\ds\sum_{\nu\in\Omega}P_{\nu\omega}R(\nu),&\quad&
(\cQ R)(\omega)&\coloneqq&\sum_{\nu\in\Omega}Q_{\nu\omega}R(\nu),
\end{array}
\label{eq:Operas}
\ee
so that the factorizations $\LL=\cP\cL$ and $\cQ=\Pi\cP^\ast\Pi^{-1}$ hold.

\medskip\noindent$(\Leftarrow)$ If $P_{\nu\omega}>0$ implies~\eqref{eq:lespecheurs}, then it follows from~\eqref{eq:rho+magic} and~\eqref{eq:lesperles} that the time-averaged entropy production vanishes $\PP$-almost surely
\begin{align*}
\overline{\ep}(\bomega)=\bra R_+,J_S\ket
&=\sum_{\omega\in\Omega}\tr\left((R_+(\omega)\otimes\rho_{\cE_\omega})U_\omega^\ast[S_{\cE_\omega},U_\omega]\right)\\
&=\sum_{\omega,\nu\in\Omega}\pi_{+\nu}P_{\nu\omega}\tr\left((U_\omega(\rho_{+\nu}\otimes\rho_{\cE_\omega})U_\omega^\ast-\rho_{+\nu}\otimes\rho_{\cE_\omega})S_{\cE_\omega}\right)\\
&=\sum_{\omega,\nu\in\Omega}\pi_{+\nu}P_{\nu\omega}\tr\left((\rho_{+\omega}-\rho_{+\nu})\otimes\rho_{\cE_\omega}S_{\cE_\omega}\right)=0.
\end{align*}

\medskip\noindent$(\Rightarrow)$ Assuming now that the time-averaged entropy production vanishes $\PP$-almost surely, we get, using~\eqref{eq:lesperles}, \eqref{eq:R+Form} and the entropy balance~\eqref{eq:balalaika2},
\be
\begin{split}
0&=\overline{\ep}(\bomega)
=\bra R_+,J_S\ket=\EE_+[\bra\rho_{+\omega_0},J_S(\omega_1)\ket]\\[4pt]
&=\EE_+[\Ent(U_{\omega_1}(\rho_{+\omega_0}\otimes\rho_{\cE_{\omega_1}})U_{\omega_1}^\ast|\cL_{\omega_1}\rho_{+\omega_0}\otimes\rho_{\cE_{\omega_1}})+S(\rho_{+\omega_0})-S(\cL_{\omega_1}\rho_{+\omega_0})].
\end{split}
\label{eq:Betty}
\ee
Recalling the definition of the right stochastic matrix $Q$, we can write
$$
\EE_+[S(\cL_{\omega_1}\rho_{+\omega_0})]=\sum_{\nu,\omega\in\Omega}\pi_{+\nu}P_{\nu\omega}S(\cL_\omega\rho_{+\nu})
=\sum_{\omega\in\Omega}\pi_{+\omega}\sum_{\nu\in\Omega}Q_{\omega\nu}S(\cL_\omega\rho_{+\nu}).
$$
Invoking the concavity of the entropy map $\rho\mapsto S(\rho)$~\cite[Theorem~3.7]{Pe08}, we derive
$$
\EE_+[S(\cL_{\omega_1}\rho_{+\omega_0})]
\le\sum_{\omega\in\Omega}\pi_{+\omega}S\left(\sum_{\nu\in\Omega}Q_{\omega\nu}\cL_\omega\rho_{+\nu}\right)
=\sum_{\omega\in\Omega}\pi_{+\omega}S((\cL\cQ^\ast\rho_+)(\omega)).
$$
By~\eqref{eq:R+magic}, we have
\be
\cL\cQ^\ast\rho_+=\cL\cQ^\ast\Pi^{-1}\cL R_+=\cL\Pi^{-1}\cP\cL R_+=\Pi^{-1}\cL\LL R_+=\Pi^{-1}\cL R_+=\rho_+,
\label{eq:ventilo}
\ee
so that
$$
\EE_+[S(\cL_{\omega_1}\rho_{+\omega_0})]\le \EE_+[S(\rho_{+\omega_0})].
$$
It follows that $\EE_+[S(\rho_{+\omega_0})-S(\cL_{\omega_1}\rho_{+\omega_0})]\ge0$. This inequality, together with the non-negativity of relative entropy, allow us to deduce from~\eqref{eq:Betty} that
\be
\EE_+[\Ent(U_{\omega_1}(\rho_{+\omega_0}\otimes\rho_{\cE_{\omega_1}})U_{\omega_1}^\ast|\cL_{\omega_1}\rho_{+\omega_0}\otimes\rho_{\cE_{\omega_1}})]=0
=\EE_+[S(\rho_{+\omega_0})-S(\cL_{\omega_1}\rho_{+\omega_0})].
\label{eq:Sanderling}
\ee
Writing the second identity as
$$
\sum_{\omega\in\Omega}\pi_{+\omega}\left(
S\left(\sum_{\nu\in\Omega}Q_{\omega\nu}\cL_\omega\rho_{+\nu}\right)
-\sum_{\nu\in\Omega}Q_{\omega\nu}S(\cL_\omega\rho_{+\nu})\right)=0,
$$
and	combining the facts that von Neumann's entropy is strictly convex~\cite[Theorem~2.10]{Carl10} and
$\pi_+$ faithful, we derive that for all $\nu,\omega\in\Omega$ such that $Q_{\omega\nu}>0$, one has
$$
\cL_\omega\rho_{+\nu}=\sum_{\mu\in\Omega}Q_{\omega\mu}\cL_\omega\rho_{+\mu}=\rho_{+\omega},
$$
where we used~\eqref{eq:ventilo} to justify the last identity.

Since $Q_{\omega\nu}>0$ iff $P_{\nu\omega}>0$, the first equality in~\eqref{eq:Sanderling} further yields that
$$
U_{\omega}(\rho_{+\nu}\otimes\rho_{\cE_{\omega}})U_{\omega}^\ast
=\cL_{\omega}\rho_{+\nu}\otimes\rho_{\cE_{\omega}}
=\rho_{+\omega}\otimes\rho_{\cE_{\omega}},
$$
provided $P_{\nu\omega}>0$. Thus, the family of states $(\rho_{+\omega})_{\omega\in\Omega}$ has the required property. 

The first equality in~\eqref{eq:Geneva} follows from the fact that, whenever $\rho_k(\bomega)=\rho_{+\omega_k}$ for $k\in\{n,n+1\}$, then the previous identity implies that $\ep_n(\bomega)=0$. Assuming~(\nameref{KMS}), the simple calculation
\begin{align*}
\bra\rho_{+\nu},J(\omega)\ket&=\tr\left((\rho_{+\nu}\otimes\rho_{\cE_\omega})(H_{\cE_\omega}-U_\omega^\ast H_{\cE_\omega} U_\omega)\right)\\
&=\tr\left((\rho_{+\nu}-\rho_{+\omega})\otimes \rho_{\cE_\omega}H_{\cE_\omega}\right)=0,
\end{align*}
shows that the
same identity leads to the second equality in~\eqref{eq:Geneva}.
Finally, \eqref{eq:greenwater} immediately leads to~\eqref{eq:stillraining}.

\medskip\noindent{\bf Remark.} If $\LL$ is irreducible and the family of states $(\varkappa_\omega)_{\omega\in\Omega}$ satisfies
$$
U_\omega(\varkappa_{\nu}\otimes\rho_{\cE_\omega})U_\omega^\ast=\varkappa_{\omega}\otimes\rho_{\cE_\omega}
$$
for all pairs $(\nu,\omega)\in\Omega\times\Omega$ such that $P_{\nu\omega}>0$, then $\varkappa_{\omega}=\rho_{+\omega}$
for all $\omega\in\Omega$.

To see this, we first note that
$$
P_{\nu\omega}\cL_\omega\pi_{+\nu}\varkappa_{\nu}=P_{\nu\omega}\pi_{+\nu}\varkappa_{\omega},
$$
for any pair $(\nu,\omega)\in\Omega\times\Omega$. Summing both sides of this identity over $\nu$ yields
\be
\cL\cP\Pi\varkappa=\Pi\varkappa,
\label{eq:tomato}
\ee
from which we deduce
$$
\LL(\cP\Pi\varkappa)=\cP(\cL\cP\Pi\varkappa)=\cP\Pi\varkappa,
$$
and since $\LL$ is irreducible, we can conclude that its unique fixed point in $\fA_{\ast+1}$ is $R_+=\cP\Pi\varkappa$.
The second relation in~\eqref{eq:R+magic} further gives
$$
\rho_+=\cL\Pi^{-1}R_+=\cL\Pi^{-1}\cP\Pi\varkappa
$$
which, combined with~\eqref{eq:tomato} yields
$$
\varkappa=\Pi^{-1}\cL\cP\Pi\varkappa=\cL\Pi^{-1}\cP\Pi\varkappa=\rho_+.
$$

\subsection{Proofs of Theorem~\ref{thm:Limits} and~\ref{thm:Level3}}
\label{sec:prooflimits}

We need some preparations. The next proposition summarizes the spectral properties of irreducible CP maps on $\fA_\ast$ and their adjoints. Recall that, denoting $r$ the spectral radius of such a map, its peripheral spectrum is that part of its spectrum lying on the circle of radius $r$. Following the definition in usage in spectral theory (see, e.g., \cite[Section~I.4]{Kato}), we say that an eigenvalue is simple when
its algebraic multiplicity is $1$.\footnote{Note that this condition is stronger than the one used in many studies of the subject where eigenvalues are called simple when their geometric multiplicity is $1$ (as, e.g., in~\cite{EHK,Schrader}).}

\begin{prop}\label{prp:PerronFrobenius}
Let $\Phi$ be an irreducible CP map on $\fA_\ast$ with spectral radius $r>0$.
\begin{enumerate}[label=(\roman*)]
\item $r$ is an eigenvalue of\/ $\Phi$, with a unique faithful eigenvector $R_0\in\fA_{\ast+1}$. Moreover, to this eigenvalue, the adjoint map $\Phi^\ast$ has a unique eigenvector $M_0>0$, such that $\langle R_0,M_0\rangle=1$.
\item There is $p\in\NN^\ast$, the period of\/ $\Phi$, and a primitive $p^\mathrm{th}$ root of unity\/ $\zeta$, such that the peripheral spectrum of\/ $\Phi$ is given by $\{r\zeta^k\mid 0\le k<p\}$, each peripheral eigenvalue being simple.
\end{enumerate}
\end{prop}
\begin{proof} These results on CP maps are scattered through the literature, see, e.g., \cite{EHK,groh,Schrader,FP09}. For the reader's convenience we give some more precise indications.
	
\medskip\noindent(i) is well known, see, e.g., \cite[Section~2]{EHK} or~\cite[Theorem~3.1]{Schrader}. 
	
\medskip\noindent(ii) Following~\cite{EHK} and setting
$$
\Psi(X)=\frac1r M_0^{-1/2}\Phi^\ast(M_0^{1/2}XM_0^{1/2})M_0^{-1/2},
$$
defines a CPU map on $\fA$, such that $\spec(\Phi)=r\,\overline{\spec(\Psi)}$, including algebraic and geometric multiplicities. Up to the simplicity of the peripheral eigenvalues, the statement follows from~\cite[Theorem~4.2]{EHK}. From the latter, we infer the existence of a unitary $U\in\fA$
such that $\Psi(U^k)=\zeta^kU^k$ for $k\in\{0,\ldots,p-1\}$. Moreover, by~\cite[Lemma~4.1]{EHK},
$\Psi(U^kX)=\zeta^kU^k\Psi(X)$ holds for any $X\in\fA$ and $k\in\ZZ$. Denoting by $R=M_0^{1/2}R_0M_0^{1/2}$ the eigenvector of $\Psi^\ast$ to its eigenvalue $1$, normalized by $\bra R,\un\ket=1$, we can write
\begin{align*}
\bra\zeta^{k}U^{-k}R,X\ket&=\zeta^{-k}\bra\Psi^\ast(R),U^kX\ket=\zeta^{-k}\bra R,\Psi(U^kX)\ket
=\bra R,U^k\Psi(X)\ket=\bra\Psi^\ast(U^{-k}R),X\ket,
\end{align*}
for all $X\in\fA$, from which we deduce that $\Psi^\ast(U^{-k}R)=\zeta^kU^{-k}R$. To show that peripheral eigenvalues of $\Psi$ are simple, we follow the argument of the proof of~\cite[Lemma~5.3]{CP15}.
Suppose that $\zeta^k$, as an eigenvalue of $\Psi$, is not simple. The Jordan normal form of $\Psi$ implies that, for some $Y\in\fA$, $\Psi(U^kY)=\zeta^kU^kY+U^k$. It follows that
\begin{align*}
\zeta^k\bra R,Y\ket=\zeta^k\bra U^kR,U^kY\ket&=\bra\Psi^\ast(U^kR),U^kY\ket\\
&=\bra U^kR,\Psi(U^kY)\ket=\bra U^kR,\zeta^kU^kY+U^k\ket=\zeta^k\bra R,Y\ket+1,
\end{align*}
which is absurd.
\end{proof}
	
\begin{lemm}\label{lem:Lalpha}
$\LL^{[\ba]}$ is irreducible/primitive for all $\ba\in\RR^\Omega$ iff\/ $\LL$ is. 
\end{lemm}
\begin{proof}
For $\omega\in\Omega$ denote by $(\varphi_{\omega,s})_{s\in\Sigma_\omega}$ an orthonormal eigenbasis of $S_{\cE_\omega}$, with $S_{\cE_\omega}\varphi_{\omega,s}=\varsigma_{\omega,s}\varphi_{\omega,s}$. One easily derives the following Kraus representation of the TPCP-map $\cL_\omega$
$$
\cL_\omega\rho=\sum_{s,s'\in\Sigma_\omega}V_{\omega,s,s'}\rho V_{\omega,s,s'}^\ast,
$$
where $V_{\omega,s,s'}\in\cB(\cH_\cS)$ is the operator associated to the sesquilinear form
$$
\bra\chi,V_{\omega,s,s'}\psi\ket\coloneqq\e^{-\varsigma_{\omega,s}/2}\bra\chi\otimes\varphi_{\omega,s'},U_\omega\psi\otimes\varphi_{\omega,s}\ket.
$$
Further setting
\be
\mathbb{V}_{\omega,s,s',\nu}\coloneqq p_\omega^\ast\sqrt{P_{\nu\omega}}V_{\nu,s,s'}p_\nu,
\label{eq:LeRegiment}
\ee
where $p_\omega:\ell^2(\Omega;\cH_\cS)\ni\psi\mapsto\psi(\omega)\in\cH_\cS$, we can write a Kraus decomposition of $\LL$ as
$$
\LL R=\sum_{\omega,\nu\in\Omega\atop s,s'\in\Sigma_\nu}\mathbb{V}_{\omega,s,s',\nu}R\,\mathbb{V}_{\omega,s,s',\nu}^\ast.
$$
Recalling~\eqref{eq:cLalphadef} and~\eqref{eq:LLalphadef}, for 
$\ba\cdot\boldsymbol{\fJ}_n=\sum_{\omega\in\Omega}\alpha_\omega\fJ_{n,\omega}$, and $X\in\fA$ one has
$$
\EE_\QQ[\e^{-\ba\cdot\boldsymbol{\fJ}_n}X(\omega_{n+1})]
=\EE\left[\tr\left(X(\omega_{n+1})\cL_{\omega_n}^{[\alpha_{\omega_n}]}\cdots\cL_{\omega_1}^{[\alpha_{\omega_1}]}\rho_{\omega_0}\right)\right]=\bra\LL^{[\ba]n}R,X\ket
$$
where $\LL^{[\ba]}$ has the representation 
$$
\LL^{[\ba]}R=\sum_{\omega,\nu\in\Omega\atop s,s'\in\Sigma_\nu}\e^{-\alpha_\nu(\varsigma_{\nu,s'}-\varsigma_{\nu,s})}
\mathbb{V}_{\omega,s,s',\nu}R\mathbb{V}_{\omega,s,s',\nu}^\ast.
$$
Note that this representation extends $\LL^{[\ba]}$ to a CP map on $\cB^1(\cH_\cS)$ whose range is in $\fA_\ast$.
In particular, the left/right eigenvectors of $\LL^{[\ba]}$ to non-zero eigenvalues are in $\fA$/$\fA_\ast$.
Comparing the Kraus families of $\LL$ and $\LL^{[\ba]}$ yields the claim.
\end{proof}

\begin{lemm}\label{lem:Stings}
Let\/ $\KK$ be a CPTP map on $\fA_\ast$.
\begin{enumerate}[label=(\roman*)]
\item If\/ $\KK$ is positivity improving and $R\in\fA_{\ast+}$, then there exists a constant $\underline{\ell}_R>0$ such that
$$
\KK S\ge\underline{\ell}_R\langle S,\un\rangle R,
$$
holds for all $S\in\fA_{\ast+}$.
\item If\/ $\KK R_+=R_+$ for some faithful $R_+\in\fA_{\ast+1}$, then there exists a constant $\bar\ell_{R_+}$, which does not depend on\/ $\KK$, and such that
$$
\KK S\le\bar\ell_{R_+}\langle S,\un\rangle R_+,
$$
holds for all $S\in\fA_{\ast+}$.
\end{enumerate}
\end{lemm}
\begin{proof}
(i) It suffices to prove the claim for $R\neq0$ and $S\in\cR=\{S\in\fA_{\ast+}\mid 1\le\Tr(S)\le\Tr(\un)\}$. Setting
$$
\ell(S)=\min\spec(S),
$$
one has $\ell(\KK S)>0$, and since $\cR$ is compact,
$$
\underline{\ell}=\inf_{S\in\cR}\ell(\KK S)=\min_{S\in\cR}\ell(\KK S)>0.
$$
Invoking the spectral representation of $S$, we can write
$$
S=\sum_{s\in\spec(S)}s P_s\ge\bar sP_{\bar s},
$$
with $\bar s=\max\spec(S)>0$. It follows that, with $\bar r=\max\spec(R)$,
$$
\KK S\ge\bar s\,\KK P_{\bar s}\ge\bar s\,\underline{\ell}\un
\ge\frac{\Tr(S)}{\Tr(\un)}\underline{\ell}\un
\ge\underline{\ell}\frac{\Tr(S)}{\Tr(\un)}\frac{R}{\bar r}
=\underline{\ell}_R\langle S,\un\rangle R.
$$

\medskip\noindent(ii) For $S\in\fA_{\ast+}$, with $\underline{r}_+=\min\spec(R_+)>0$, one has
$$
S=\sum_{s\in\spec(S)}sP_s\le\sum_{s\in\spec(S)}s\,\un\le\Tr(S)\un
=\bra S,\un\ket\un\le\frac1{\underline{r}_+}\bra S,\un\ket R_+,
$$
and hence
$$
\KK S\le\frac1{\underline{r}_+}\bra S,\un\ket\KK R_+=\bar\ell_{R_+}\bra S,\un\ket R_+.
$$
\end{proof}

For $\PP,\QQ\in\cP(\bX)$ we write $\PP\preccurlyeq\QQ$ whenever there exists a constant $1\le C<\infty$ such that $\PP(A)\le C\QQ(A)$ for all $A\in\cX$. Note that this implies $\langle f,\PP\rangle\le C\langle f,\QQ\rangle$ for all non-negative $f\in C(\bX)$ as well as $\PP\ll\QQ$.

\begin{lemm}\label{lem:QQorder}
If\/ $\LL$ is irreducible with ESS $R_+$, then the following hold. 
\begin{enumerate}[label=(\roman*)]
\item $\QQ_{R_0}\preccurlyeq\QQ_{R_+}$ for any $R_0\in\fA_{\ast+1}$.
\item If $R_0$ is faithful, then $\QQ_{R_+}\preccurlyeq\QQ_{R_0}$, and in particular $\QQ_{R_0}$ and $\QQ_{R_+}$ are equivalent.
\end{enumerate}
\end{lemm}
\begin{proof}
\noindent(i) Since $R_+$ is faithful, one has
$$
R=R_+^{1/2}(R_+^{-1/2}RR_+^{-1/2})R_+^{1/2}
\le\|R^{1/2}R_+^{-1/2}\|^2R_+=CR_+
$$
from which we deduce that
$$
P_{\omega_1\omega_2}\cdots P_{\omega_{n-1}\omega_n}
\tr(\cL_{\omega_n,\xi_n}\cdots\cL_{\omega_1,\xi_1}R_0(\omega_1))
\le CP_{\omega_1\omega_2}\cdots P_{\omega_{n-1}\omega_n}
\tr(\cL_{\omega_n,\xi_n}\cdots\cL_{\omega_1,\xi_1}R_+(\omega_1)),
$$
and hence $\QQ_{R_0}(A)\le C\QQ_{R_+}(A)$ for all $A=[x_1\cdots x_n]$ in the
semi-algebra of finite cylinders. Since elements of the algebra $\cX_0$ generated by finite cylinders are finite disjoint unions of such cylinders~\cite[Theorem~0.1]{Wa}, this inequality extends by additivity to any $A\in\cX_0$. Define the probability $\bar\QQ=\frac12(\QQ_{R_0}+\QQ_{R_+})$ and let $A\in\cX$, the $\sigma$-algebra generated by $\cX_0$. By~\cite[Theorem~0.7]{Wa}, for any $\epsilon>0$ there exists $B\in\cX_0$ such that $\bar\QQ(A\Delta B)<\epsilon/2$. It follows that $\QQ_{R_0}(A\Delta B)<\epsilon$ and $\QQ_{R_+}(A\Delta B)<\epsilon$ and hence
$$
|\QQ_{R_0}(A)-\QQ_{R_0}(B)|<\epsilon,\qquad|\QQ_{R_+}(A)-\QQ_{R_+}(B)|<\epsilon.
$$
We conclude that
\begin{align*}
\QQ_{R_0}(A)&=\QQ_{R_0}(B)+(\QQ_{R_0}(A)-\QQ_{R_0}(B))\\
&\le C\QQ_{R_+}(B)+\epsilon\\
&\le C\QQ_{R_+}(A)+C(\QQ_{R_+}(B)-\QQ_{R_+}(A))+\epsilon\\
&\le C\QQ_{R_+}(A)+(C+1)\epsilon,
\end{align*}
which shows that $\QQ_{R_0}(A)\le C\QQ_{R_+}(A)$.

The same argument provides a proof of~(ii).
\end{proof}

We write the concatenation of two finite words $\boldsymbol{x},\boldsymbol{y}\in\bX_{\mathrm{fin}}$ as $\boldsymbol{xy}$. The following Lemma establishes some decoupling properties of the
full statistics $\QQ_{R_0}$ which are central to~\cite{CJPS} and will allow us to invoke their main results. 

\begin{lemm}\label{lem:decoupling}
Assume that\/ $\LL$ is irreducible, with ESS $R_+$. The following hold for any $R_0\in\fA_{\ast+1}$.
\begin{enumerate}[label=(\roman*)]
\item The Selective Lower Decoupling property:
There are constants $\tau\in\NN$ and $C>0$ such that, for any $\boldsymbol{x},\boldsymbol{y}\in\bX_\mathrm{fin}$
there is $\boldsymbol{z}\in\bX_\mathrm{fin}$, with $|\boldsymbol{z}|\le\tau$, and
\be
\QQ_{R_0}([\boldsymbol{xzy}])\ge C\QQ_{R_0}([\boldsymbol{x}])\QQ_{R_+}([\boldsymbol{y}]).
\label{eq:omicronsucks}
\ee
\item The Upper Decoupling property: Their is a constant $C$ such that, for all $N\in\NN$ and all $\boldsymbol{x},\boldsymbol{y},\boldsymbol{z}\in\bX_\mathrm{fin}$ with $|\boldsymbol{z}|=N$,
\be
\QQ_{R_+}([\boldsymbol{xzy}])
\le\sum_{\boldsymbol{w}\in\bX_{\mathrm{fin}}\atop |\boldsymbol{w}|=N}\QQ_{R_+}([\boldsymbol{xwy}])
\le C\QQ_{R_+}([\boldsymbol{x}])\QQ_{R_+}([\boldsymbol{y}]).
\label{eq:deltastill}
\ee
\item There is a constant $c$ such that, for any $\boldsymbol{x}\in\bX_\mathrm{fin}$ and any $k\in\NN$, there is $\boldsymbol{z}\in\bX_\mathrm{fin}$ with $|\boldsymbol{z}|=k$ and
$$
\QQ_{R_0}([\boldsymbol{xz}])\ge\e^{-ck}\QQ_{R_0}([\boldsymbol{x}]).
$$
\end{enumerate}
\end{lemm}
\begin{proof}
\noindent(i) Since, for any $\tau\in\NN$,
$$
\sum_{\boldsymbol{z}\in\bX_{\mathrm{fin}}\atop|\boldsymbol{z}|\le\tau}\QQ_{R_0}([\boldsymbol{xzy}])
\le|\{\boldsymbol{z}\in\bX_{\mathrm{fin}}\mid |\boldsymbol{z}|\le\tau\}|
\,\max_{\boldsymbol{z}\in\bX_{\mathrm{fin}}\atop|\boldsymbol{z}|\le\tau}\QQ_{R_0}([\boldsymbol{xzy}]),
$$
it suffices to show that the sum at the left-hand side is bounded below by the right-hand side of~\eqref{eq:omicronsucks}. With $x_k=(\omega_k,\xi_k)$, $y_k=(\nu_k,\eta_k)$ and $z_k=(\mu_k,\rho_k)$, we have
\begin{align*}
\QQ_{R_0}([\boldsymbol{xzy}])=&(P_{\omega_1\omega_2}\cdots P_{\omega_{n-1}\omega_n})P_{\omega_n\mu_1}(P_{\mu_1\mu_2}\cdots P_{\mu_{l-1}\mu_l})P_{\mu_l\nu_1}(P_{\nu_1\nu_2}\cdots P_{\nu_{m-1}\nu_m})\nonumber\\
&\langle(\cL_{\nu_m,\eta_m}\cdots\cL_{\nu_1,\eta_1})(\cL_{\mu_l,\rho_l}\cdots\cL_{\mu_1,\rho_1})(\cL_{\omega_n,\xi_n}\cdots\cL_{\omega_1,\xi_1})R_0(\omega_1),\un\rangle,
\end{align*}
and using the fact that $\sum_{\xi\in\Xi_\omega}\cL_{\omega\xi}=\cL_\omega$, we get
\begin{align*}
	\sum_{\rho_1,\ldots,\rho_l}\QQ_{R_0}([\boldsymbol{xzy}])=&
	(P_{\omega_1\omega_2}\cdots P_{\omega_{n-1}\omega_n})P_{\omega_n\mu_1}(P_{\mu_1\mu_2}\cdots P_{\mu_{l-1}\mu_l})P_{\mu_l\nu_1}(P_{\nu_1\nu_2}\cdots P_{\nu_{m-1}\nu_m})\\[-8pt]
	&\langle(\cL_{\nu_m,\eta_m}\cdots\cL_{\nu_1,\eta_1})(\cL_{\mu_l}\cdots\cL_{\mu_1})(\cL_{\omega_n,\xi_n}\cdots\cL_{\omega_1,\xi_1})R_0(\omega_1),\un\rangle.
\end{align*}
Setting
$$
R(\mu|\boldsymbol{x})=(P_{\omega_n\mu}\cL_{\omega_n,\xi_n})(P_{\omega_{n-1}\omega_n}\cL_{\omega_{n-1},\xi_{n-1}})\cdots(P_{\omega_1\omega_2}\cL_{\omega_1,\xi_1})R_0(\omega_1),
$$
we can write
\begin{align*}
	\sum_{\rho_1,\ldots,\rho_l}\QQ_{R_0}([\boldsymbol{xzy}])=&
	P_{\mu_l\nu_1}(P_{\nu_1\nu_2}\cdots P_{\nu_{m-1}\nu_m})\\[-8pt]
	&\langle(\cL_{\nu_m,\eta_m}\cdots\cL_{\nu_1,\eta_1})\cL_{\mu_l}(P_{\mu_{l-1}\mu_l}\cL_{\mu_{l-1}}\cdots P_{\mu_1\mu_2}\cL_{\mu_1})R(\mu_1|\boldsymbol{x}),\un\rangle,
\end{align*}
and hence
\begin{align*}
	\sum_{z_1,\ldots,z_{l-1},\rho_l}\QQ_{R_0}([\boldsymbol{xzy}])
	&=P_{\mu_l\nu_1}(P_{\nu_1\nu_2}\cdots P_{\nu_{m-1}\nu_m})
	\langle(\cL_{\nu_m,\eta_m}\cdots\cL_{\nu_1,\eta_1})\cL_{\mu_l}(\LL^{l-1}R(\,\cdot\,|\boldsymbol{x}))(\mu_l),\un\rangle,\\[-8pt]
	&=\langle(\LL^{l-1}R(\,\cdot\,|\boldsymbol{x}))(\mu_l),
	\cL_{\mu_l}^\ast(P_{\mu_l\nu_1}\cL_{\nu_1,\eta_1}^\ast)\cdots (P_{\nu_{m-1}\nu_m}\cL_{\nu_m,\eta_m}^\ast)\un\rangle,\\
	&=\langle(\LL^{l-1}R(\,\cdot\,|\boldsymbol{x}))(\mu_l),X(\mu_l|\boldsymbol{y})\rangle,
\end{align*}
with
$$
X(\mu|\boldsymbol{y})=\cL_{\mu}^\ast(P_{\mu\nu_1}\cL_{\nu_1,\eta_1}^\ast)\cdots (P_{\nu_{m-1}\nu_m}\cL_{\nu_m,\eta_m}^\ast)\un.
$$
It follows that
\be
\sum_{z_1,\ldots,z_l}\QQ_{R_0}([\boldsymbol{xzy}])
=\sum_{\mu}\langle(\LL^{l-1}R(\,\cdot\,|\boldsymbol{x}))(\mu),X(\mu|\boldsymbol{y})\rangle
=\langle\LL^{l-1}R(\,\cdot\,|\boldsymbol{x}),X(\,\cdot\,|\boldsymbol{y})\rangle.
\label{eq:sunny30}
\ee
Since $\LL$ is irreducible, there is $\tau\in\NN^\ast$ such that $\KK=\sum_{k<\tau}\LL^k$ is positivity
improving. Hence, invoking Lemma~\ref{lem:Stings}~(i), we can write
$$
\sum_{\boldsymbol{z}\in\bX_\mathrm{fin}\atop|\boldsymbol{z}|\le\tau}\QQ_{R_0}([\boldsymbol{xzy}])
=\langle\KK R(\,\cdot\,|\boldsymbol{x}),X(\,\cdot\,|\boldsymbol{y})\rangle
\ge\underline{\ell}_{R_+}\langle R(\,\cdot\,|\boldsymbol{x}),\un\rangle\langle R_+,X(\,\cdot\,|\boldsymbol{y})\rangle
$$
with
$$
\langle R(\,\cdot\,|\boldsymbol{x}),\un\rangle=\langle(\cL_{\omega_n\xi_n})(P_{\omega_{n-1}\omega_n}\cL_{\omega_{n-1},\xi_{n-1}})\cdots(P_{\omega_1\omega_2}\cL_{\omega_1,\xi_1})R_0(\omega_1),\un\rangle=\QQ_{R_0}([\boldsymbol{x}]),
$$
and
\begin{align*}
\langle R_+,X(\,\cdot\,|\boldsymbol{y})\rangle
&=\sum_\mu P_{\nu_1\nu_2}\cdots P_{\nu_{m-1}\nu_m}\langle\cL_{\nu_m,\eta_m}\cdots\cL_{\nu_1,\eta_1}
P_{\mu\nu_1}\cL_\mu R_+(\mu),\un\rangle\\
&=P_{\nu_1\nu_2}\cdots P_{\nu_{m-1}\nu_m}\langle\cL_{\nu_m,\eta_m}\cdots\cL_{\nu_1,\eta_1}
(\LL R_+)(\nu_1),\un\rangle=\QQ_{R_+}([\boldsymbol{y}]),
\end{align*}
concluding the proof.

\medskip\noindent(ii) Note that, the first inequality in~\eqref{eq:deltastill} being obvious, it suffices to derive the second one. For $\boldsymbol{x},\boldsymbol{y}\in\bX_\mathrm{fin}$, starting with Relation~\eqref{eq:sunny30} and invoking Lemma~\ref{lem:Stings}~(ii), we get
$$
\sum_{\boldsymbol{w}\in\bX_{\mathrm{fin}}\atop|\boldsymbol{w}|=N}\QQ_{R_+}([\boldsymbol{xwy}])
=\langle\LL^{N-1}R(\,\cdot\,|\boldsymbol{x}),X(\,\cdot\,|\boldsymbol{y})\rangle
\le\bar\ell_{R_+}\bra R(\,\cdot\,|\boldsymbol{x}),\un\ket
\bra R_+,X(\,\cdot\,|\boldsymbol{y})\ket,
$$
where $R(\,\cdot\,|\boldsymbol{x})$ and $X(\,\cdot\,|\boldsymbol{y})$ are as in the proof of Part~(i) (with $R_0=R_+$). The claimed inequality follows from the last identities in the proof of Part~(i).

\medskip\noindent(iii) This is a simple variation on the first part of~\cite[Lemma~3.5]{CJPS}, including its proof. Given $\boldsymbol{x}\in\bX_\mathrm{fin}$ and $k\in\NN$, one has
$$
\QQ_{R_0}([\boldsymbol{x}])
=\sum_{\boldsymbol{z}\in\bX_\mathrm{fin}\atop|\boldsymbol{z}|=k}\QQ_{R_0}([\boldsymbol{xz}])
\le |\fX|^k
\max_{\boldsymbol{z}\in\bX_\mathrm{fin}\atop|\boldsymbol{z}|=k}\QQ_{R_0}([\boldsymbol{xz}]).
$$
Thus, there is $\boldsymbol{z}\in\bX_\mathrm{fin}$ such that $|\boldsymbol{z}|=k$ and
$$
\QQ_{R_0}([\boldsymbol{xz}])\ge|\fX|^{-k}\QQ_{R_0}([\boldsymbol{x}]).
$$
\end{proof}

\subsubsection{Proof of Theorem~\ref{thm:Limits}}

\medskip\noindent\ref{it:Limits-i} Observing that
$$
\cL_\omega^{[\alpha]}\rho
=\tr_{\cH_{\cE_\omega}}\left((\un\otimes\e^{-\alpha S_{\cE_\omega}})U_\omega(\rho\otimes\e^{-(1-\alpha)S_{\cE_\omega}})U_\omega^\ast\right),
$$
it is clear that the map $\ba\mapsto\LL^{[\ba]}$ is real analytic, and differentiation yields
\be
\partial_{\alpha_\nu}\bra\LL^{[\ba]}R,X\ket\big|_{\ba=0}=\bra\LL_\nu R,X\ket,
\label{eq:lamer}
\ee
with
$$
(\LL_\nu R)(\omega)=P_{\nu\omega}\tr_{\cH_{\cE_\nu}}\left([U_\nu,S_{\cE_\nu}](R(\nu)\otimes\rho_{\cE_\nu})U_\nu^\ast\right).
$$
In particular, invoking~\eqref{eq:springoutthere},
\be
\bra\LL_\nu R,\un\ket=\beta_\nu\bra R,J_\nu\ket,
\label{eq:pelleas}
\ee
and hence
\begin{align*}
\QQ_{R_0}[S_N\fJ_\omega]
=-\partial_{\alpha_\omega}\QQ_{R_0}[\e^{-\ba\cdot S_N\boldsymbol{\fJ}}]\big|_{\ba=0}
&=-\partial_{\alpha_\omega}\bra\LL^{[\ba]N}R_0,\un\ket\big|_{\ba=0}\\
&=\sum_{k=0}^{N-1}\bra\LL^{[\ba](N-k-1)}(-\partial_{\alpha_\omega}\LL^{[\ba]})\LL^{[\ba]k}R_0,\un\ket\big|_{\ba=0}\\
&=-\sum_{k=0}^{N-1}\beta_\omega\bra\LL^kR_0,J_\omega\ket
=-\sum_{k=0}^{N-1}\beta_\omega\EE[\bra\rho_k(\bomega),J_\omega(\omega_{k+1})\ket].
\end{align*}
The assertion thus follows from Theorems~\ref{thm:pergo} and~\ref{thm:ascv}.

\medskip\noindent\ref{it:Limits-ii}
Invoking Lemma~\ref{lem:decoupling}~(ii), it follows from~\cite[Lemma~A.2]{CJPS} that
the shift $\phi$ is totally ergodic for $\QQ_{R_+}$.\footnote{Total ergodicity means that $\QQ_{R_+}$ is $\phi^n$-ergodic for any $n\in\NN^\ast$.} In particular, we have
$$
\lim_{N\to\infty}\frac1NS_N\boldsymbol{\fJ}
=\EE_{R_+}[\boldsymbol{\fJ}],
$$ 
$\QQ_{R_+}$-a.s. From~(i) and~\eqref{eq:greenwater} we deduce 
$\EE_{R_+}[\fJ_\nu]=-\beta_\nu\EE_+[\bra\rho_{\omega_0},\bar J_\nu(\bomega)\ket]=-\beta_\nu\langle R_+,J_\nu\rangle$
and the result follows from Lemma~\ref{lem:QQorder}~(i) which implies $\QQ_{R_0}\ll\QQ_{R_+}$.

\medskip\noindent\ref{it:Limits-iii}
Applying~\cite[Theorem~2.7~(i)]{CJPS}, yields that the limit
$$
e(\ba)=\lim_{N\to\infty}\frac1N\log\EE_{R_+}[\e^{-\ba\cdot S_N\boldsymbol{\fJ}}]
$$
exists, is finite, and defines a Lipschitz function on $\RR^\Omega$. To identify the limit, we observe that
$$
\EE_{R_+}[\e^{-\ba\cdot S_N\boldsymbol{\fJ}}]=\bra\LL^{[\ba]N}R_+,\un\ket.
$$
Since $\LL$ is irreducible, so is $\LL^{[\ba]}$ for all $\ba\in\RR^\Omega$ by Lemma~\ref{lem:Lalpha}.
Thus, denoting by $\ell(\ba)$ the spectral radius of $\LL^{[\ba]}$, there is $\fA_{\ast+}\ni R_{\ba}>0$ such that $\LL^{[\ba]}R_{\ba}=\ell(\ba)R_{\ba}$,
and consequently
$$
\|R_{\ba}^{1/2}R_+^{-1/2}\|^{-2}\bra\LL^{[\ba]N}R_{\ba},\un\ket
\le\bra\LL^{[\ba]N}R_+,\un\ket
\le\|R_+^{1/2}R_{\ba}^{-1/2}\|^2\bra\LL^{[\ba]N}R_{\ba},\un\ket.
$$
It follows that $e(\ba)=\log\ell(\ba)$ for all $\ba\in\RR^\Omega$.

Next, for $R_0\in\fA_{\ast+1}$, it follows from Lemma~\ref{lem:QQorder}~(ii) that
\be
\EE_{R_0}[\e^{-\ba\cdot S_N\boldsymbol{\fJ}}]\le C_1\,\EE_{R_+}[\e^{-\ba\cdot S_N\boldsymbol{\fJ}}]
\label{eq:uppu}
\ee
for some constant $C_1>0$. Moreover, by Lemma~\ref{lem:decoupling}~(i), there is $M\in\NN$ and $C_2>0$
such that, for all $\bx\in\bX_{\rm fin}$,
$$
\sum_{\boldsymbol{z}\in\bX_\mathrm{fin},|\boldsymbol{z}|\le M}\QQ_{R_0}([\boldsymbol{zx}])
\ge C_2\QQ_{R_+}([\boldsymbol{x}]),
$$
and hence, using the fact that $S_{M+N}\boldsymbol{\fJ}=S_M\boldsymbol{\fJ}+S_N\boldsymbol{\fJ}\circ\phi^M$, and with an obvious abuse of notation,
\begin{align}
C_2\,\EE_{R_+}[\e^{-\ba\cdot S_N\boldsymbol{\fJ}}]
&=C_2\sum_{\bx\in\bX_{\rm fin}\atop |\bx|=N}
\e^{-\ba\cdot S_N\boldsymbol{\fJ}(\bx)}\QQ_{R_+}([\bx])\nonumber\\
&\le\sum_{\boldsymbol{z}\in\bX_\mathrm{fin}\atop |\boldsymbol{z}|\le M}
\sum_{\bx\in\bX_{\rm fin}\atop|\bx|=N}
\e^{-\ba\cdot S_N\boldsymbol{\fJ}(\bx)}\QQ_{R_0}([\boldsymbol{zx}])\nonumber\\
&\le\sum_{\bx\in\bX_{\rm fin}\atop|\bx|=M+N}
\e^{-\ba\cdot S_{M+N}\boldsymbol{\fJ}(\bx)}\e^{\ba\cdot S_M\boldsymbol{\fJ}(\bx)}
\QQ_{R_0}([\boldsymbol{x}])\nonumber\\
&\le\e^{M|\ba|\|\boldsymbol{\fJ}\|_\infty}\,\EE_{R_0}[\e^{-\ba\cdot S_{M+N}\boldsymbol{\fJ}}].
\label{eq:lool}
\end{align}
Combining the upper and lower bounds~\eqref{eq:uppu}--\eqref{eq:lool} yields
$$
\lim_{N\to\infty}\frac1N\log\EE_{R_0}[\e^{-\ba\cdot S_N\boldsymbol{\fJ}}]
=\lim_{N\to\infty}\frac1N\log\EE_{R_+}[\e^{-\ba\cdot S_N\boldsymbol{\fJ}}]=\log\ell(\ba).
$$
The convexity of $\RR^\Omega\ni\ba\mapsto\log\ell(\ba)$ is a consequence of Hölder's inequality.
Finally, for $\ba_0\in\RR^\Omega$, Proposition~\ref{prp:PerronFrobenius}~(ii) implies that $\ell(\ba_0)$ is a simple positive eigenvalue of\, $\LL^{[\ba_0]}$. Since $\LL^{[\ba]}$ is an entire analytic function of $\ba\in\CC^\Omega$, it follows from regular perturbation theory~\cite[Chapter~2, Theorem~1.8]{Kato} that $\ell$ extends to an analytic function in a complex neighborhood of $\ba_0$.

\medskip\noindent\ref{it:Limits-iv}
Invoking regular perturbation theory for $\ba$ in a sufficiently small complex neighborhood $\cU\subset\CC^\Omega$ of $\bzero$, we can write
$$
\bra\LL^{[\ba]N}R_0,\un\ket
=\ell(\ba)^N\sum_{k=0}^{p-1}\zeta^{kN}\bra R_0,X_k(\ba)\ket\bra R_k(\ba),\un\ket
+\bra\LL^{[\ba]N}_<R_0,\un\ket,
$$
where $\zeta$ is a primitive $p^\mathrm{th}$-root of unity, $\langle R_k(\ba),X_j(\ba)\rangle=\delta_{kj}$,
$R_+=R_0(\ba)+O(\ba)$, $\un=X_0(\ba)+O(\ba)$ and $\LL^{[\ba]}_<$ has spectral radius $r(\ba)$ such that
$r(\ba)<1-2\delta<1-\delta<|\ell(\ba)|$ for some small $\delta>0$. It follows that
\be
\frac1N\log\EE_{R_0}[\e^{-\ba\cdot S_N\boldsymbol{\fJ}}]=\log\ell(\ba)+\frac1N\log\left(1+O(\ba)+O\left(\frac{1-2\delta}{1-\delta}\right)^N\right),
\label{eq:complexbound}
\ee
for $\ba\in\cU$ and $N\in\NN^\ast$. Thus, one has
$$
\lim_{N\to\infty}\frac1N\log\EE_{R_0}[\e^{-\ba\cdot(S_N\boldsymbol{\fJ}-\EE_{R_0}[S_N\boldsymbol{\fJ}])}]=
\log\ell(\ba)-\ba\cdot(\nabla\ell)(0)
$$
and the result follows from~\cite[Proposition~1]{Bryc}.

\medskip\noindent\ref{it:Limits-v} By Lemma~\ref{lem:decoupling}~(ii) there is a constant $C$ such that, for $\boldsymbol{x},\boldsymbol{y}\in\fX_{\mathrm{fin}}$ and $n\ge|\boldsymbol{x}|$,
$$
\QQ_{R_+}([\boldsymbol{x}]\cap\phi^{-n}([\boldsymbol{y}]))
=\sum_{\boldsymbol{z}\in\bX_{\mathrm{fin}}\atop|\boldsymbol{x}|+|\boldsymbol{z}|=n}
\QQ_{R_+}([\boldsymbol{xzy}])
\le C\QQ_{R_+}([\boldsymbol{x}])\QQ_{R_+}([\boldsymbol{y}]).
$$
A simple approximation argument, similar to the one in the proof of Lemma~\ref{lem:QQorder}, gives that for any Borel sets $A,B\subset\bX$,
\be
\limsup_{n\to\infty}\QQ_{R_+}(A\cap\phi^{-n}(B))\le C\QQ_{R_+}(A)\QQ_{R_+}(B).
\label{eq:Friday}
\ee
This estimate and a simple variation of~\cite[Theorem~2.1]{Ornstein72} imply that $\QQ_{R_+}$ is $\phi$-mixing, from what~\eqref{eq:return} immediately follows since $\QQ_{R_0}\ll\QQ_{R_+}$. For the Reader's convenience, we briefly sketch Ornstein's argument.

First, let us show that $\phi$ is weak-mixing for $\QQ_{R_+}$, i.e., for any $A,B\in\cX$,
$$
\lim_{n\to\infty}\frac1n\sum_{k=0}^{n-1}\left|\QQ_{R_+}( A\cap\phi^{-k}(B))-\QQ_{R_+}(A)\QQ_{R_+}(B)\right|=0.
$$
By a well known characterization of weak-mixing~\cite[Theorem~1.26]{Wa}, this means that the
Koopman operator $U:f\mapsto f\circ\phi$ which acts unitarily on the Hilbert space $L^2(\bX,\d\QQ_{R_+})$ has no point spectrum on the orthogonal complement of the constant functions. To prove this, suppose that, on the contrary, $U\varphi=\e^{\i\theta}\varphi$ for some non-vanishing $\varphi\perp 1$. It follows that for any $n\in\NN$ we have $U^n\varphi=\e^{\i n\theta}\varphi\perp 1$. As shown in the proof of Part~(ii), $\QQ_{R_+}$ is $\phi^n$-ergodic. Thus, by another well known spectral characterization~\cite[Theorem~1.6]{Wa}, $1$ is a simple eigenvalue of $U^n$ and hence $\theta\notin2\pi\QQ$. 
As a consequence, $\tilde\phi:z\mapsto\e^{\i\theta}z$ is a uniquely ergodic map of the unit circle $\mathbb{T}$~\cite[Theorem~6.2]{Wa}, such that $\varphi\circ\phi=\tilde\phi\circ\varphi$.
Moreover, ergodicity and the fact that $U|\varphi|=|U\varphi|=|\varphi|$ imply that $|\varphi|$ is constant. We can therefore assume w.l.o.g.\;that $|\varphi|=1$, which makes $\mu=\QQ_{R_+}\circ\varphi^{-1}$ a probability measure on $\mathbb{T}$. It further follows from
$$
\bra f\circ\tilde\phi,\mu\ket=\bra f\circ\tilde\phi\circ\varphi,\QQ_{R_+}\ket
=\bra f\circ\varphi\circ\phi,\QQ_{R_+}\ket=\bra f\circ\varphi,\QQ_{R_+}\ket
=\bra f,\mu\ket,
$$
that $\mu$ is $\tilde\phi$-invariant, and by unique ergodicity, that $\mu$ is the normalized Haar measure of $\mathbb{T}$. Setting
$S_a=\{\e^{\i\alpha}\mid |\alpha|<a\}$, $A_a=\varphi^{-1}(S_a)$, and using the fact that $\theta\notin2\pi\QQ$,
one can find a sequence $k_n\to\infty$ such that $\e^{\i\theta k_n}\to1$ and hence $\mu(S_a\Delta\tilde\phi^{-k_n}(S_a))\to0$. It follows that
$$
\limsup_{n\to\infty}\mu(S_a\cap\tilde\phi^{-n}(S_a))
=\mu(S_a)=\QQ_{R_+}(A_a),
$$
and since $\mu(S_a\cap\tilde\phi^{-n}(S_a))=\QQ_{R_+}(A_a\cap\phi^{-n}(A_a))$, the estimate~\eqref{eq:Friday} yields, for sufficiently small $a>0$,
$$
\QQ_{R_+}(A_a)=\limsup_{n\to\infty}\QQ_{R_+}(A_a\cap\phi^{-n}(A_a))\le C\QQ_{R_+}(A_a)^2=C\mu(S_a)\QQ_{R_+}(A_a)=\frac{Ca}{\pi}\QQ_{R_+}(A_a)
<\QQ_{R_+}(A_a)
$$
which is absurd, and establishes weak-mixing.

Next, we prove the mixing property. We shall invoke~\cite[Theorem~1.24]{Wa}, which characterizes weak-mixing of $\phi$ w.r.t.\;$\QQ_{R_+}$ by the ergodicity of $\phi\times\phi$ w.r.t.\;$\QQ_{R_+}\times\QQ_{R_+}$. Defining the $\phi\times\phi$-invariant probability $\mu_n(A\times B)=\QQ_{R_+}(\phi^{-n}(A)\cap B)$ we have to show that
 $\QQ_{R_+}\times\QQ_{R_+}$ is the
unique weak-$\ast$ limit point of the sequence $(\mu_n)_{n\in\NN}$.
Let $\bar\mu$ be such a limit point. The estimate~\eqref{eq:Friday} yields $\bar\mu\preccurlyeq\QQ_{R_+}\times\QQ_{R_+}$ and hence
$\bar\mu\ll\QQ_{R_+}\times\QQ_{R_+}$, and since the latter is $\phi\times\phi$-ergodic, we can conclude that $\bar\mu=\QQ_{R_+}\times\QQ_{R_+}$.

\medskip\noindent\ref{it:Limits-vi}
Follows from~(iii) and the Gärtner--Ellis theorem~\cite[Theorem~2.3.6]{DZ98}.

\subsubsection{Proof of Theorem~\ref{thm:Level3}}

In the stationary case $R_0=R_+$, Lemma~\ref{lem:decoupling} allows us to invoke~\cite[Theorem~2.13]{CJPS}
which yields the existence of the limit~\eqref{eq:Qflim} in Part~\ref{it:Level3-i}, as well as Parts~\ref{it:Level3-ii}. It remains to prove the independence statement of Part~\ref{it:Level3-i}, the large deviation bounds~\eqref{eq:ld3} in the non-stationary cases, and Part~\ref{it:Level3-iii}.

\medskip\noindent\ref{it:Level3-i} Setting
$$
Q_{R_0,n}(f)\coloneqq\frac1n\log\EE_{R_0}[\e^{n\langle\mu_n,f\rangle}],
$$
we note that for $f,g\in C(\boldsymbol{\fX})$,
$$
|Q_{R_0,n}(f)-Q_{R_0,n}(g)|\le\|f-g\|_\infty.
$$
Given $f\in C(\boldsymbol{\fX})$, let $(f_m)_{m\in\NN}\subset C(\boldsymbol{\fX})$ be such that $f_m$ is $\cX_m$-measurable and $\lim_{m}f_m=f$ (e.g., $f_m(\bx)=f(\bx\boldsymbol{z})$ for $|\bx|=m$ and
some fixed $\boldsymbol{z}\in\boldsymbol{\fX}$). One has
$$
\EE_{R_0}[\e^{n\langle\mu_n,f_m\rangle}]
=\sum_{\bx\in\fX_{\mathrm{fin}}\atop|\bx|=n+m}\e^{\sum_{k=0}^{n-1}f_m(x_{1+k},\ldots,x_{m+k})}\PP_{R_0}([\bx])
$$
and proceeding as in the proof of Theorem~\ref{thm:Limits}~(iii), one shows
$$
\lim_{n\to\infty}Q_{R_0,n}(f_m)=\lim_{n\to\infty}Q_{R_+,n}(f_m).
$$
Finally, from
$$
|Q_{R_0,n}(f)-Q_{R_+,n}(f)|
\le|Q_{R_0,n}(f)-Q_{R_0,n}(f_m)|+|Q_{R_0,n}(f_m)-Q_{R_+,n}(f_m)|+|Q_{R_+,n}(f_m)-Q_{R_+,n}(f)|,
$$
we derive
$$
\limsup_{n\to\infty}|Q_{R_0,n}(f)-Q_{R_+,n}(f)|\le2\|f-f_m\|_\infty,
$$
and taking $m\to\infty$ shows that
$$
Q(f)=\lim_{n\to\infty}Q_{R_0,n}(f)
$$
is independent of $R_0$ and satisfies $|Q(f)-Q(g)|\le\|f-g\|_\infty$.

\medskip\noindent\ref{it:Level3-ii} The proof of~\cite[Theorem~2.13]{CJPS} relies on the stationarity of the
path-space measure through the first statement of~\cite[Lemma~3.5]{CJPS}. In the non-stationary cases, this
statement can be replaced by Lemma~\ref{lem:decoupling}~(iii), and the construction of the map $\psi_{n,t}$
in~\cite[Proposition~3.1]{CJPS} modified as follows. Given integers $t\ge n\ge1$, set
$$
N=2\left\lfloor\frac{t}{2(n+\tau)}\right\rfloor,\qquad t'=Nn,
$$
where $\lfloor\,\cdot\,\rfloor$ denotes the floor function and $\tau$ is the constant appearing in Lemma~\ref{lem:decoupling}~(i). Invoking the latter, to $\bx=\bx_1\cdots\bx_N\in\bX_{\rm fin}$ with $|\bx_i|=n$, we associate $\boldsymbol{z}_1\in\bX_{\rm fin}$, $|\boldsymbol{z}_1|\le\tau$, such that 
$$
\QQ_{R_0}([\bx_1\boldsymbol{z}_1\bx_2])
\ge C\QQ_{R_0}([\bx_1])\QQ_{R_+}([\bx_2]).
$$
In a similar way, there is $\boldsymbol{z}_2\in\bX_{\rm fin}$, with $|\boldsymbol{z}_2|\le\tau$, and
$$
\QQ_{R_0}([\bx_1\boldsymbol{z}_1\bx_2\boldsymbol{z}_2\bx_3])
\ge C\QQ_{R_0}([\bx_1\boldsymbol{z}_1\bx_2])\QQ_{R_+}([\bx_3])
\ge C^2\QQ_{R_0}([\bx_1])\QQ_{R_+}([\bx_2])\QQ_{R_+}([\bx_3]).
$$
Repeating this scheme yields $\boldsymbol{z}_1,\ldots\boldsymbol{z}_{N-1}\in\bX_{\rm fin}$, $|\boldsymbol{z}_i|\le\tau$, such that
$$
\QQ_{R_0}([\bx_1\boldsymbol{z}_1\bx_2\boldsymbol{z}_2\cdots\boldsymbol{z}_{N-1}\bx_N])
\ge C^{N-1}\QQ_{R_0}([\bx_1])\QQ_{R_+}([\bx_2])\cdots\QQ_{R_+}([\bx_N]).
$$
Finally, invoking Lemma~\ref{lem:decoupling}~(iii), we can find $\boldsymbol{z}_N\in\bX_{\rm fin}$
such that $|\boldsymbol{z}_1|+\cdots+|\boldsymbol{z}_N|=t-t'$ and
\begin{align*}
\QQ_{R_0}([\bx_1\boldsymbol{z}_1\bx_2\boldsymbol{z}_2\cdots\boldsymbol{z}_{N-1}\bx_N\boldsymbol{z}_N])
&\ge\e^{-c|\boldsymbol{z}_N|}
\QQ_{R_0}([\bx_1\boldsymbol{z}_1\bx_2\boldsymbol{z}_2\cdots\boldsymbol{z}_{N-1}\bx_N])\\
&\ge\e^{-c|\boldsymbol{z}_N|}C^{N-1}\QQ_{R_0}([\bx_1])\QQ_{R_+}([\bx_2])\cdots\QQ_{R_+}([\bx_N]).
\end{align*}
By Lemma~\ref{lem:QQorder}~(i), one easily concludes that~\cite[Proposition~3.1]{CJPS} holds with $\PP=\QQ_{R_0}$
and the map $\psi_{n,t}$ defined as
$$
\psi_{n,t}(\bx)=\bx_1\boldsymbol{z}_1\bx_2\boldsymbol{z}_2\cdots\bx_N\boldsymbol{z}_N.
$$
The remaining arguments leading to the large deviation estimates are the same as in~\cite{CJPS}.

\medskip\noindent\ref{it:Level3-iii} The upper decoupling property, Lemma~\ref{lem:decoupling}~(ii), allows us to invoke~\cite[Proposition~6.3]{CJPS} which yields the required statements.

\subsection{Proof of Theorem~\ref{thm:FT}}
\label{ssec:FTproof}

We will use the notations introduced in~\eqref{eq:Operas}, as well as the following ones
\be
\begin{array}{rclcrcl}
	(\Theta R)(\omega)&\coloneqq&\Theta R(\omega),&\quad&(\cL^{[\ba]}R)(\omega)&\coloneqq&\cL_\omega^{[\alpha_\omega]}R(\omega).
\end{array}
\label{eq:Operas2}
\ee
It follows from Definition~\eqref{eq:cLalphadef} and Relation~\eqref{eq:cLrevForm} that, for any $\alpha\in\RR$,
\begin{align*}
\cL_{\omega}^{[\alpha]}\Theta
&=\sum_{\xi\in\Sigma\times\Sigma}\e^{-\alpha \delta\xi}\cL_{\omega,\xi}\Theta
=\sum_{\xi\in\Sigma\times\Sigma}\e^{(1-\alpha)\delta\xi}\Theta\cL_{\omega,\wh\xi}^\ast\\
&=\sum_{\xi\in\Sigma\times\Sigma}\e^{-(1-\alpha)\delta\xi}
\Theta\cL_{\omega,\xi}^\ast=\Theta\cL_{\omega}^{[1-\alpha]\ast}.
\end{align*}
This implies $\cL^{[\ba]}\Theta=\Theta\cL^{[\bun-\ba]\ast}$, while, obviously,
$[\Theta,\cP]=0$. Factorizing $\LL^{[\ba]}=\cP\cL^{[\ba]}$, we derive from Definition~\eqref{eq:LLalphadef}
$$
\Theta\LL^{[\ba]}\Theta=\Theta\cP\cL^{[\ba]}\Theta=\Theta\cP\Theta\cL^{[\bun-\ba]\ast}=\cP\cL^{[\bun-\ba]\ast},
$$
for all $\ba\in\RR^\Omega$. 
Writing Assumption~(\nameref{DB}) as $\cP=\Pi\cP^\ast\Pi^{-1}$ and using the fact that $[\Pi,\cL^{[\ba]}]=0$ we further get
$$
\Theta\LL^{[\ba]}\Theta=\Pi\cP^\ast\Pi^{-1}\cL^{[\bun-\ba]\ast}=\Pi(\cL^{[\bun-\ba]}\cP)^\ast\Pi^{-1}.
$$
Since $\spec(\cL^{[\bun-\ba]}\cP)\setminus\{0\}=\spec(\cP\cL^{[\bun-\ba]})\setminus\{0\}$, we conclude that $\LL^{[\ba]}$ and $\LL^{[\bun-\ba]}$ have identical spectral radii, {\sl i.e.}, 
the symmetry~\eqref{eq:cumulsym} holds. Theorem~\ref{thm:FT} now follows from Theorem~\ref{thm:Limits}~(iv)-(v), in particular~\eqref{eq:I-lim} yields
$$
I(-\bs)=\sup_{\ba\in\RR^\Omega}\left(-\bs\cdot\ba-e(-\ba)\right),
$$
and the symmetry~\eqref{eq:cumulsym} gives, 
$$
I(-\bs)
=\sup_{\ba\in\RR^\Omega}\left(-\bs\cdot\ba-e(\bun+\ba)\right)
=\sup_{\ba'\in\RR^\Omega}\left(\bs\cdot(\bun+\ba')-e(-\ba')\right)
=I(\bs)+\bs\cdot\bun,
$$
with $\bun+\ba=-\ba'$.

\subsection{Proof of Theorem~\ref{Thm:Level3FT}}

(i) For all $\omega\in\Omega$, Assumption~(\nameref{TRI}) implies $[\theta_\omega, S_{\cE_\omega}]=0$, and hence
$[\theta_\omega,\Pi_s(\omega)]=0$ for any $s\in\Sigma$. By definition~\eqref{eq:twotimeentropy}, for $\rho\in\cA_\ast$, $X\in\cA$, and $\xi=(s,s')\in\Xi_\omega$, 
\begin{align*}
	\bra\cL_{\omega,\xi}\Theta\rho,X\ket
	&=\e^{-s}\tr\left((X\otimes\Pi_{s'}(\omega))U_\omega(\theta\otimes\theta_\omega)(\rho^\ast\otimes\Pi_s(\omega))(\theta\otimes\theta_\omega)U_\omega^\ast\right)\\
	&=\e^{-s}\tr\left((X\otimes\Pi_{s'}(\omega))(\theta\otimes\theta_\omega)U_\omega^\ast(\rho^\ast\otimes\Pi_s(\omega))U_\omega(\theta\otimes\theta_\omega)\right)\\
	&=\e^{-s}\overline{\tr\left((\rho^\ast\otimes\Pi_s(\omega))U_\omega(\theta\otimes\theta_\omega)(X\otimes\Pi_{s'}(\omega))(\theta\otimes\theta_\omega)U_\omega^\ast\right)}\\
	&=\e^{\delta\xi}\overline{\bra\rho,\cL_{\omega,\hat\xi}\theta X\theta\ket}
	=\e^{\delta\xi}\bra\cL_{\omega,\hat\xi}^\ast\rho,\Theta X\ket
	=\bra\e^{\delta\xi}\Theta\cL_{\omega,\hat\xi}^\ast\rho,X\ket,
\end{align*}
so that 
\be
\cL_{\omega,\xi}\Theta=\e^{\delta\xi}\Theta\cL^\ast_{\omega,\hat\xi}.
\label{eq:cLrevForm}
\ee
Assumption~(\nameref{DB}) and Relation~\eqref{eq:cLrevForm} further yield, for any $\bx\in\bX_{\rm fin}$,  with $|\bx|=n$,
\begin{align*}
	\wh\QQ_{R_+}([\bx])&=P_{\omega_n\omega_{n-1}}\cdots P_{\omega_2\omega_1}
	\langle\un,\cL_{\omega_1,\hat\xi_1}\cdots\cL_{\omega_n,\hat\xi_n}R_+(\omega_n)\rangle\\
	&=\frac{\pi_{+\omega_{n-1}}}{\pi_{+\omega_n}}P_{\omega_{n-1}\omega_n}\cdots
	\frac{\pi_{+\omega_1}}{\pi_{+\omega_2}}P_{\omega_1\omega_2}
	\langle\cL^\ast_{\omega_n,\hat\xi_n}\cdots\cL^\ast_{\omega_1,\hat\xi_1}\un,R_+(\omega_n)\rangle\\
	&=\frac{\pi_{+\omega_{1}}}{\pi_{+\omega_n}}P_{\omega_1\omega_2}\cdots P_{\omega_{n-1}\omega_n}
	\langle\e^{-\delta\xi_n}\Theta\cL_{\omega_n,\xi_n}\Theta\cdots
	\e^{-\delta\xi_1}\Theta\cL_{\omega_1,\xi_1}\Theta\un,R_+(\omega_n)\rangle\\
	&=\e^{-\sum_{k=1}^n\delta\xi_k}\frac{\pi_{+\omega_{1}}}{\pi_{+\omega_n}}
	P_{\omega_1\omega_2}\cdots P_{\omega_{n-1}\omega_n}
	\langle\Theta R_+(\omega_n),\cL_{\omega_n,\xi_n}\cdots\cL_{\omega_1,\xi_1}\un\rangle.
\end{align*}
Setting
$$
c=\log\left(\frac{\max_k\pi_{+k}}{\min_k\pi_{+k}}\frac{\max\spec R_+}{\min\spec R_+}\right),
$$
we thus get
$$
\e^{-c+\sum_{k=1}^{n}\delta\xi_k}
\le\frac{\QQ_{R_+}([\bx])}{\wh\QQ_{R_+}([\bx])}\le\e^{c+\sum_{k=1}^{n}\delta\xi_k},
$$
and since $\sigma_n=S_n\boldsymbol{1}\cdot\boldsymbol{\fJ}=\sum_{k=1}^n\delta\xi_k$, taking the logarithm yields the claims.

\medskip\noindent(ii) The relative Rényi $\alpha$-entropy of $\wh\QQ_{R_+}|_{\cX_n}$ w.r.t.\;$\QQ_{R_+}|_{\cX_n}$ is given by
$$
\Ent_\alpha[\wh\QQ_{R_+|\cX_n}|\QQ_{R_+|\cX_n}]
=\log\EE_{R_+}\left[\left(\frac{\d\wh\QQ_{R_+|\cX_n}}{\d\QQ_{R_+|\cX_n}}
\right)^\alpha\right]=\log\EE_{R_+}\left[\e^{-\alpha\tilde\sigma_n}\right],
$$
and the claim follows directly from Theorem~\ref{thm:Limits}~(iii) and the estimate~\eqref{eq:physequ}.

\medskip\noindent(iii) By~\eqref{eq:epirred} and Theorem~\ref{thm:Limits}~(i), one has, taking~\eqref{eq:physequ} into account,
$$
\overline{\ep}=\lim_{n\to\infty}\frac1n\EE_{R_+}[\sigma_n]
=\lim_{n\to\infty}\frac1n\EE_{R_+}[\tilde\sigma_n].
$$
Since
\be
\Ent(\QQ_{R_+|\cX_n}|\wh\QQ_{R_+|\cX_n})=\EE_{R_+}[\tilde\sigma_n],
\label{eq:Entform}
\ee
it immediately follows from $\wh\QQ_{R_+}=\QQ_{R_+}$ that $\overline{\ep}=0$.

Assuming now that $\overline{\ep}=0$, we observe that
$$
\frac1n\EE_{R_+}[\sigma_n]=\EE_{R_+}[\sigma_1]\to0
$$
implies $\EE_{R_+}[\sigma_n]=0$ for all $n$ and hence by~\eqref{eq:physequ} and~\eqref{eq:Entform},
$$
\Ent(\wh\QQ_{R_+|\cX_n}\times\QQ_{R_+}|\QQ_{R_+|\cX_n}\times\QQ_{R_+})
=\Ent(\wh\QQ_{R_+|\cX_n}|\QQ_{R_+|\cX_n})
=\Ent(\QQ_{R_+|\cX_n}|\wh\QQ_{R_+|\cX_n})\le c.
$$
The lower semicontinuity of relative entropy thus yields
$$
\Ent(\wh\QQ_{R_+}|\QQ_{R_+})
\le\liminf_{n\to\infty}\,\Ent(\wh\QQ_{R_+|\cX_n}\times\QQ_{R_+}|\QQ_{R_+|\cX_n}\times\QQ_{R_+})\le c
$$
which implies that $\wh\QQ_{R_+}\ll\QQ_{R_+}$, and since $\QQ_{R_+}$ is ergodic and $\wh\QQ_{R_+}\in\cP_\phi(\boldsymbol{\fX})$, we can conclude that $\wh\QQ_{R_+}=\QQ_{R_+}$.

\medskip\noindent(iv) Denote by $C_\mathrm{fin}(\boldsymbol{\fX})$ the subspace of $C(\boldsymbol{\fX})$ consisting of the functions that are $\cX_m$-measurable for some $m\in\NN$.
If $f$ is $\cX_m$-measurable, then $\bra f,\mu_n\ket=S_nf$ is $\cX_{m+n}$-measurable, and using again~\eqref{eq:physequ},
\begin{align*}
	Q(f)&=\lim_{n\to\infty}\frac1n\log\bra\e^{n\bra f,\mu_n\ket},\QQ_{R_+}\ket\\
	&=\lim_{n\to\infty}\frac1n\log\bra\e^{\tilde\sigma_{n+m}+n\bra f,\mu_n\ket},\wh\QQ_{R_+}\ket\\
	&=\lim_{n\to\infty}\frac1n\log\bra\e^{\sigma_{n+m}+n\bra f,\mu_n\ket},\wh\QQ_{R_+}\ket\\
	&=\lim_{n\to\infty}\frac1n\log\bra\e^{\sigma_{m}\circ\phi^n+n\bra f+\sigma_1,\mu_n\ket},\wh\QQ_{R_+}\ket\\
	&=\lim_{n\to\infty}\frac1n\log\bra\e^{n\bra f+\sigma_1,\mu_n\ket},\wh\QQ_{R_+}\ket=Q(\hat f-\sigma_1)
\end{align*}
where we have set $\hat f(\bx)=f(\hat\bx)$ for $\bx\in\fX^m$,  and used the facts that
$$
(S_n f)(\hat\bx)=(S_n\hat f)(\bx)
$$
for $\bx\in\fX^{m+n}$ and $\hat\sigma_1=-\sigma_1$. From the density of $C_\mathrm{fin}(\boldsymbol{\fX})$ in $C(\boldsymbol{\fX})$ and~\eqref{eq:II-lim} we deduce
\begin{align*}
	\II(\QQ)&=\sup_{f\in C_\mathrm{fin}(\bX)}\left(\langle f,\QQ\rangle-Q(f)\right)
	=\sup_{f\in C_\mathrm{fin}(\bX)}\left(\langle f,\QQ\rangle-Q(\hat f-\sigma_1)\right)\\
	&=\sup_{f\in C_\mathrm{fin}(\bX)}\left(\langle\hat f-\sigma_1,\QQ\rangle-Q(f)\right)
	=\sup_{f\in C_\mathrm{fin}(\bX)}\left(\langle f,\wh\QQ\rangle-Q(f)\right)-\langle\sigma_1,\QQ\rangle=\II(\wh\QQ)-\langle\sigma_1,\QQ\rangle.
\end{align*}

\subsection{Proof of Theorem~\ref{thm:thermo}}

\noindent\ref{it:thermo-i} One has
$$
\LL_\z R=\sum_{\omega,\nu\in\Omega\atop s,s'\in\Sigma_\nu}\mathbb{V}_{\omega,s,s',\nu,\z}R\,\mathbb{V}_{\omega,s,s',\nu,\z}^\ast
$$
where, in terms of the Kraus family~\eqref{eq:LeRegiment},
$$
\mathbb{V}_{\omega,s,s',\nu,\z}
=\e^{\zeta_\nu\varsigma_{\nu,s}/2\bar\beta-(\bar\beta-\zeta_\nu)(F_{\nu,\bzero}-F_{\nu,\z})}
\mathbb{V}_{\omega,s,s',\nu}.
$$
Thus, arguing as in the proof of Theorem~\ref{thm:Limits}, we conclude that $\LL_\z$ is irreducible/primitive
iff $\LL$ is.

\medskip\noindent\ref{it:thermo-ii} Besides~\eqref{eq:Operas} and~\eqref{eq:Operas2}, for $\gamma\in\RR$ define
$$
(\cR^\gamma R)(\omega)\coloneqq\rho_{+\omega}^{\gamma/\bar\beta} R(\omega)
$$
where $\rho_+=(\rho_{+\omega})_{\omega\in\Omega}$ is the family of states associated to the equilibrium ($\z=\bzero$) ESS $R_+$ via~\eqref{eq:R+magic}. Setting $\KK^{[\ba]}_\z\coloneqq\cL^{[\ba]}_\z\cP$, we compute
\begin{align*}
(\cR^\gamma\KK_\z^{[\ba]}\cR^{-\gamma}R)(\omega)
&=\sum_{\nu\in\Omega}\rho_{+\omega}^{\gamma/\bar\beta}\cL_{\omega\z}^{[\alpha_\omega]}P_{\nu\omega}\rho_{+\nu}^{-\gamma/\bar\beta}R(\nu)\\
&=\sum_{\nu\in\Omega}P_{\nu\omega}\tr_{\cH_{\cE_\omega}}\left(
(\rho_{+\omega}^{\gamma/\bar\beta}\otimes\rho_{\cE_\omega\z}^{\alpha_\omega})U_\omega(\rho_{+\nu}^{-\gamma/\bar\beta}R(\nu)\otimes\rho_{\cE_\omega\z}^{1-\alpha_\omega})U_\omega^\ast\right).
\end{align*}
Invoking Proposition~\ref{prop:ch3ne}, we get
\begin{align*}
U_\omega(\rho_{+\nu}^{-\gamma/\bar\beta}R(\nu)\otimes\rho_{\cE_\omega\z}^{1-\alpha_\omega})U_\omega^\ast
&=U_\omega(\rho_{+\nu}\otimes\rho_{\cE_\omega\bzero})^{-\gamma/\bar\beta}(R(\nu)\otimes\rho_{\cE_\omega\z}^{1-\alpha_\omega+\gamma/\beta_\omega})U_\omega^\ast\\
&=(\rho_{+\omega}\otimes\rho_{\cE_\omega\bzero})^{-\gamma/\bar\beta}U_\omega(R(\nu)\otimes\rho_{\cE_\omega\z}^{1-\alpha_\omega+\gamma/\beta_\omega})U_\omega^\ast,
\end{align*}
so that, with $\bb^{-1}=((\bar\beta-\zeta_\omega)^{-1})_{\omega\in\Omega}$, we derive
\begin{align*}
(\cR^\gamma\KK_\z^{[\ba]}\cR^{-\gamma}R)(\omega)
&=\sum_{\nu\in\Omega}P_{\nu\omega}\tr_{\cH_{\cE_\omega}}\left(
(\un\otimes\rho_{\cE_\omega\z}^{\alpha_\omega-\gamma/\beta_\omega})U_\omega(R(\nu)\otimes\rho_{\cE_\omega\z}^{1-\alpha_\omega+\gamma/\beta_\omega})U_\omega^\ast\right)\\
&=(\KK_\z^{[\ba-\gamma\bb^{-1}]}R)(\omega).
\end{align*}
Since $\spec(\LL^{[\ba]}_\z)\setminus\{0\}=\spec(\KK^{[\ba]}_\z)\setminus\{0\}$, we conclude that the spectral radius
of $\LL^{[\ba]}_\z$ satisfies
\be
\ell_\z(\ba-\gamma\bb^{-1})=\ell_\z(\ba),
\label{eq:zigzag}
\ee
for all $\ba\in\RR^\Omega$, $\gamma\in\RR$ and $\z\in\RR^\Omega$.

As a consequence of the translation symmetry~\eqref{eq:zigzag}, we get
$$
\bb^{-1}\cdot(\nabla\ell_\z)(0)=0.
$$
By first order perturbation theory and Relations~\eqref{eq:lamer}-\eqref{eq:pelleas}, one has
\be
(\partial_{\alpha_\omega}\ell_\z)(0)
=\partial_{\alpha_\omega}\langle\LL_\z^{[\ba]}R_{+\z},\un\rangle|_{\ba=\bzero}
=(\bar\beta-\zeta_\omega)\bra R_{+\z},J_{\omega\z}\ket,
\label{eq:zagzig}
\ee
which allows us to conclude
$$
\sum_{\nu\in\Omega}\bar J_{\nu\z}(\bomega)=\sum_{\nu\in\Omega}\bra R_{+\z},J_{\nu\z}\ket
=\bb^{-1}\cdot(\nabla\ell_\z)(0)=0.
$$

\medskip\noindent\ref{it:thermo-iii} Follows immediately from~(ii) and Theorem~\ref{thm:Limits}~(ii).

\medskip\noindent\ref{it:thermo-iv} It follows from~\eqref{eq:zigzag} that the covariance $C_\z$ of the limiting Gaussian measure, as given by Theorem~\ref{thm:Limits}~\ref{it:Limits-iii}, satisfies $\bb^{-1}\in\Ker C_\z$, so that $\Ran C_\z\subset\mathfrak{Z}_\z$.

\medskip\noindent\ref{it:thermo-v} The claim follows from Theorem~\ref{thm:Limits}~\ref{it:Limits-iii} and~\eqref{eq:zigzag}.

\medskip\noindent\ref{it:thermo-vi} By  Theorem~\ref{thm:Limits}~\ref{it:Limits-v} and the symmetry~\eqref{eq:zigzag} one has
\begin{align*}
I_\z(\bs)&=\sup_{\ba\in\RR^\Omega}\left(\ba\cdot\bs-e_\z(-\ba)\right)\\
&=\sup_{\ba\in\RR^\Omega}\left(\ba\cdot\bs-e_\z(-\ba+\gamma\bb^{-1})\right)\\
&=\sup_{\ba\in\RR^\Omega}\left((\ba+\gamma\bb^{-1})\cdot\bs-e_\z(-\ba)\right)=I_\z(\vs)+\gamma\bb^{-1}\cdot\bs
\end{align*}
for any $\gamma\in\RR$ and hence $\bb^{-1}\cdot\bs$ has to vanish whenever $I_\z(\bs)$ is finite.

\medskip\noindent\ref{it:thermo-vii} By~\eqref{eq:zagzig}, one has
$$
\bar J_{\omega\z}=(\bar\beta-\zeta_\omega)^{-1}(\partial_{\alpha_\omega}\ell_\z)(\bzero),
$$
and hence
\begin{align*}
L_{\omega\nu}&={\bar\beta}^{-1}
\partial_{\zeta_\nu}\partial_{\alpha_\omega}\ell_\z(\ba)|_{\ba=\z=\bzero}
+\delta_{\omega\nu}{\bar\beta}^{-2}\partial_{\alpha_\omega}\ell_\z(\ba)|_{\ba=\z=\bzero}\\
&={\bar\beta}^{-1}
\partial_{\zeta_\nu}\partial_{\alpha_\omega}\ell_\z(\ba)|_{\ba=\z=\bzero}
+\delta_{\omega\nu}{\bar\beta}^{-1}\bar J_{\omega\bzero}\\
&={\bar\beta}^{-1}
\partial_{\zeta_\nu}\partial_{\alpha_\omega}\ell_\z(\ba)|_{\ba=\z=\bzero},
\end{align*}
where we have invoked~\eqref{eq:stillraining}.
Since $\ell_\z(\ba)=\e^{e_\z(\ba)}$, the two symmetries~\eqref{eq:cumulsym} and~\eqref{eq:ch3tsymmetry} gives
$$
(\partial_{\alpha_\omega}\ell_\z)(\ba)
=-(\partial_{\alpha_\omega}\ell_\z)(\bun-\ba+\gamma\bb^{-1}),
$$
and hence
$$
(\partial_{\zeta_\nu}\partial_{\alpha_\omega}\ell_\z)(\ba)
=-(\partial_{\zeta_\nu}\partial_{\alpha_\omega}\ell_\z)(\bun-\ba+\gamma\bb^{-1})
-(\partial_{\alpha_\nu}\partial_{\alpha_\omega}\ell_\z)(\bun-\ba+\gamma\bb^{-1})\frac\gamma{(\bar\beta-\zeta_\nu)^2}.
$$
Setting $\ba=\bzero$, $\z=\bzero$ and $\gamma=-\bar\beta$ yields
$$
\partial_{\zeta_\nu}\partial_{\alpha_\omega}\ell_\z(\ba)|_{\ba=\z=\bzero}
=\frac1{2\bar\beta}(\partial_{\alpha_\nu}\partial_{\alpha_\omega}\ell_\bzero)(\bzero)
$$ 
from which we conclude
\be
L_{\omega\nu}
=\frac1{2\bar\beta^2}(\partial_{\alpha_\nu}\partial_{\alpha_\omega}\ell_\bzero)(\bzero).
\label{eq:Gallavotti}
\ee
Since, by Proposition~\ref{prop:ch3ne}, $(\partial_{\alpha_\nu}\ell_\bzero)(\bzero)=\bar J_{\nu\bzero}=0$, the last indentity is equivalent to~\eqref{eq:GCLR}. 

\medskip\noindent\ref{it:thermo-viii} 
The Onsager reciprocity relations are a consequence of Relation~\eqref{eq:Gallavotti} and Clairaut's theorem on equality of mixed partial derivatives.
In view of~\eqref{eq:stillraining}, the Fluctuation--Dissipation Relations follow from combining the formulas after~\eqref{eq:CentralCova} with~\eqref{eq:zagzig} and~\eqref{eq:Gallavotti}.

The remaining part of the proof will only involve equilibrium ($\z=\bzero$) quantities, and to simplify notations we will omit the subscript $\bzero$.

\medskip\noindent\ref{it:thermo-ix} By the proof of Theorem~\ref{thm:Limits}~(iv), and more precisely the
estimate~\eqref{eq:complexbound}, there is a complex neighborhood $\cU\ni\bzero$ such that the analytic
function $\cU\ni\ba\mapsto e_N(\ba)=N^{-1}\log\EE_{R_+}[\e^{-\ba\cdot S_N\boldsymbol{\fJ}}]$ satisfies
$$
\sup_{\ba\in\cU, N\in\NN^\ast}|e_N(\ba)|<\infty
$$
and $\lim_{N\to\infty}e_N(\ba)=e(\ba)$ for $\ba\in\cU\cap\RR^\Omega$. It follows from Vitali's convergence theorem (see, e.g., \cite[Theorem~B.1]{JOPP}), that
$$
(\partial_{\alpha_\nu}\partial_{\alpha_\omega}e)(\bzero)
=\lim_{N\to\infty}(\partial_{\alpha_\nu}\partial_{\alpha_\omega}e_N)(\bzero).
$$
An elementary calculation, taking into account that $\bra\fJ_\omega,\QQ_{R_+}\ket=0$, gives
\begin{align*}
(\partial_{\alpha_\nu}\partial_{\alpha_\omega}e_N)(\bzero)
&=\frac1N\EE_{R_+}[(S_N\fJ_\omega-\EE_{R_+}[S_N\fJ_\omega])(S_N\fJ_\nu-\EE_{R_+}[S_N\fJ_\nu])\\
&=\frac1N\sum_{n,m=0}^{N-1}\overline{\EE}_{R_+}[\fJ_\omega\circ\phi^n
\fJ_\nu\circ\phi^m]=\frac1N\sum_{n,m=0}^{N-1}\overline{\EE}_{R_+}[\fJ_\omega\circ\phi^{n-m}\fJ_\nu]\\
&=\frac1N\sum_{j=0}^{2(N-1)}\sum_{k=-j}^{j}\overline{\EE}_{R_+}[\fJ_\omega\circ\phi^k\fJ_\nu],
\end{align*}
and the claim follows from Frobenius summation formula~\cite[Item 87]{PSI}.

\bibliographystyle{capalpha}
\bibliography{biblio}

\newcommand{\etalchar}[1]{$^{#1}$}
\providecommand{\bysame}{\leavevmode \hbox to3em{\hrulefill}\thinspace}
\providecommand{\og}{``}
\providecommand{\fg}{''}
\providecommand{\smfandname}{and}
\providecommand{\smfedsname}{eds.}
\providecommand{\smfedname}{ed.}
\providecommand{\smfmastersthesisname}{Master Thesis}
\providecommand{\smfphdthesisname}{Thesis}
\begin{thebibliography}{WBKM00}

\bibitem[AGPS15]{AGPS15}
{\scshape Attal, S., Guillotin-Plantard, N. {\normalfont \smfandname} Sabot,
  C.}: Central limit theorems for open quantum random walks and quantum
  measurement records. Ann. H. Poincar{\'e} \textbf{16}, 15--43 (2015),
  \href{https://arxiv.org/abs/1206.1472}{<arXiv:1206.1472>}.

\bibitem[AJR21]{ajr20}
{\scshape Andr\'eys, S., Joye, A. {\normalfont \smfandname} Raqu\'epas, R.}:
  Fermionic walkers driven out of equilibrium. J. Stat. Phys. \textbf{184}, 14
  (2021), \href{https://arxiv.org/abs/2009.00604}{<arXiv:2009.00604>}.

\bibitem[AP06]{attal2006repeated}
{\scshape Attal, S. {\normalfont \smfandname} Pautrat, Y.}: From repeated to
  continuous quantum interactions. Ann. H. Poincar{\'e} \textbf{7}, 59--104
  (2006),
  \href{https://arxiv.org/abs/math-ph/0311002}{<arXiv:math-ph/0311002>}.

\bibitem[Arv69]{Arv69}
{\scshape Arveson, W.}: Subalgebras of ${C}^\ast$-algebras. Acta Math.
  \textbf{123}, 141--224 (1969),
  \href{https://doi.org/10.1007/BF02392388}{[DOI:10.1007/BF02392388]}.

\bibitem[BB20]{BoBr}
{\scshape Bougron, J.-F. {\normalfont \smfandname} Bruneau, L.}: Linear
  response theory and entropic fluctuations in repeated interaction systems. J.
  Stat. Phys. \textbf{181}, 1636--1677 (2020),
  \href{https://arxiv.org/abs/2002.10989}{<arXiv:2002.10989>}.

\bibitem[BCJP21]{benoist2021entropy}
{\scshape Benoist, T., Cuneo, N., Jak{\v{s}}i{\'c}, V. {\normalfont
  \smfandname} Pillet, C.-A.}: On entropy production of repeated quantum
  measurements {II}. {E}xamples. J. Stat. Phys. \textbf{182}, 44 (2021),
  \href{https://arxiv.org/abs/2012.03885}{<arXiv:2012.03885>}.

\bibitem[BDBP11]{Electron}
{\scshape Bruneau, L., De~Bièvre, S. {\normalfont \smfandname} Pillet, C.-A.}:
  Scattering induced current in a tight-binding band. J. Math. Phys.
  \textbf{52}, 022109 (2011),
  \href{https://arxiv.org/abs/1012.2657}{<arXiv:1012.2657>}.

\bibitem[BJM06]{BJMas}
{\scshape Bruneau, L., Joye, A. {\normalfont \smfandname} Merkli, M.}:
  Asymptotics of repeated interaction quantum systems. J. Funct. Anal.
  \textbf{239}, 310--344 (2006),
  \href{https://arxiv.org/abs/math-ph/0511026}{<arXiv:math-ph/0511026>}.

\bibitem[BJM08]{BJMrd}
\bysame : Random repeated interaction quantum systems. Commun. Math. Phys.
  \textbf{284}, 553--581 (2008),
  \href{https://arxiv.org/abs/0710.5908}{<arXiv:0710.5908>}.

\bibitem[BJM10a]{BJMmatrices}
\bysame : Infinite products of random matrices and repeated interaction
  dynamics. Ann. IHP Proba. Stat. \textbf{46}, 442--464 (2010),
  \href{https://arxiv.org/abs/math/0703675}{<arXiv:math/0703675>}.

\bibitem[BJM10b]{bruneau2010repeated}
\bysame : Repeated and continuous interactions in open quantum systems. Ann. H.
  Poincaré \textbf{10}, 1251--1284 (2010),
  \href{https://arxiv.org/abs/0905.2558}{<arXiv:0905.2558>}.

\bibitem[BJM14]{BJM}
\bysame : Repeated interactions in open quantum systems. J. Math. Phys.
  \textbf{55}, 075204 (2014),
  \href{https://arxiv.org/abs/1305.2472}{<arXiv:1305.2472>}.

\bibitem[BJP{\etalchar{+}}15]{benoist2015full}
{\scshape Benoist, T., Jak{\v{s}}i{\'c}, V., Panati, A., Pautrat, Y.
  {\normalfont \smfandname} Pillet, C.-A.}: Full statistics of energy
  conservation in two-time measurement protocols. Phys. Rev. E \textbf{92},
  032115 (2015), \href{https://arxiv.org/abs/1503.07333}{<arXiv:1503.07333>}.

\bibitem[BJPP18]{benoist2018entropy}
{\scshape Benoist, T., Jak{\v{s}}i{\'c}, V., Pautrat, Y. {\normalfont
  \smfandname} Pillet, C.-A.}: On entropy production of repeated quantum
  measurements {I}. {G}eneral theory. Commun. Math. Phys. \textbf{357}, 77--123
  (2018), \href{https://arxiv.org/abs/1607.00162}{<arXiv:1607.00162>}.

\bibitem[BP09]{OAM}
{\scshape Bruneau, L. {\normalfont \smfandname} Pillet, C.-A.}: Thermal
  relaxation of a {QED} cavity. J.~Stat.~Phys. \textbf{134}, 1071--1095 (2009),
  \href{https://doi.org/10.1007/s10955-008-9656-2}{[DOI:10.1007/s10955-008-9656-2]}.

\bibitem[BPP20]{benoist2018heat}
{\scshape Benoist, T., Panati, A. {\normalfont \smfandname} Pautrat, Y.}: Heat
  conservation and fluctuations for open quantum systems in the two-time
  measurement picture. J. Stat. Phys. \textbf{178}, 893--925 (2020),
  \href{https://arxiv.org/abs/1810.09999}{<arXiv:1810.09999>}.

\bibitem[BPR19]{benoist2019control}
{\scshape Benoist, T., Panati, A. {\normalfont \smfandname} Raqu{\'e}pas, R.}:
  Control of fluctuations and heavy tails for heat variation in the two-time
  measurement framework. Ann. H. Poincar{\'e} \textbf{20}, 631--674 (2019),
  \href{https://arxiv.org/abs/1802.02073}{<arXiv:1802.02073>}.

\bibitem[Bre92]{Br}
{\scshape Breiman, L.}: \emph{{P}robability}. SIAM, Philadelphia, 1992,
  \href{https://doi.org/10.1137/1.9781611971286}{[DOI:10.1137/1.9781611971286]}.

\bibitem[Bru14]{Bru14}
{\scshape Bruneau, L.}: Mixing properties of the one-atom maser. J. Stats.
  Phys. \textbf{155}, 888--908 (2014),
  \href{https://arxiv.org/abs/1312.6521}{<arXiv:1312.6521>}.

\bibitem[Bry93]{Bryc}
{\scshape Bryc, W.}: A remark on the connection between the large deviation
  principle and the central limit theorem. Stat. Prob. Lett. \textbf{18},
  253--256 (1993),
  \href{https://doi.org/10.1016/0167-7152(93)90012-8}{[DOI:10.1016/0167-7152(93)90012-8]}.

\bibitem[BS57]{BS57}
{\scshape Beck, A. {\normalfont \smfandname} Schwartz, J.~T.}: A vector-valued
  random ergodic theorem. Proc. A.M.S. \textbf{8}, 1049--1059 (1957),
  \href{https://doi.org/10.2307/2032681}{[DOI:10.2307/2032681]}.

\bibitem[Car10]{Carl10}
{\scshape Carlen, E.~A.}: Trace inequalities and quantum entropy: An
  introductory course. In \emph{Entropy and the {Q}uantum} (Sims, R.
  {\normalfont \smfandname} Ueltschi, D., \smfedsname), Contemporary
  Mathematics, vol. 529, AMS, Providence, RI, 2010, p.~33--83.

\bibitem[CHT11]{CHT}
{\scshape Campisi, M., Hänggi, P.,  {\normalfont \smfandname} Talkner, P.}:
  Quantum fluctuation relations: {F}oundations and applications. Rev. Mod.
  Phys. \textbf{83}, 771--791 (2011),
  \href{https://arxiv.org/abs/1012.2268}{<arXiv:1012.2268>}.

\bibitem[CJPS19]{CJPS}
{\scshape Cuneo, N., Jak{\v{s}}i{\'c}, V., Pillet, C.-A. {\normalfont
  \smfandname} Shirikyan, A.}: Large deviations and fluctuation theorem for
  selectively decoupled measures on shift spaces. Rev. Math. Phys. \textbf{31},
  1950036 (2019), \href{https://arxiv.org/abs/1712.09038}{<arXiv:1712.09038>}.

\bibitem[CM12]{CM12}
{\scshape Chetrite, R. {\normalfont \smfandname} Mallick, K.}: Quantum
  fluctuation relations for the {L}indblad master equation. J. Stat. Phys.
  \textbf{148}, 480--501 (2012),
  \href{https://arxiv.org/abs/1112.1303}{<arXiv:1112.1303>}.

\bibitem[CP15]{CP15}
{\scshape Carbone, R. {\normalfont \smfandname} Pautrat, Y.}: Homogeneous open
  quantum random walks on a lattice. J. Stat. Phys. \textbf{160}, 1125--1153
  (2015), \href{https://arxiv.org/abs/1408.1113}{<arXiv:1408.1113>}.

\bibitem[CP16]{CP16}
\bysame : Open quantum random walks: Reducibility, period, ergodic properties.
  Ann. Henri Poincaré \textbf{17}, 99--135 (2016),
  \href{https://arxiv.org/abs/1405.2214}{<arXiv:1405.2214>}.

\bibitem[Dav74]{davies1974markovian}
{\scshape Davies, E.~B.}: Markovian master equations. Commun. Math. Phys.
  \textbf{39}, 91--110 (1974),
  \href{https://doi.org/10.1007/BF01608389}{[DOI:10.1007/BF01608389]}.

\bibitem[Dav75]{davies1975markovian}
\bysame : Markovian master equations {III}. Ann. Inst. H. Poincaré Proba.
  Stat. \textbf{11}, 265--273 (1975),
  \href{http://www.numdam.org/item/AIHPB_1975__11_3_265_0/}{[Numdam]}.

\bibitem[Dav76a]{davies1976markovian}
\bysame : Markovian master equations {II}. Math. Ann. \textbf{219}, 147--158
  (1976), \href{https://doi.org/10.1007/BF01351898}{[DOI:10.1007/BF01351898]}.

\bibitem[Dav76b]{davies1976quantum}
\bysame : \emph{Quantum {T}heory of {O}pen {S}ystems}. Academic Press, London,
  1976.

\bibitem[DDRM08]{DRM08}
{\scshape Dereziński, J., De~Roeck, W. {\normalfont \smfandname} Maes, C.}:
  Fluctuations of quantum currents and unravelings of master equations. J.
  Stat. Phys. \textbf{131}, 341--356 (2008),
  \href{https://arxiv.org/abs/cond-mat/0703594}{<arXiv:cond-mat/0703594>}.

\bibitem[DHP20]{DHP20}
{\scshape Damak, M., Hammami, M. {\normalfont \smfandname} Pillet, C.-A.}: A
  detailed fluctuation theorem for heat fluxes in harmonic networks out of
  thermal equilibrium. J.~Stat.~Phys. \textbf{180}, 263--296 (2020),
  \href{https://arxiv.org/abs/1905.03536}{<arXiv:1905.03536>}.

\bibitem[DR09]{DR09}
{\scshape De~Roeck, W.}: Large deviation generating function for currents in
  the {P}auli--{F}ierz model. Rev. Math. Phys. \textbf{21}, 549--585 (2009),
  \href{https://arxiv.org/abs/0704.3400}{<arXiv:0704.3400>}.

\bibitem[DRM06]{RM06}
{\scshape De~Roeck, W. {\normalfont \smfandname} Maes, C.}: Steady state
  fluctuations of the dissipated heat for a quantum stochastic model. Rev.
  Math. Phys. \textbf{18}, 619--653 (2006),
  \href{https://arxiv.org/abs/cond-mat/0406004}{<arXiv:cond-mat/0406004>}.

\bibitem[DZ10]{DZ98}
{\scshape Dembo, A. {\normalfont \smfandname} Zeitouni, O.}: \emph{Large
  {D}eviations. {T}echniques and {A}pplications}. Springer, Berlin, 2010,
  \href{https://doi.org/10.1007/978-3-642-03311-7}{[DOI:10.1007/978-3-642-03311-7]}.

\bibitem[EHK78]{EHK}
{\scshape Evans, D. {\normalfont \smfandname} H\o{}egh-Krohn, R.}: Spectral
  properties of positive maps on {C}*-algebras. J. London Math. Soc. 345--355
  (1978),
  \href{https://doi.org/10.1112/jlms/s2-17.2.345}{[DOI:10.1112/jlms/s2-17.2.345]}.

\bibitem[EHM09]{EHM}
{\scshape Esposito, M., Harbola, U. {\normalfont \smfandname} Mukamel, S.}:
  Nonequilibrium fluctuations, fluctuation theorems, and counting statistics in
  quantum systems. Rev. Mod. Phys. \textbf{81}, 1665--1702 (2009),
  \href{https://arxiv.org/abs/0811.3717}{<arXiv:0811.3717>}.

\bibitem[FJM86]{FJM}
{\scshape Filipowicz, P., Javanainen, J. {\normalfont \smfandname} Meystre,
  P.}: Theory of a microscopic maser. Phys. Rev. A \textbf{34}, 3077--3087
  (1986),
  \href{https://doi.org/10.1103/PhysRevA.34.3077}{[DOI:10.1103/PhysRevA.34.3077]}.

\bibitem[FP09]{FP09}
{\scshape Fagnola, F. {\normalfont \smfandname} Pellicer, R.}: Irreducible and
  periodic positive maps. Commun. Stoch. Anal. \textbf{3}, 407--418 (2009),
  \href{https://doi.org/10.31390/cosa.3.3.06}{[DOI:10.31390/cosa.3.3.06]}.

\bibitem[Gal96]{Gal96}
{\scshape Gallavotti, G.}: Extension of {O}nsager’s reciprocity to large
  fields and the chaotic hypothesis. Phys. Rev. Lett. \textbf{77}, 4334--4337
  (1996),
  \href{https://arxiv.org/abs/chao-dyn/9603003}{<arXiv:chao-dyn/9603003>}.

\bibitem[GC95a]{GC1}
{\scshape Gallavotti, G. {\normalfont \smfandname} Cohen, E. G.~D.}: Dynamical
  ensembles in nonequilibrium statistical mechanics. Phys. Rev. Lett.
  \textbf{74}, 2694--7 (1995),
  \href{https://arxiv.org/abs/chao-dyn/9410007}{<arXiv:chao-dyn/9410007>}.

\bibitem[GC95b]{GC2}
\bysame : Dynamical ensembles in stationary states. J. Stat. Phys. \textbf{80},
  931--970 (1995),
  \href{https://arxiv.org/abs/chao-dyn/9501015}{<arXiv:chao-dyn/9501015>}.

\bibitem[Gro81]{groh}
{\scshape Groh, U.}: The peripheral point spectrum of {S}chwarz operators on
  ${C}^\ast$-algebras. Math. Z. \textbf{176}, 311--318 (1981),
  \href{https://doi.org/10.1007/BF01214608}{[DOI:10.1007/BF01214608]}.

\bibitem[HJ17]{fqw}
{\scshape Hamza, E. {\normalfont \smfandname} Joye, A.}: Thermalization of
  fermionic quantum walkers. J.~Stat.~Phys. \textbf{166}, 1365--1392 (2017),
  \href{https://arxiv.org/abs/1611.07477}{<arXiv:1611.07477>}.

\bibitem[HJPR17]{HJPR}
{\scshape Hanson, E.~P., Joye, A., Pautrat, Y. {\normalfont \smfandname}
  Raqu{\'e}pas, R.}: Landauer's principle in repeated interaction systems.
  Commun. Math. Phys. \textbf{349}, 285--327 (2017),
  \href{https://arxiv.org/abs/1510.00533}{<arXiv:1510.00533>}.

\bibitem[HJPR18]{HJPR2}
\bysame : Landauer's principle for trajectories of repeated interaction
  systems. Ann. H. Poincar{\'e} \textbf{19}, 1939--1991 (2018),
  \href{https://arxiv.org/abs/1705.08281}{<arXiv:1705.08281>}.

\bibitem[JLP13]{jakvsic2013entropic}
{\scshape Jak{\v{s}}i{\'c}, V., Landon, B. {\normalfont \smfandname} Pillet,
  C.-A.}: Entropic fluctuations in {XY} chains and reflectionless {J}acobi
  matrices. Ann. H. Poincar{\'e} \textbf{14}, 1775--1800 (2013),
  \href{https://arxiv.org/abs/1209.3675}{<arXiv:1209.3675>}.

\bibitem[JOPP12]{JOPP}
{\scshape Jak\v{s}i\'{c}, V., Ogata, Y., Pillet, C.-A. {\normalfont
  \smfandname} Pautrat, Y.}: Entropic fluctuations in quantum statistical
  mechanics. an introduction. In \emph{Quantum {T}heory from {S}mall to {L}arge
  {S}cales} (Fr\"ohlich, J., Salmhofer, M., Mastropietro, V., de~Roeck, W.
  {\normalfont \smfandname} Cugliandolo, L., \smfedsname), Lecture Notes of the
  Les Houches Summer School, vol.~95, Oxford University Press, Oxford, 2012,
  p.~213--410, \href{https://arxiv.org/abs/1106.3786}{<arXiv:1106.3786>}.

\bibitem[JP14]{JPlan}
{\scshape Jak{\v{s}}i{\'c}, V. {\normalfont \smfandname} Pillet, C.-A.}: A note
  on the {L}andauer principle in quantum statistical mechanics. J. Math. Phys.
  \textbf{55}, 075210 (2014),
  \href{https://arxiv.org/abs/1406.0034}{<arXiv:1406.0034>}.

\bibitem[JPPP15]{jakvsic2015energy}
{\scshape Jak{\v{s}}i{\'c}, V., Panangaden, J., Panati, A. {\normalfont
  \smfandname} Pillet, C.-A.}: Energy conservation, counting statistics, and
  return to equilibrium. Lett. Math. Phys. \textbf{105}, 917--938 (2015),
  \href{https://arxiv.org/abs/1409.8610}{<arXiv:1409.8610>}.

\bibitem[JPRB11]{jakvsic2011entropic}
{\scshape Jak{\v{s}}i{\'c}, V., Pillet, C.-A. {\normalfont \smfandname}
  Rey-Bellet, L.}: Entropic fluctuations in statistical mechanics {I}.
  {C}lassical dynamical systems. Nonlinearity \textbf{24}, 699 (2011),
  \href{https://arxiv.org/abs/1009.3248}{<arXiv:1009.3248>}.

\bibitem[JPS17]{jakvsic2017entropic}
{\scshape Jak{\v{s}}i{\'c}, V., Pillet, C.-A. {\normalfont \smfandname}
  Shirikyan, A.}: Entropic fluctuations in thermally driven harmonic networks.
  J. Stat. Phys. \textbf{166}, 926--1015 (2017),
  \href{https://arxiv.org/abs/1606.01498}{<arXiv:1606.01498>}.

\bibitem[JPW14]{JPW}
{\scshape Jak{\v{s}}i{\'c}, V., Pillet, C.-A. {\normalfont \smfandname}
  Westrich, M.}: Entropic fluctuations of quantum dynamical semigroups. J.
  Stat. Phys. \textbf{154}, 153--187 (2014),
  \href{https://arxiv.org/abs/1305.4409}{<arXiv:1305.4409>}.

\bibitem[Kat80]{Kato}
{\scshape Kato, T.}: \emph{Perturbation {T}heory for {L}inear {O}perators}.
  Springer, Berlin, 1980,
  \href{https://doi.org/10.1007/978-3-642-66282-9}{[DOI:10.1007/978-3-642-66282-9]}.

\bibitem[KM04]{KM04}
{\scshape Kümmerer, B. {\normalfont \smfandname} Maassen, H.}: A pathwise
  ergodic theorem for quantum trajectories. J. Phys. A: Math. Gen. \textbf{37},
  11889--11896 (2004), \href{https://arxiv.org/abs/0406213}{<arXiv:0406213>}.

\bibitem[KS09]{Schenkeretal}
{\scshape Kang, Y. {\normalfont \smfandname} Schenker, J.}: Diffusion of wave
  packets in a {M}arkov random potential. J. Stat. Phys. \textbf{134},
  1005--1022 (2009), \href{https://arxiv.org/abs/0808.2784}{<arXiv:0808.2784>}.

\bibitem[Kü06]{Kumm06}
{\scshape Kümmerer, B.}: Quantum {M}arkov processes and applications in
  physics. In \emph{Quantum {I}ndependent {I}ncrement {P}rocesses {II}}
  (Schüermann, M. {\normalfont \smfandname} Franz, U., \smfedsname), Lect.
  Notes Math., vol. 1866, Springer, Berlin, 2006, p.~259--330,
  \href{https://doi.org/10.1007/11376637}{[DOI:10.1007/11376637]}.

\bibitem[LS78]{LS}
{\scshape Lebowitz, J. {\normalfont \smfandname} Spohn, H.}: Irreversible
  thermodynamics for quantum systems weakly coupled to thermal reservoirs. Adv.
  Chem. Phys. \textbf{38}, 109--142 (1978),
  \href{https://doi.org/10.1002/9780470142578.ch2}{[DOI:10.1002/9780470142578.ch2]}.

\bibitem[MS19]{movassagh2019ergodic}
{\scshape Movassagh, R. {\normalfont \smfandname} Schenker, J.}: An ergodic
  theorem for homogeneously distributed quantum channels with applications to
  matrix product states. Preprint, 2019,
  \href{https://arxiv.org/abs/1909.11769}{<arXiv:1909.11769>}.

\bibitem[MS20]{movassagh2020theory}
\bysame : Theory of ergodic quantum processes. Preprint, 2020,
  \href{https://arxiv.org/abs/2004.14397}{<arXiv:2004.14397>}.

\bibitem[MWM85]{MWM}
{\scshape Meschede, D., Walther, H. {\normalfont \smfandname} M{\"u}ller, G.}:
  One-atom maser. Phys. Rev. Lett. \textbf{54}, 551--554 (1985),
  \href{https://doi.org/10.1103/PhysRevLett.54.551}{[DOI:10.1103/PhysRevLett.54.551]}.

\bibitem[NP12]{nechita2012random}
{\scshape Nechita, I. {\normalfont \smfandname} Pellegrini, C.}: Random
  repeated quantum interactions and random invariant states. Prob. Theor. Rel.
  Fields \textbf{152}, 299--320 (2012),
  \href{https://arxiv.org/abs/0902.2634}{<arXiv:0902.2634>}.

\bibitem[Ons31a]{Ons1}
{\scshape Onsager, L.}: Reciprocal relations in irreversible processes {I}.
  Phys. Rev. \textbf{37}, 405 (1931),
  \href{https://doi.org/10.1103/PhysRev.37.405}{[DOI:10.1103/PhysRev.37.405]}.

\bibitem[Ons31b]{Ons2}
\bysame : Reciprocal relations in irreversible processes {II}. Phys. Rev.
  \textbf{38}, 2265 (1931),
  \href{https://doi.org/10.1103/PhysRev.38.2265}{[DOI:10.1103/PhysRev.38.2265]}.

\bibitem[Orn72]{Ornstein72}
{\scshape Ornstein, D.~S.}: On the root problem in ergodic theory. In
  \emph{Proceedings of the {S}ixth {B}erkeley {S}ymposium on {M}athematical
  {S}tatistics and {P}robability. Volume {II}: {P}robability {T}heory} (Le~Cam,
  L.~M., Neyman, J. {\normalfont \smfandname} Scott, E.~L., \smfedsname),
  University of California Press, Berkeley, California, 1972, p.~347--356.

\bibitem[Pet08]{Pe08}
{\scshape Petz, D.}: \emph{{Q}uantum {I}nformation {T}heory and {Q}uantum
  {S}tatistics}. Springer, Berlin, 2008,
  \href{https://doi.org/10.1007/978-3-540-74636-2}{[DOI:10.1007/978-3-540-74636-2]}.

\bibitem[Pil85]{Pillet1}
{\scshape Pillet, C.-A.}: Some results on the quantum dynamics of a particle in
  a {M}arkovian potential. Commun. Math. Phys. \textbf{102}, 237--254 (1985),
  \href{https://doi.org/10.1007/BF01229379}{[DOI:10.1007/BF01229379]}.

\bibitem[Pil86a]{Pillet2}
\bysame : Asymptotic completeness for a quantum particle in a {M}arkovian short
  range potential. Commun. Math. Phys. \textbf{105}, 259--280 (1986),
  \href{https://doi.org/10.1007/BF01211102}{[DOI:10.1007/BF01211102]}.

\bibitem[Pil86b]{Pillet}
\bysame : M\'ecanique quantique dans un potentiel al\'eatoire markovien.
  \smfphdthesisname, ETH Z\"urich, 1986,
  \href{https://doi.org/10.3929/ethz-a-000397036}{[DOI:10.3929/ethz-a-000397036]}.

\bibitem[PS98]{PSI}
{\scshape P\`olya, G. {\normalfont \smfandname} Szeg\"o, G.}: \emph{Problems
  and {T}heorems in {A}nalysis {I}}. Springer, Berlin, 1998,
  \href{https://doi.org/10.1007/978-3-642-61983-0}{[DOI:10.1007/978-3-642-61983-0]}.

\bibitem[Raq20]{raquepas2020fermionic}
{\scshape Raqu{\'e}pas, R.}: On fermionic walkers interacting with a correlated
  structured environment. Lett. Math. Phys. \textbf{110}, 121--145 (2020),
  \href{https://arxiv.org/abs/1902.03703}{<arXiv:1902.03703>}.

\bibitem[RBH01]{RBH}
{\scshape Raimond, J.-M., Brune, M. {\normalfont \smfandname} Haroche, S.}:
  Manipulating quantum entanglement with atoms and photons in a cavity. Rev.
  Mod. Phys. \textbf{73}, 565--582 (2001),
  \href{https://doi.org/10.1103/RevModPhys.73.565}{[DOI:10.1103/RevModPhys.73.565]}.

\bibitem[RMM07]{RM07}
{\scshape Rondoni, L. {\normalfont \smfandname} Mejía-Monasterio, C.}:
  Fluctuations in non-equilibrium statistical mechanics: models, mathematical
  theory, physical mechanisms. Nonlinearity \textbf{20}, 1--37 (2007),
  \href{https://doi.org/10.1088/0951-7715/20/10/R01}{[DOI:10.1088/0951-7715/20/10/R01]}.

\bibitem[RW14]{RW14}
{\scshape Reeb, D. {\normalfont \smfandname} Wolf, M.~M.}: An improved
  {L}andauer principle with finite-size corrections. New J. Phys. \textbf{16},
  103011 (2014), \href{https://arxiv.org/abs/1306.4352}{<arXiv:1306.4352>}.

\bibitem[Sch06]{Schrader}
{\scshape Schrader, R.}: Perron-{F}robenius theory for positive maps on trace
  ideals. Fields Inst. Commun. \textbf{30}, 107--182 (2006),
  \href{https://arxiv.org/abs/math-ph/0007020}{<arXiv:math-ph/0007020>}.

\bibitem[Tak79]{Ta1}
{\scshape Takesaki, M.}: \emph{Theory of {O}perator {A}lgebras {I}}. Springer,
  New {Y}ork, 1979,
  \href{https://doi.org/10.1007/978-1-4612-6188-9}{[DOI:10.1007/978-1-4612-6188-9]}.

\bibitem[TLH07]{TLH07}
{\scshape Talkner, P., Lutz, E. {\normalfont \smfandname} Hänggi, P.}:
  Fluctuation theorems: {W}ork is not an observable. Phys. Rev. E \textbf{75},
  050102 (2007),
  \href{https://arxiv.org/abs/cond-mat/0703189}{<arXiv:cond-mat/0703189>}.

\bibitem[Wal81]{Wa}
{\scshape Walters, P.}: \emph{{A}n {I}ntroduction to {E}rgodic {T}heory}.
  Springer, Berlin, 1981.

\bibitem[WBKM00]{WBKM}
{\scshape Wellens, T., Buchleitner, A., K{\"u}mmerer, B. {\normalfont
  \smfandname} Maassen, H.}: Quantum state preparation via asymptotic
  completeness. Phys. Rev. Lett. \textbf{85}, 3361--3364 (2000),
  \href{https://doi.org/10.1103/PhysRevLett.85.3361}{[DOI:10.1103/PhysRevLett.85.3361]}.

\end{thebibliography}
\end{document}